\documentclass[preprint,12pt]{elsarticle}

\usepackage{amssymb}
\usepackage{graphicx}
\usepackage[usenames,dvipsnames]{xcolor}
\usepackage{hyperref}
\usepackage{amsthm}
\usepackage{mathrsfs,amsmath,amscd,mathtools}
\usepackage{booktabs}
\usepackage{lscape}
\usepackage{color}
\usepackage{multirow}
\usepackage{booktabs}
\usepackage{subfig}
\usepackage{float}
\usepackage{algorithm2e}
\usepackage{caption}
\usepackage{filecontents}
\usepackage[normalem]{ulem}
\usepackage[flushleft]{threeparttable}
\usepackage{stfloats}
\usepackage{lscape}
\usepackage{longtable}
\usepackage{listings}
\usepackage{courier}
\usepackage{soul}
\usepackage{forest}
\usepackage[graphicx]{realboxes}
\usepackage{makecell}
\usepackage{array}
\usepackage{adjustbox}
\usepackage{xspace}

\graphicspath{{./}}

\DeclareMathOperator*{\EE}{\mathbb{E}}

\newcommand{\GGPS}{$GI/GI/1/PS$\xspace}

\newcommand{\TVGGPS}{$GI_t/GI_t/1/PS$\xspace}

\newtheorem{proposition}{Proposition}

\newtheorem{remark}{Remark}

\journal{}

\begin{document}
	
	\begin{frontmatter}
		
		
		
%		\title{Using simulation to study effect of response time stabilizing controls in single-server processor sharing queues with slowly time-varying arrival rate}
		\title{Stabilizing the virtual response time in single-server processor sharing queues with slowly time-varying arrival rates}
		
		\author{Yongkyu Cho}
		\author{Young Myoung Ko\corref{cor}}
		\cortext[cor]{Corresponding author.}
		\ead{youngko@postech.ac.kr}
		\address{Department of Industrial and Management Engineering\\Pohang University of Science and Technology\\ 77, Cheongam-ro, Nam-gu, Pohang, Gyeongbuk, Korea}
		% \address{Address\fnref{label3}}
		
		\begin{abstract}
			Motivated by the work of \citet{journal:W2015}, who studied stabilization of the mean virtual waiting time (excluding service time) in a $GI_t/GI_t/1/FCFS$ queue, this paper investigates the stabilization of the mean virtual response time in a single-server processor sharing (PS) queueing system with a time-varying arrival rate and a service rate control (a \TVGGPS\ queue). We propose and compare a modified square-root (SR) control and a difference-matching (DM) control to stabilize the mean virtual response time of a \TVGGPS\ queue. Extensive simulation studies with various settings of arrival processes and service times show that the DM control outperforms the SR control for heavy-traffic conditions, and that the SR control performs better for light-traffic conditions.
		\end{abstract}
		
		
	\end{frontmatter}
	
	
	\section{Introduction} \label{sec:intro}
	Modern data centers consume tremendous amounts of energy to supply networking, computing, and storage services to global IT companies. Concerns about energy consumption have prompted researchers to explore operational methods that maximize energy efficiency and satisfy a certain level of quality of service (QoS), \citep{journal:AV2011,journal:KC2014,proc:LLSAGT2015}. QoS can be achieved by adding constraints that impose upper bounds for response time-related metrics, e.g., the mean virtual response time and the tail probability of the response time. In general, these constraints are binding, because of the conflict between the QoS-related metrics and energy consumption. Binding the QoS-related constraints implies that the metrics are maintained as a constant value, and suggests the need to investigate the stabilization of response times. Although some proposed methodologies \citep{journal:AV2011,journal:KC2014,proc:LLSAGT2015} assume the stationarity of data traffic arrival processes, nonstationary properties, such as time-varying arrival rates from real data \citep{web:CAIDA}, make it difficult to analyze queueing system performance.
	
	In this paper, therefore, we study the service rate controls that stabilize the mean virtual response time to a certain target value in a single server PS queue representing a computer server in a data center under time-varying arrival rates and controllable service rates, i.e., a \TVGGPS\ queue. Our approach is similar to \citet{journal:W2015} and \citet{journal:MW2016}, who considered three different service rate controls, two of which were designed to stabilize the mean (virtual) waiting time in a $GI_t/GI_t/1/FCFS$ queue. The slowly time-varying traffic patterns of internet services \citep{web:CAIDA} justify our use of pointwise stationary approximation (PSA) \cite{journal:GK1991}. We adopt different heavy-traffic approximation results (HTA), because our objective is to stabilize the mean virtual response time, which is one of our target performance measures. We propose two service rate control schemes:
%	\begin{subequations}
		\begin{align}
		\mu_{SR}(t;s)&\equiv\frac{(s\lambda(t)+1)\beta+\sqrt{(s\lambda(t)+1)^2\beta^2+4s\lambda(t)\beta^2(V_{FCFS}-1)}}{2s}, \label{sr}\\
		\mu_{DM}(t;s)&\equiv\beta\left(\lambda(t)+\frac{V_{PS}}{s}\right), \label{dm}
		\end{align}
%	\end{subequations}
	where $s$ is the desired response time, $\beta$ is the mean job size, $\lambda(t)$ is the arrival rate function, $V_{FCFS}\equiv(C_a^2+C_s^2)/2$, and $V_{PS}\equiv(C_a^2+C_s^2)/(1+C_s^2)$, with $C_a^2$ and $C_s^2$ are the squared coefficient of variations (SCV) of the base interarrival and service time distributions. Equation (\ref{sr}) is a modification of the well-known square-root control (SR) suggested in \citet{journal:W2015}, and Equation (\ref{dm}) is a new control scheme, which we call the difference-matching (DM) control, because it maintains the difference between $\mu(t)$ and $\beta\lambda(t)$ as a constant $\beta V_{PS}/s$. The DM control is easy to implement thanks to its simplicity.

	Figure~\ref{fig:erln_performance_s_0.1} shows the mean queue length process, $\EE[Q(t)]$ (green line), and the mean virtual response time process, $\EE[R(t)]$ (red line), of the simulated \TVGGPS\ queues with an Erlang base arrival distribution and a lognormal job size distribution with the SR control as in Equation~(\ref{sr}) and three different time-varying arrival rates. The dotted black lines are 95\% confidence intervals and the dotted blue line plots the arrival rate function; its dedicated y-axis is on the right. The plots show that the response time is almost perfectly stabilized by the SR control under the light-traffic condition.
	
	Figure \ref{fig:erer_performance_s_10} depicts the performance measures when the target response time is relatively long. While the stabilization looks poor for both controls, their \emph{relative amplitude} -- one of our performance measures described in Section~\ref{subsec:metric} -- is under 10\%. Figure~\ref{fig:dm_er_er_0001} depicts that the DM control achieves the target response time, which implies that using the DM control shows better \emph{accuracy} -- the other performance measure in Section~\ref{subsec:metric} -- under the heavy traffic condition (long response time). 
	\begin{figure}
		\centering
		\subfloat[][$\gamma=0.1$]
		{
			\centering		\resizebox{0.28\textwidth}{!}{\includegraphics{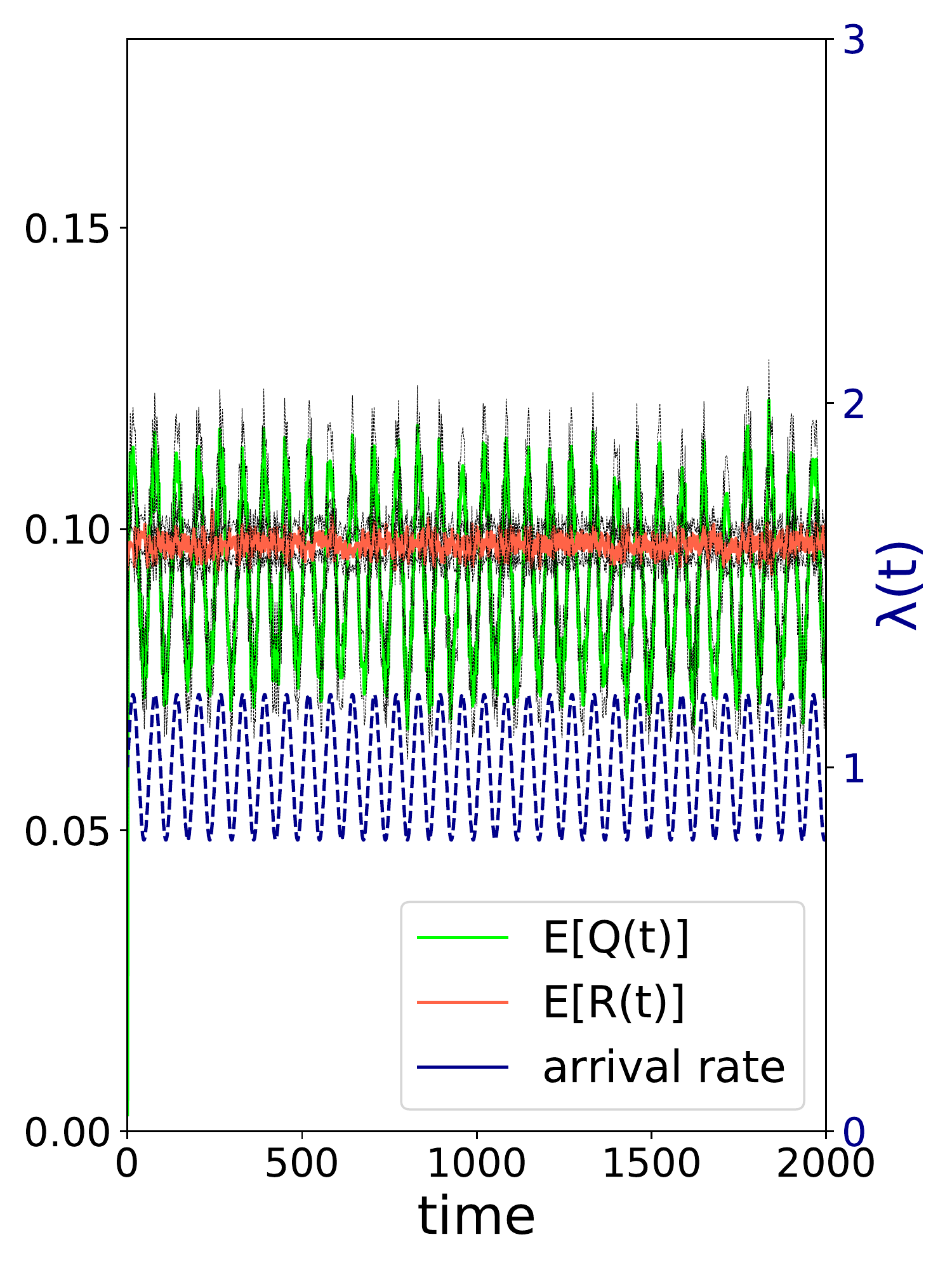}}
			\label{fig:s_01_sr_erln_01}
		}
		~
		\subfloat[][$\gamma=0.01$]
		{
			\centering\resizebox{0.28\textwidth}{!}{\includegraphics{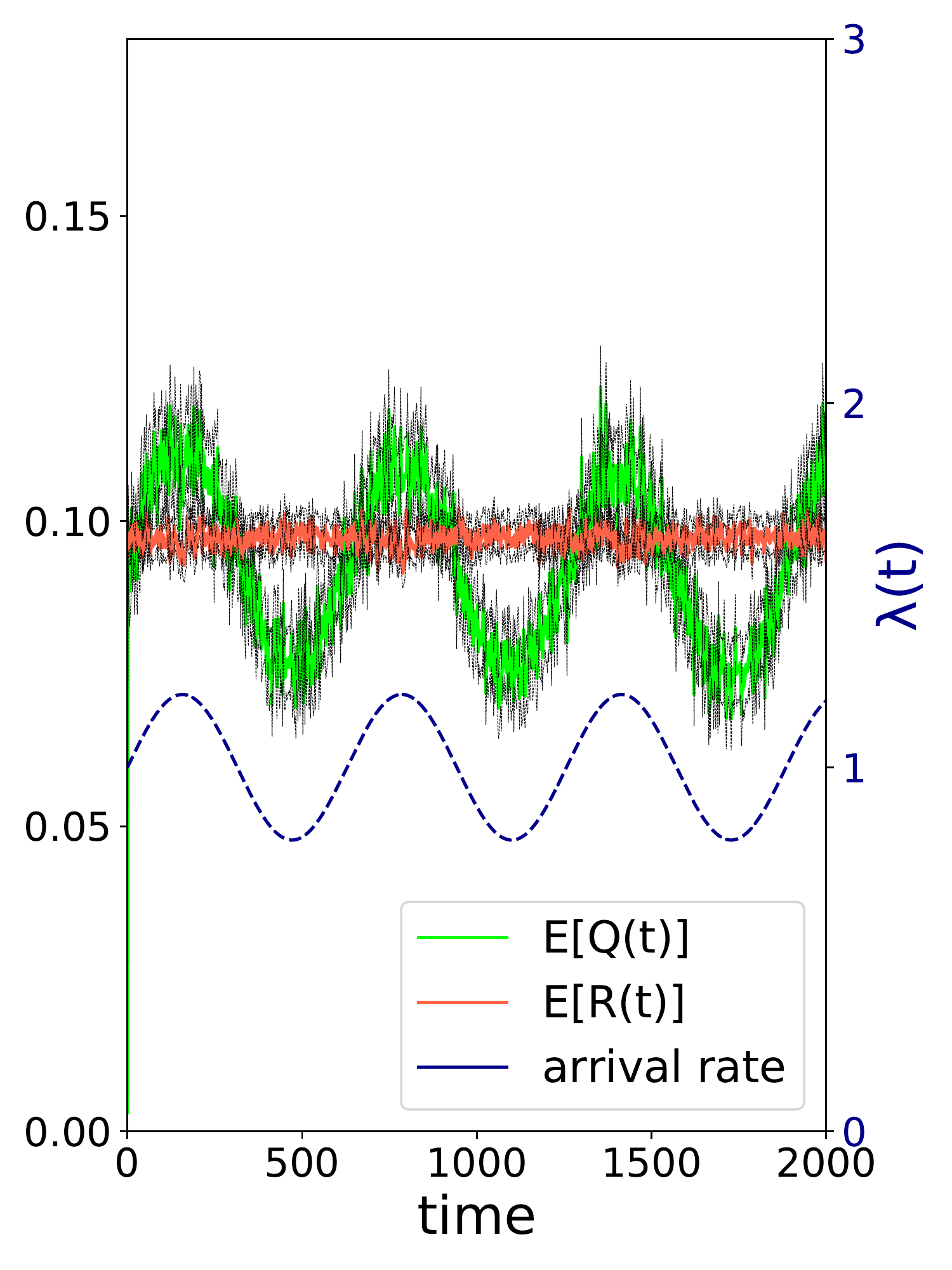}}
			\label{s_01_sr_erln_001}
		}
		~
		\subfloat[][$\gamma=0.001$]
		{
			\centering\resizebox{0.28\textwidth}{!}{\includegraphics{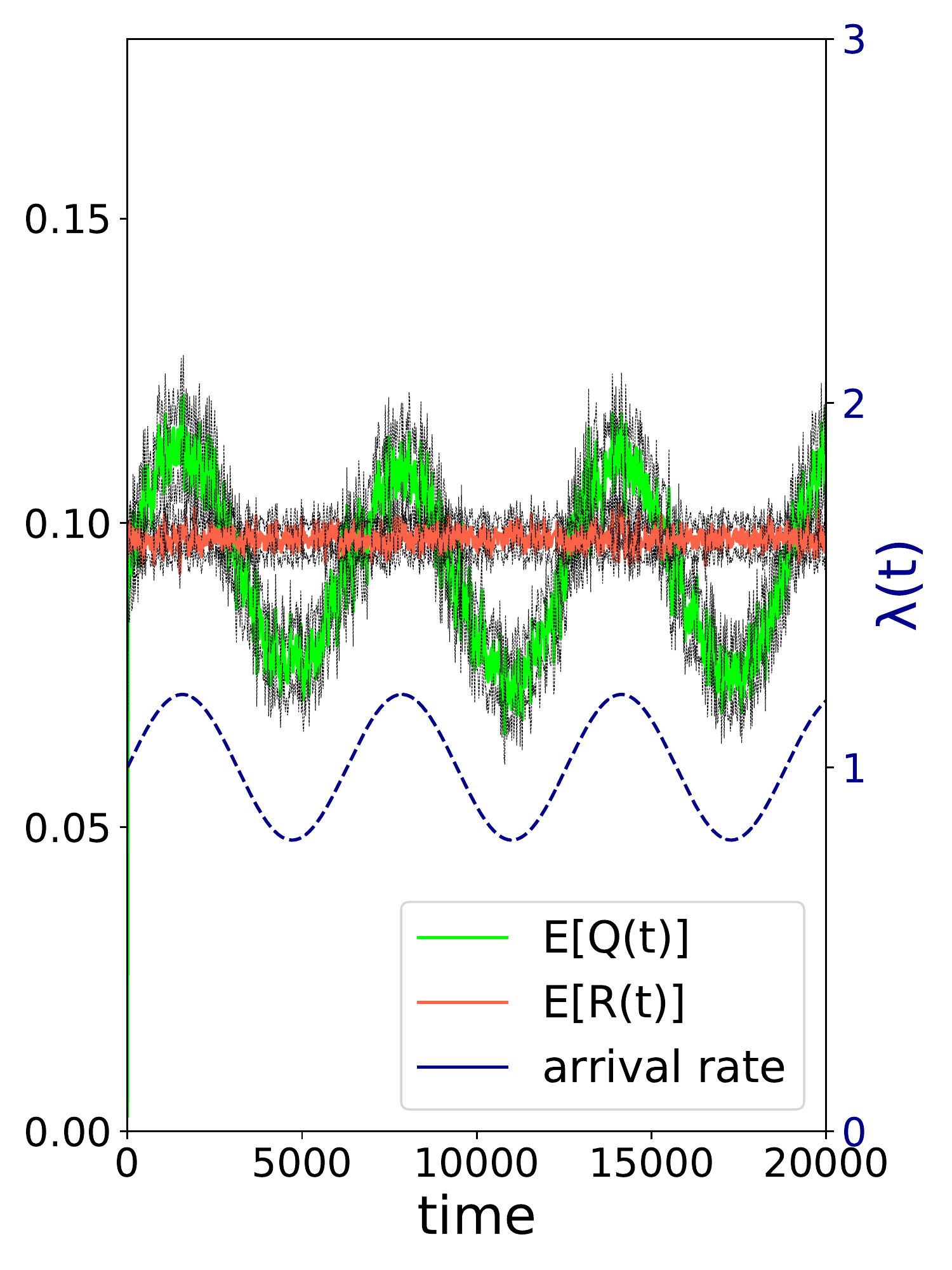}}
			\label{s_01_sr_erln_0001}
		}
		
		\caption{General performance measures of $ER_t/LN_t/1/PS$ queues under the SR control where $\lambda(t)=1+0.2\sin{(\gamma t)}$ with target response time $0.1$ (light-traffic)}
		%	\begin{minipage}
		%		{0.65\textwidth}{\footnotesize *\suc\ instance is too huge and extremely sparse to plot more than 1 scenario}
		%	\end{minipage}
		\label{fig:erln_performance_s_0.1}
	\end{figure}
	
	\begin{figure}
		\centering
		\subfloat[][SR control]
		{
	%		\centering\resizebox{0.45\textwidth}{!}{\includegraphics{TVGG1PS_10_SR_ER_LN_0001_normal}}
			\centering\resizebox{0.45\textwidth}{!}{\includegraphics{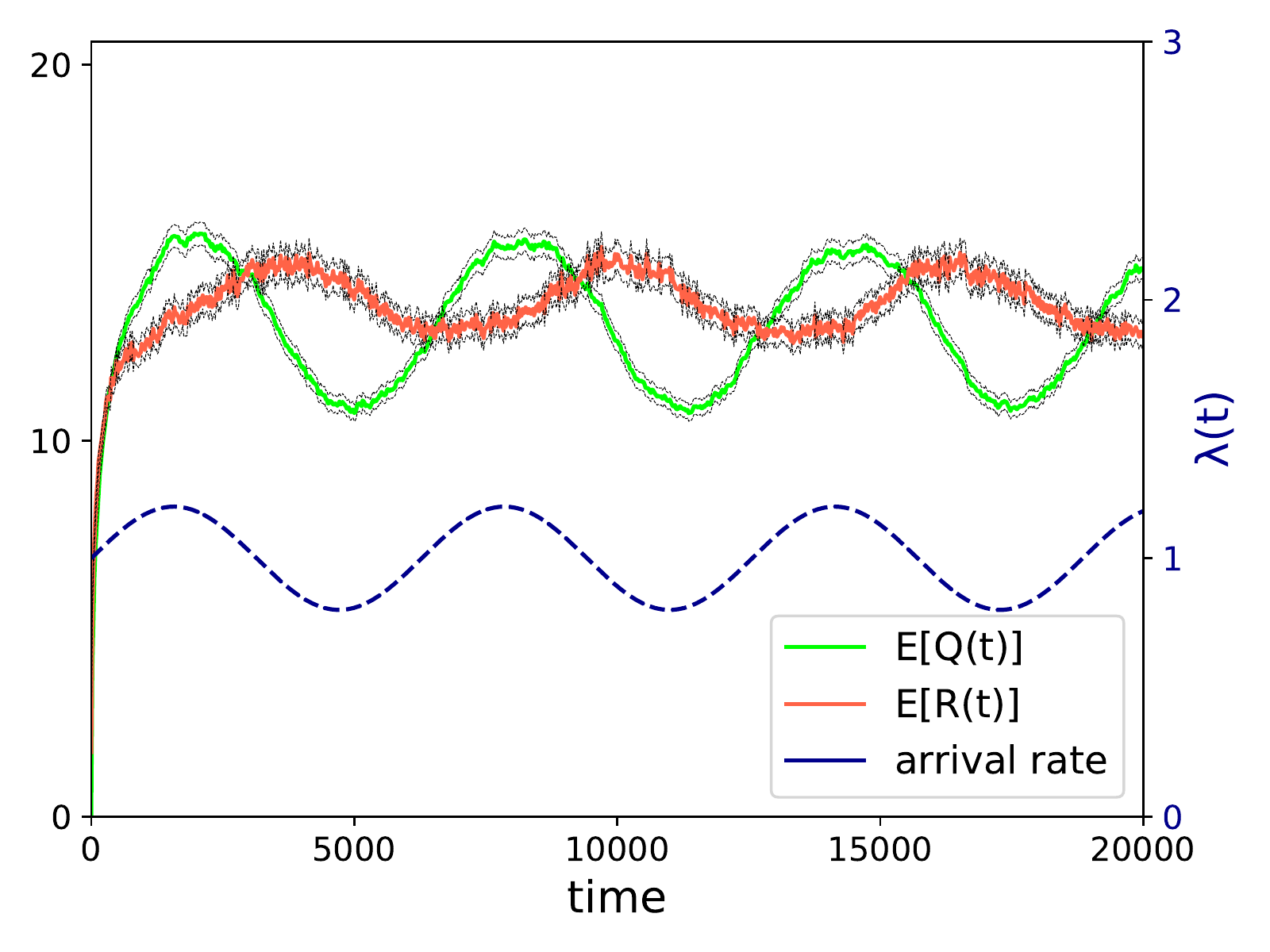}}

			\label{fig:sr_er_er_0001}
		}
		~
		\subfloat[][DM control]
		{
%			\centering\resizebox{0.45\textwidth}{!}{\includegraphics{TVGG1PS_10_PD_ER_LN_0001_normal}}
			\centering\resizebox{0.45\textwidth}{!}{\includegraphics{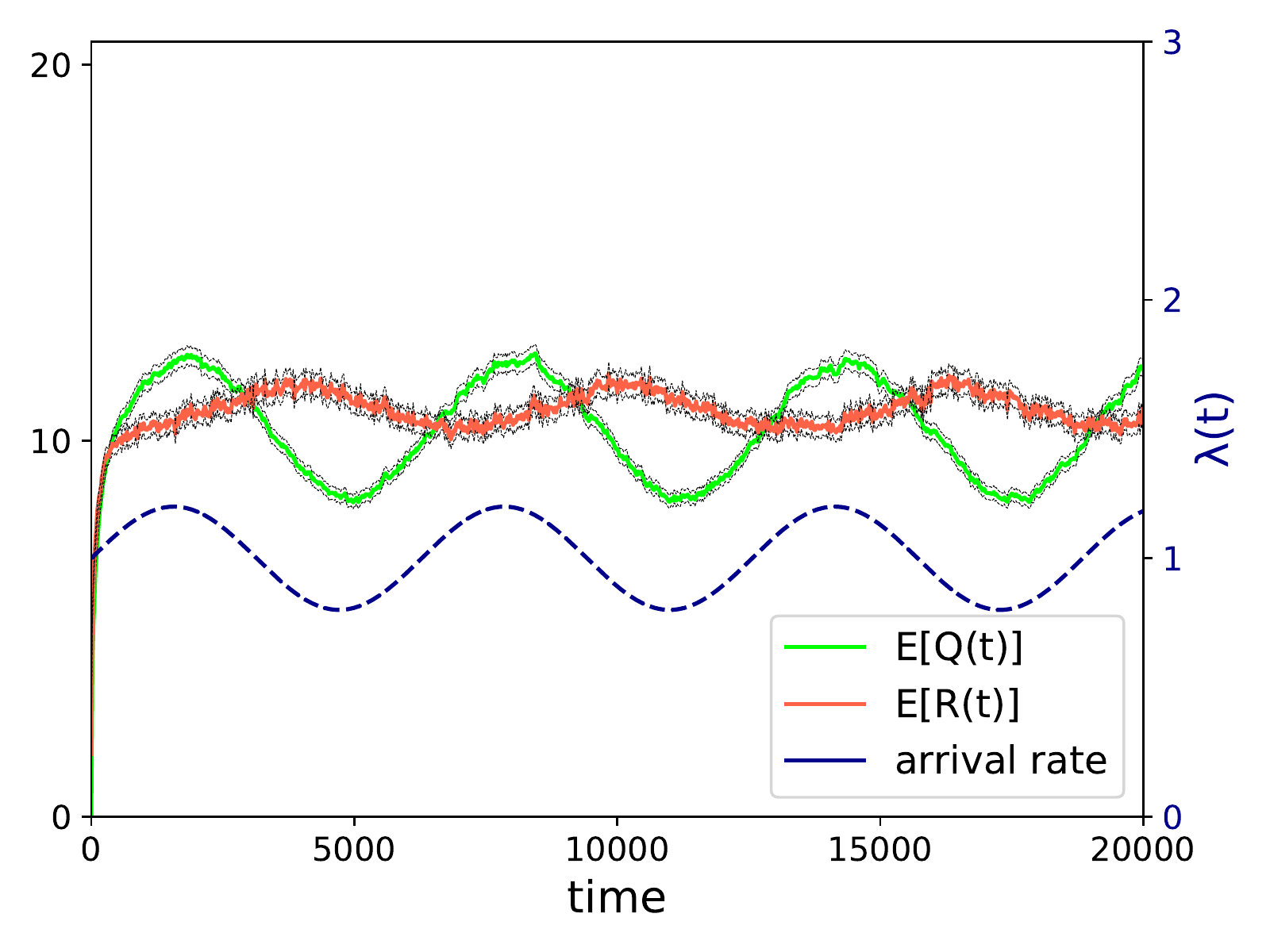}}

			\label{fig:dm_er_er_0001}
		}
		
		\caption{General performance measures of $ER_t/ER_t/1/PS$ queues with target response time $10$ (heavy-traffic)}
		\label{fig:erer_performance_s_10}
	\end{figure}
	
	This paper contributes to the published literature on queueing systems by studying the response time stabilizing controls for a \TVGGPS\ queue; proposing a new control scheme, i.e., DM control, for heavy-traffic conditions; undertaking extensive simulations of the proposed control schemes; and gaining insights into their effectiveness for data centers.
	
	The remainder of this paper is organized as follows. Section~\ref{chap2:model} introduces a single-server PS queueing model with a time-varying arrival rate and a controllable service rate. We explain some details for simulating a \TVGGPS\ queue, which is not straightforward, unlike its stationary counterpart. Section~\ref{chap2:method} explains the procedure to derive the two service rate controls, and some simple characteristics of the controls. Section~\ref{chap2:simulation} reports the results of the simulations including the interesting phenomena we find. Section~\ref{chap2:conclusion} concludes and suggests some future research directions.
	
%	The remainder of this research study is organized as follows. In Section \ref{chap2:model}, we indroduce the queueing model of interest: a single-server PS queue with time-varying arrival rate where service rate is deterministic and under control. Since the Nonstationary non-Poisson process (NSNP) is not trivial to catch the concept unlike relatively common nonstationary Poisson process (NSPP), we explain some details including the simulation algorithm we use. Section \ref{chap2:method} explains the procedure to derive the two service rate controls. Some simple characteristics of the controls are also accompanied. In Section \ref{chap2:simulation}, we report the simulation experimental results. Post-analysis on some interesting phenomena is also provided later in the section. We summarize this study in Section \ref{chap2:conclusion}.

\section{The model}\label{chap2:model}
%Section~\ref{chap2:queue}, we describe the queueing model of interest: a single-server queue with nonstationary non-Poisson arrivals under the PS discipline and service rate control. We also explain the procedures to simulate such queueing systems in Section~\ref{subsec:simulating_nsnp}. Before diving into the details, we define some notations below:
Section~\ref{chap2:queue} introduces a single-server queueing model with nonstationary non-Poisson arrivals under the PS discipline and the service rate control. Section~\ref{subsec:simulating_nsnp} explains the procedures to simulate such queueing systems. Throughout this paper, we use the following notations:
\begin{itemize}
	\item $f(t)$: arbitrary periodic function with a period $\mathcal{T}_f$
	\item $\bar{f}$: spatial scale average of $f$; $f\equiv\int_t^{t+\mathcal{T}_f}f(x)\mathrm{d}x/\mathcal{T}_f$ for any $t\in [0,\infty)$
	\item $\lambda(t)$: arrival rate function
	\item $\mu(t)$: service rate function
	\item $T_i$: base inter-arrival times between $i^{\textrm{th}}$ and $i-1^{\textrm{st}}$ job; i.i.d. random variables having a general distribution function $F(\cdot)$ with a mean  $\tau\equiv\EE[T_i]<\infty$ and an SCV $C_a^2\equiv SCV(T_i)<\infty$
	\item $S_i$: service requirement that the $i^{\textrm{th}}$ job brings; i.i.d. random variables having a general distribution function $G(\cdot)$ with a mean $\beta\equiv\EE[S_i]<\infty$ and an SCV $C_s^2\equiv SCV(S_i)<\infty$
	%\item $V_i$: the duration that the $i^{\textrm{th}}$ job spends in service
	%\item $W_i$: the duration that the $i^{\textrm{th}}$ job waits before getting serviced 
	\item $\rho(t)$: instantaneous traffic intensity; $\rho(t)\equiv\lambda(t)\beta/\mu(t)$
	\item $A_i$: time when the $i^{\textrm{th}}$ job arrives
	%\item $B_i$: the time when the $i^{\textrm{th}}$ job begins getting serviced
	\item $D_i$: time when the $i^{\textrm{th}}$ job departs	
	\item $A(t)$: arrival process; number of job arrivals during interval $(0,t]$
	\item $D(t)$: departure process; number of job departures during interval $(0,t]$
	\item $Q(t)$: queue length process; number of jobs in the system at time $t$
	\item $R(t)$: virtual response time process; sojourn time that a virtual customer arriving at time $t$ spends in the system
	%\item $W(t)$: the virtual waiting time at time $t$
	%\item $R(t)$: the total remaining workload in the system at time $t$
\end{itemize}

\subsection{The $GI_t/GI_t/1/PS$ queue} \label{chap2:queue}
We consider a single server processor sharing queueing system where arrivals follow an NSNP. We assume that the time-dependent arrival rate function $\lambda(\cdot)$ is continuous and bounded finitely both below and above. Under the assumption, the cumulative arrival function $\Lambda(t)\equiv\int_0^t\lambda(s)\textrm{d}s$ is well-defined for $t\ge 0$ and so is the inverse $\Lambda^{-1}(\cdot)$. 

Each job has its own service requirement, e.g., job size, to be processed by a server. Assume that the job size is determined upon arrival in ICT service systems, e.g., packet size or file size. Let $S_i$ be the service requirement that the $i^{\textrm{th}}$ job brings, and assume that $S_i$'s are independent and identically distributed. Appropriate control schemes dynamically determine service rate function $\mu(\cdot)$. Assume that function $\mu(\cdot)$ is continuous and bounded so that it can be integrate on compact intervals to obtain a cumulative service function $M(t)\equiv\int_{0}^{t}\mu(s)\textrm{d}s$. The amount of service processed by the server during time interval $(t_1,t_2]$ is $M(t_2)-M(t_1)\equiv\int_{t_1}^{t_2}\mu(s)\textrm{d}s$.

The PS policy is a \textit{work-conserving} service discipline which is commonly used to describe computer systems (especially CPUs) \cite{book:G2012}. All jobs in the system evenly share the server or processor at any given time, e.g., if the processor runs at a processing speed of $\mu$ bits/s and there are $n$ jobs, then each job is processed by $\mu/n$ bits/s. 

\subsection{Simulating the $GI_t/GI_t/1/PS$ queue} \label{subsec:simulating_nsnp}
Simulating a \TVGGPS\ queue is difficult and computationally expensive because of non-Poisson arrivals, time nonhomogeneity, processor sharing, and other factors. Therefore, we combine two algorithms \citep{journal:GN2009,journal:MW2016} for simulation. The first algorithm by \citet{journal:GN2009} provides the supporting theory for generating an NSNP from its stationary counterpart, and the second algorithm by \citet{journal:MW2016} gives a numerical approximation method to relieve the computational burden when the rate function is periodic. 

\subsubsection{The arrival process} \label{subsubsec:arrival}
Let $A(t)$ be the NSNP arrival process we want to simulate. Construct the process by applying the \textit{change of time} to a stationary renewal process. Let $N(t)$ be the stationary renewal process with i.i.d. interrenewal times $\{T_i, i \ge 1\}$. Then,
%\begin{subequations}
	\begin{align}
	\EE[N(t)] &= \frac{t}{\tau}, \\
	Var[N(t)] &= \EE[N(t)]SCV(T_i)+o(t),
	\end{align}
%\end{subequations}
where $\tau\equiv\EE[T_i].$
In particular, we call $N(t)$ the standard equilibrium renewal process (SERP) when $\tau=1$ and $T_1$ is a random variable having the stationary excess distribution given by
\begin{align}
F_e(t)\equiv\frac{1}{\EE[T_i]}\int_{0}^t 1-F(s)\textrm{d}s.  
\end{align}
By defining $A(\cdot)$ to be the composition of $N(\cdot)$ and $\Lambda(\cdot)$, i.e., $A(t)= N(\Lambda(t))$ with $\EE[A(t)]=\Lambda(t)$, we construct an NSNP. %Based on the assumption on $\lambda(t)$, $\Lambda(t)$ is continuous and strictly increasing so that it has a well-defined continuous and strictly increasing inverse function $\Lambda^{-1 }(\cdot)$. The inverse function can be defined by $\Lambda^{-1}(x)\equiv\inf\left\{y>0:\int_{0}^{y}\lambda(s)\mathrm{d}s\ge x\right\}$ for any $x\ge 0$. In addition, we define a two parametric function $\Lambda^{-1}(x;a)\equiv\inf\left\{ y>0: \int_a^y\lambda(s)\mathrm{d}s\ge x \right\}$ which is also well-defined on the positive real line.    
%\begin{remark}
%	Given a cumulative arrival rate function $\Lambda$ and the corresponding NSNP $A$, the base renewal process $N(t)$ can be recovered by letting $N=A\circ\Lambda^{-1}$.
%\end{remark}
We generate samples from the arrival process $A(\cdot)$ using the \textit{inversion method} described in \citet{journal:GN2009}. Algorithm \ref{alg:NSNP} describes the procedure. An NSNP $A(t)$ generated by Algorithm \ref{alg:NSNP} has the following property:
\begin{algorithm}
	\caption{The inversion method to generate an NSNP from a stationary renewal process \citep{journal:GN2009}}
	\label{alg:NSNP}
	%	\KwData{$T$, simulation running time}
	%	\KwData{$N$, previously generated stationary renewal time points (1-dimensional array)}
	\KwResult{An 1-dimensional array $A$ such that $A[i]$ is the $i^{\textrm{th}}$ arrival time}
	\Begin
	{
		$T^\textrm{run}\leftarrow$ simulation running time \\
		$A\leftarrow$ empty 1-dimensional array \\
		$A[1]\leftarrow$ random number generated by the equilibrium pdf: $f_e(t)=1-\frac{F(t)}{\EE[T_i]}$ \\
		$n\leftarrow 1$ \\
		\While{$A[n]<T^\textrm{run}$}
		{
			$n\leftarrow n+1$\\
			$x\leftarrow$ random number generated by the df: $F(\cdot)$ // stationary inter-renewal time\\
			$A[n]\leftarrow\Lambda^{-1}(x;A[n-1])$ // $\Lambda^{-1}(x;a)\equiv\inf\left\{ y\ge a: \int_a^y\lambda(s)\mathrm{d}s\ge x \right\}$
		}
		\Return{A}
	}		
\end{algorithm}
Constructing the arrival process prompts the following remark.
\begin{remark}[Gerhardt and Nelson, 2009]
	$\EE[A(t)]=\Lambda(t)$, for all $t\ge 0$, and $Var[A(t)] \approx \Lambda(t)SCV(T_i)$, for large $t$.
\end{remark}
We note that NSNP is a generalization of the simple nonstationary Poisson process (NSPP), where $T_i$ is exponentially distributed. It can be verified easily this by plugging 1 into $SCV(T_i)$.

\subsubsection{The service times} \label{subsubsec:service}
The service completion time is determined as soon as a job arrives when the FCFS discipline applies. Under the PS policy, however, it is not determined upon arrival, because future arrivals will affect the service times of of the jobs already existing in the system. Express the service completion time or the departure time $D_i $ of the $i^{\textrm{th}}$ job that brings a random amount of service requirement $S_i$ as:
%\begin{subequations}
\begin{align}
D_i=\inf{\left\{x\ge A_i: \int_{A_i}^{x}\frac{1}{Q(s)}\mu(s)\textrm{d}s\ge S_i\right\}}, \label{departure_time}
\end{align}
%\end{subequations}
where $A_i$ is the arrival time of the $i^{\textrm{th}}$ job and $Q(s)$ is the number of customers in the system at time $s$. 

\subsubsection{The response time process}
Let $R(t;v)$ denote the entire time that a job spends in the system if it arrives at time $t$ and brings a $v$ amount of service requirement. Since $R(t;v)$ has a what-if characteristic, this is often called \textit{virtual} response time (or virtual sojourn time) at time $t$. When we use $R(t)$ omitting $v$, we still assume a random service requirement. Our primary interest in the \TVGGPS\ queue is the mean virtual response time process $\EE[R(t)]$ for $t\ge 0$. Note that the stochastic nature of $Q(t)$ in Equation \ref{departure_time} means that $R(t)$ cannot be obtained conveniently as its FCFS counterpart where the \textit{Lindely's recursion} is applicable.

To obtain the virtual response time process $\{R(t),t\ge 0\}$ in a $GI_t/GI_t/1/PS$ queue, we store the \textit{path} of the queue for every replication of the simulation. The path contains the status of the system at each recording epoch. After a replication is terminated, re-run the simulations from each recording epoch during a replication length (say $t_1,t_2,\ldots$), given the stored status at time $t_k$, with a newly inserted job which is the \textit{virtual job}. Each re-run of the simulation terminates when the virtual job is finished and results in a realization of a virtual response time $R(t_k)$. We obtain the expected process $\{\EE[R(t)],t\ge 0\}$ by averaging at 10,000 replications.

\section{Methods}\label{chap2:method}
As mentioned in Section~\ref{sec:intro}, we combine the pointwise stationary approximation (PSA) and the heavy-traffic approximation, which were used by \citet{journal:W2015} and \citet{proc:MW2015} to stabilize the waiting times (excluding service times) in $GI_t/GI_t/1/FCFS$ queues, and adjust the combined approximations to stabilize the response times (waiting time + service time) in $GI_t/GI_t/1/PS$ queues. Below, we explain our methods.
\subsection{Pointwise stationary approximation with heavy-traffic limits} \label{subsec:ht}
We briefly visit the pointwise stationary approximation (PSA) \citep{journal:GK1991,journal:W1991}, which is known to be an appropriate approximation when the arrival rate changes slowly relative to the average service time \citep{journal:W1991,journal:W2015}. Thus, we consider that the performance at different times is similar to the performance of the stationary counterpart with the instantaneous model parameters. %PSA is suggested after the discovery of the fact that a stationary model can seriously underestimate delays even when the arrival rate is only modestly nonstationary \cite{GKS1991}. This finding raised serious concerns about the prevalent use of a simple stationary model for almost any real system in that period.

%\subsubsection{Heavy-traffic approximations}
The heavy-traffic limit theory for \GGPS\ queues was initially developed by \citet{journal:G1994} and further studied by \citet{journal:G2004} and \citet{journal:ZZ2008}. \citet{journal:ZZ2008} provide the following approximate mean virtual response time ($R$) in steady state for \GGPS queues: 
\begin{align}
\EE[R]\approx\frac{\beta}{\mu}\cdot\frac{1}{1-\rho}\cdot V_{PS}. \label{rt_gg1ps}
\end{align} 

\subsection{Two service rate controls} \label{subsec:controls}
\citet{journal:W2015} derived the PSA-based service rate control to stabilize the \emph{waiting time}. We take a similar approach, but our service rate control stabilizes the \emph{response time}. We derive two service rate controls based on $GI/GI/1/FCFS$ and \GGPS\ heavy-traffic approximations. Hereinafter, we use the subscripts $FCFS$ and $PS$ to indicate the discipline from which the result derives, e.g., variability factor $V_{FCFS}$ and $V_{PS}$. 
%The reason why we derive a control from the \GG\ result is originated from the conjecture that \TVGGPS\ queue might also behave like \TVGG\ in light-traffic since the probability that two or more jobs simultaneously present in the system is extremely low.  

%\begin{conjecture}
%	The performance measures of \TVGGPS converges to those of \TVGG\ as $\rho(t)\to 0$.
%\end{conjecture}	

\subsubsection{The square-root (SR) control}
In queueing systems, the workload processes are identical under any work-conserving disciplines. Thus, we derive a control based on a $GI/GI/1/FCFS$ queue as an experimental trial, which we later discover to be appropriate for \TVGGPS\ queues under light-traffic conditions (see Section \ref{subsec:interpretation_lt} for the details). 

The heavy-traffic approximation for the expected steady state response time in a $GI/GI/1/FCFS$ queue is \citep{book:CY2001}: 
\begin{align}
\EE[R_{FCFS}]\approx\frac{\beta}{\mu}+\frac{\beta}{\mu}\cdot\frac{\rho}{1-\rho}\cdot V_{FCFS}, \label{rt_gg1fcfs}
\end{align} 
where $\mu$ is the service rate, $\beta$ is the mean job size, $\rho$ is the traffic intensity, and $V_{FCFS}\equiv (C_a^2+C_s^2)/2$ is the variability parameter, given the SCVs for the arrival base and job size distributions. Approximate the expected response time at time $t$ in a $GI_t/GI_t/1/FCFS$ queue based on the PSA:
\begin{align}
\EE[R_{FCFS}(t)]\approx\frac{\beta}{\mu(t)}+\frac{\beta}{\mu(t)}\cdot\frac{\rho(t)}{1-\rho(t)}\cdot V_{FCFS}, \label{htpsa_fcfs}
\end{align} 
where $\rho(t)\equiv \lambda(t)\beta/\mu(t)$ is the instantaneous traffic intensity at time $t$. Fixing the LHS by a target response time $s$ and adjusting the terms gives:
\begin{align}
s\mu(t)^2-\beta\left(s\lambda(t)+1\right)\mu(t)+\lambda(t)\beta^2(1-V_{FCFS})=0.
\end{align}
Finally, obtain the solution to the quadratic equation above:
\begin{align}
\mu_{SR}(t;s)\equiv\frac{(s\lambda(t)+1)\beta+\sqrt{(s\lambda(t)+1)^2\beta^2+4s\lambda(t)\beta^2(V_{FCFS}-1)}}{2s}. \label{eqn:srcontrol}
\end{align}
We call Equation~(\ref{eqn:srcontrol}) the \textit{square-root} (SR) control, which is the naming convention used by \citet{journal:W2015}.

\subsubsection{The difference-matching (DM) control}
Recall that Equation~(\ref{rt_gg1ps}) is the heavy-traffic approximation for the steady-state mean virtual response time ($R_{PS}$) in a \GGPS\ queue:
\begin{align*}
\EE[R_{PS}]\approx\frac{\beta}{\mu}\cdot\frac{1}{1-\rho}\cdot V_{PS}, 
\end{align*} 
where $\mu$, $\beta$, and $\rho$ are defined as in Equation (\ref{rt_gg1fcfs}), and $V_{PS}\equiv(C_a^2+C_s^2)/(1+C_s^2)$ is the variability parameter for the PS queue. Approximate the expected response time process based on the PSA:
\begin{align}
\EE[R_{PS}(t)]\approx\frac{\beta}{\mu(t)}\cdot\frac{1}{1-\rho(t)}\cdot V_{PS}. \label{htpsa_ps}
\end{align} 
Fixing the LHS by a certain constant $s$ and adjusting the terms gives a service rate control that is much simpler than $\mu_{SR}(t)$:
\begin{align}
\mu_{DM}(t;s)\equiv\beta\left(\lambda(t)+\frac{V_{PS}}{s}\right). \label{eqn:dmcontrol}
\end{align}
As mentioned in Section~\ref{sec:intro}, we call Equation~(\ref{eqn:dmcontrol}) the \textit{difference-matching} (DM) control because $\mu_{DM}(t;s) - \beta\lambda(t)$ is a constant $\beta V_{PS}/s$. 

\subsubsection{Simple analysis on the two service rate controls}
The two controls derived above result in different service rate functions except when both base distributions have SCVs 1. The most representative example is the $M_t/M_t/1/PS$ queue. Applying $V_{FCFS}=1$, the SR control (\ref{eqn:srcontrol}) reduces to $\beta\left(\lambda(t)+1/s\right)$, which is the same as the DM control (\ref{eqn:dmcontrol}) with $V_{PS}=1$. It prompts the following remark.
\begin{remark} \label{remark:coincision}
	For the time-varying queues having both the base distributions (arrival base and job size) of SCV 1, the two controls coincide.
\end{remark}

Another simple but interesting phenomenon is that both controls become identical as we decrease or increase the target response time $s$.
\begin{proposition} \label{proposition:asymptotic}
	The two controls coincide as $s\to\infty$ (heavy-traffic) or $s\to 0$ (light-traffic).
\end{proposition}
\begin{proof}
	Both the SR control in Equation~(\ref{eqn:srcontrol}) and the DM control in Equation~(\ref{eqn:dmcontrol}) converge to $\beta\lambda(t)$ as $s\to\infty$ and the traffic intensity converges to 1.
	On the other hand, taking $s\to 0$ results in $\mu_{SR}(t;s)\to\infty$ and $\mu_{DM}(t;s)\to\infty$, which implies that the traffic intensity becomes zero.	
\end{proof}

\section{Simulation experiments}\label{chap2:simulation}
We investigate the performance of the two service rate controls through simulation experiments. Table~\ref{table:doe} summarizes the simulation parameters. 
\subsection{Simulation setting}
We use the sinusoidal arrival rate function $\lambda(t)=a+b\sin{(\gamma t)}$ with constants $a=1$, $b=0.2$, and $\gamma=0.001,0.01,0.1$. Therefore, we have three functions of the same amplitude but of different periods. Two are the slowly time-varying functions ($\gamma=0.001,0.01$) and the third one ($\gamma=0.1$) is not. We include the third, however, to observe how the controls work when the arrival rate is a quickly time-varying function. 

To observe the asymptotic behavior, we set the replication length to at least three cycles of the periods, e.g., we conduct simulations for period $\gamma=0.001$ on a 20,000 unit time and periods $\gamma=0.01,0.1$ on a 2,000 unit time considering the length of periods. For the target response time $s$, we use two different values: 0.1 for the short and 10.0 for the long response times. Because the service rate controls are inversely proportional to $s$, each value of $s$ results in light-traffic and heavy-traffic, respectively. For each independent system, we conduct 10,000 replications to obtain the ensemble average of the performance measures.

We consider three different distributions for arrival base and job size distribution: Erlang distribution (ER); exponential distribution (EXP); and lognormal distribution (LN). The distributions have mean 1 and different SCVs. Specifically, we use ER with $SCV=0.5$ and LN with $SCV=2$. The SCV of EXP is always 1 by definition. We make five pairs of base arrival/job size distributions: EXP/EXP, ER/ER, LN/LN, ER/LN, and LN/ER. Note that the combination EXP/EXP corresponds to a queueing system with NSPP arrival process and exponential service requirement, i.e., $M_t/M_t/1/PS$. Table \ref{table:variability} summarizes the variability factors, $V_{FCFS}$ and $V_{PS}$, associated with each pair of distributions.
\begin{table}
	\centering
	\caption{Simulation parameters}
	\label{table:doe}
	\begin{tabular}{|l|c|l|l|l|l|l|}
		\hline
		System                               & \multicolumn{6}{c|}{\TVGGPS}       \\ \hline
		Arrival rate function & \multicolumn{6}{c|}{$\lambda(t)=1+0.2\sin{(\gamma t)}$}           \\
		Periodic coefficient                 & \multicolumn{6}{c|}{$\gamma=0.001,0.01,0.1$}           \\
		Replication length                   & \multicolumn{6}{c|}{$l=20000,2000,2000$}            \\ \hline
		Service rate function                  & \multicolumn{6}{c|}{$\mu_{SR}(t)$, $\mu_{DM}(t)$}      \\ \hline
		Target response time & \multicolumn{6}{c|}{0.1 (light-traffic), 10.0 (heavy-traffic)} \\ \hline
		
		Number of replication		          & \multicolumn{6}{c|}{10000}                            \\ \hline
		& \multicolumn{6}{c|}{Exponential (SCV=1.0)} \\
		Distributions            & \multicolumn{6}{c|}{Erlang (SCV=0.5)}           \\
		& \multicolumn{6}{c|}{Lognormal (SCV=2.0)}     \\ \hline
	\end{tabular}
\end{table}

\begin{table}
	\centering
	\caption{Variability factor for each distribution pair}
	\label{table:variability}
	\begin{tabular}{|l|c|c|}
		\hline
		Distribution pair (arrival base/job size) & $V_{FCFS}$ & $V_{PS}$ \\ 		\hline
		Exponential/Exponential & 1 & 1 \\ 
		Erlang/Erlang & 0.5 & 0.6667 \\ 
		Lognormal/Lognormal & 2 & 1.3333 \\ 
		Erlang/Lognormal & 1.25 & 0.8333 \\ 
		Lognormal/Erlang & 1.25 & 1.6666 \\ 		\hline
	\end{tabular} 
\end{table}

\subsection{Two metrics to evaluate the effectiveness of the controls} \label{subsec:metric}
We use two metrics to measure the performance of the two controls. First, we define the relative amplitude (RA) by
\begin{align*}
\frac{\textrm{amplitude of $\EE[R(t)]$}}{\textrm{spatial average of $\EE[R(t)]$}}\times 100\%,
\end{align*}
as a measure of \textit{stabilization}. Second, we define the relative gap (RG) by 
\begin{align*}
\frac{\textrm{target response time}-\textrm{spatial average of $\EE[R(t)]$}}{\textrm{spatial average of $\EE[R(t)]$}}\times 100\%,
\end{align*}
as a measure of \textit{accuracy}. We obtain the two metrics by numerically calculating:
\begin{align}
&\textrm{RA}\left[\EE[R(t)]\right]\nonumber \\
&\equiv\frac{\mathcal{T}_{\EE\circ R}\times\left[\max_{t\in[x,x+\mathcal{T}_{\EE\circ R}]}\left\{\EE[R(t)]\right\}-\min_{t\in[x,x+\mathcal{T}_{\EE\circ R}]}\left\{\EE[R(t)]\right\}\right]}{2\times \int_x^{x+\mathcal{T}_{\EE\circ R}}\EE[R(t)]dt}\times 100\%, \\
&\textrm{RG}\left[\EE[R(t)]\right]\equiv\frac{s-\frac{\int_x^{x+\mathcal{T}_{\EE\circ R}}\EE[R(t)]dt}{\mathcal{T}_{\EE\circ R}}}{s}\times 100\%,
\end{align}
where $\mathcal{T}_{\EE\circ R}$ is the period of the expected response time process $\EE[R(t)]$, $x$ is an arbitrary long time after the process has been stabilized, and $s$ is the target response time. For the values of $\mathcal{T}_{\EE\circ R}$, we use the same values as the periods of the arrival rate functions since we observe that the periods are the same for both $\EE[R(t)]$ and $\lambda(t)$.

The two measures above are favorable as they become closer to 0\%. For RA, there is no negative value since the amplitude is a positive amount. Note, however, that RG allows a negative value such that the control overestimates the service rate which gives a smaller spatial average than our original target.

\subsection{Results} \label{subsec:result}
Table \ref{table:tvgg1ps_expexp}-\ref{table:tvgg1ps_lner} in \ref{app:performance} report both the absolute values (amplitude and spatial average) and the relative values (RA and RG). For the performance of the controls, we heuristically call them \textit{good} if they control the response time with $RA\le 10\%$ and $|RG|\le 0.1\%$, and \textit{poor} otherwise. 

In the following plots, the green line corresponds to the mean queue length $\EE[Q(t)]$ and the red line to the mean virtual response time $\EE[R(t)]$ of the simulated \TVGGPS\ queues under the various combinations of control and distribution. The dotted black lines are the 95\% confidence intervals, and the dotted blue line plots the arrival rate function and has its dedicated y-axis on the right.

In the following subsections, we summarize the results of Tables \ref{table:tvgg1ps_expexp}-\ref{table:tvgg1ps_lner} by their traffic intensity. We obtain each traffic intensity by targeting the response time (short or long) we desire according to Proposition \ref{proposition:asymptotic}. Specifically, the instantaneous traffic intensity is approximately $\rho(t)\approx 0.1$ when $s=0.1$ and $\rho(t)\approx 0.9$ when $s=10.0$, for the distribution pairs.

% 여기서부터 다시 시작

\subsubsection{Control performances in light-traffic systems (s=0.1)} \label{subsec:result:lt}
Figures \ref{fig:s_01_sr_erer} and \ref{fig:s_01_dm_erer} depict the two expected processes $\EE[Q(t)]$ and $\EE[R(t)]$ in light-traffic systems under the two controls where the base distribution pair is Erlang/Erlang. The figures show universally good stabilizing performances ($|RA|\le 5\%$), even under the quickly time-varying arrival rate ($\gamma=0.1$), but, the accuracy of the DM control is poor. Specifically, the expected response time process stabilizes around 0.15 although the target is 0.1, which corresponds to about 0.5\% of RG (Figure \ref{fig:s_01_dm_erer}). Intuitively, this poor performance stems from the inaccuracy of the heavy-traffic approximation in light-traffic systems. Meanwhile, the SR control results in only about 0.05\% of RG (Figure \ref{fig:s_01_sr_erer}). Throughout the simulation experiments, we observe this tendency consistently from all of the distribution pairs (see Section \ref{subsec:interpretation_lt} for the details).

\begin{figure}
	\centering
	\subfloat[][$\gamma=0.1$]
	{
		\centering\resizebox{0.28\textwidth}{!}{\includegraphics{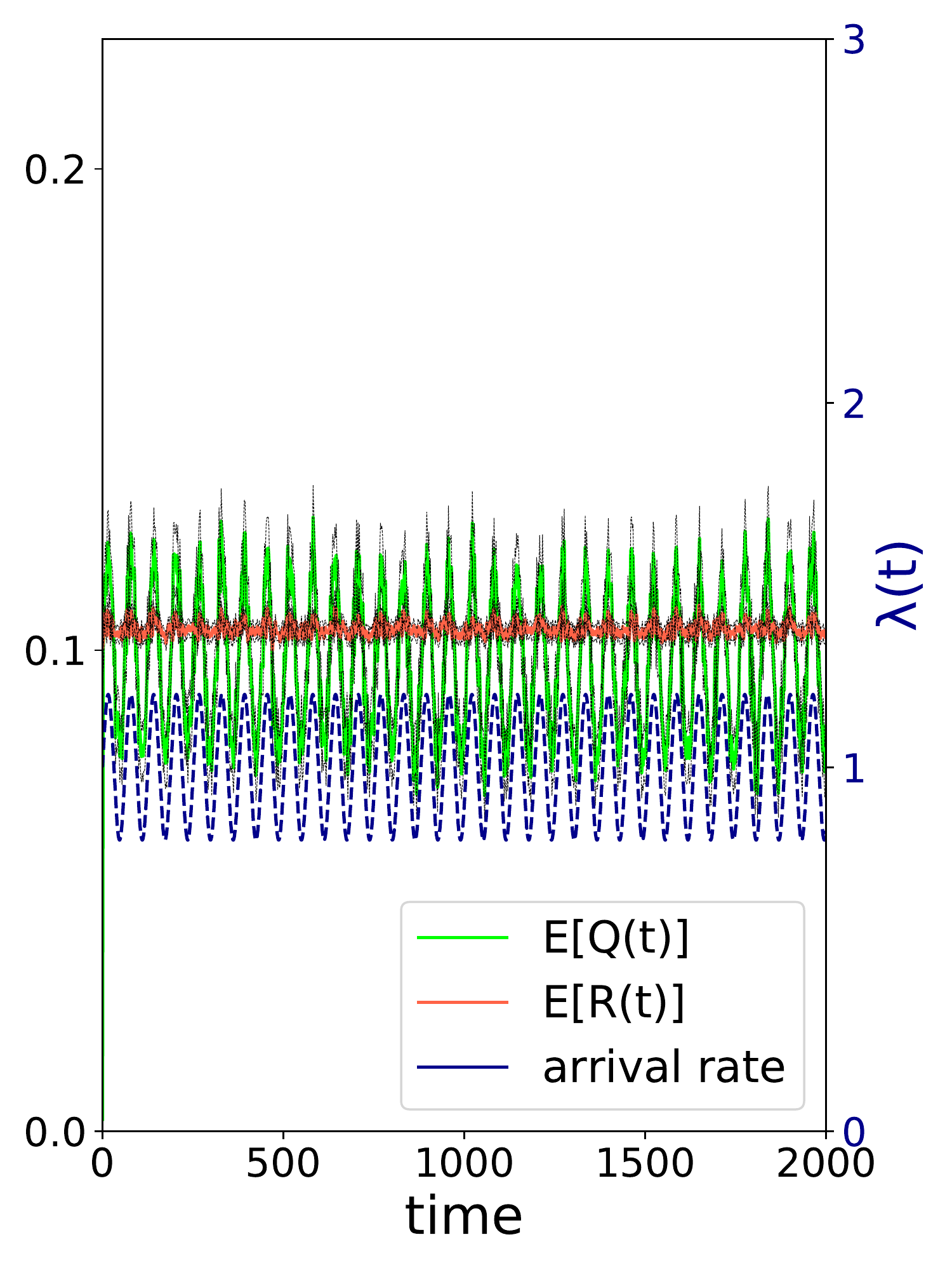}}
		\label{fig:s_01_sr_erer_01}
	}
	~
	\subfloat[][$\gamma=0.01$]
	{
		\centering\resizebox{0.28\textwidth}{!}{\includegraphics{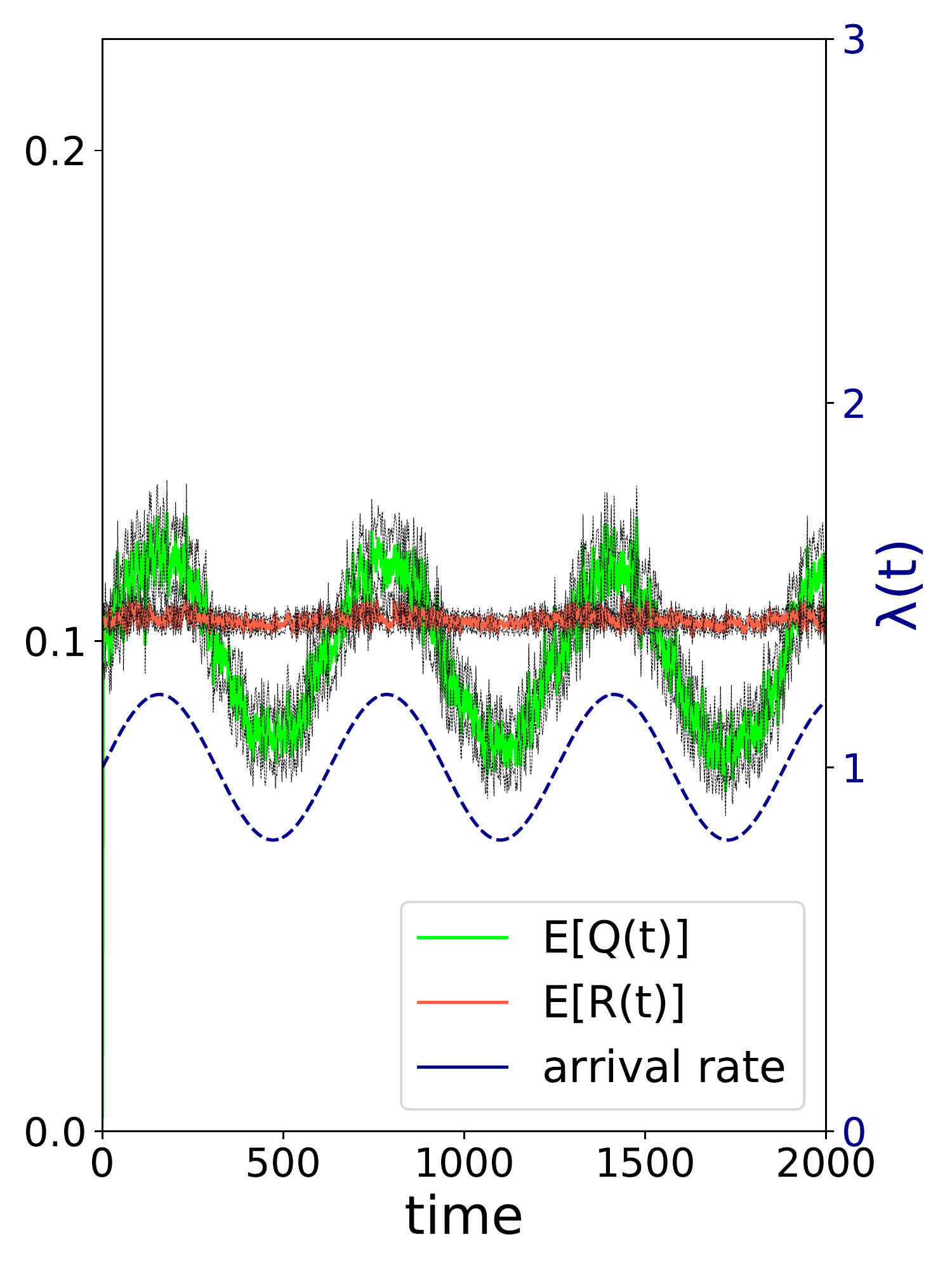}}
		\label{fig:s_01_sr_erer_001}
	}
	~
	\subfloat[][$\gamma=0.001$]
	{
		\centering\resizebox{0.28\textwidth}{!}{\includegraphics{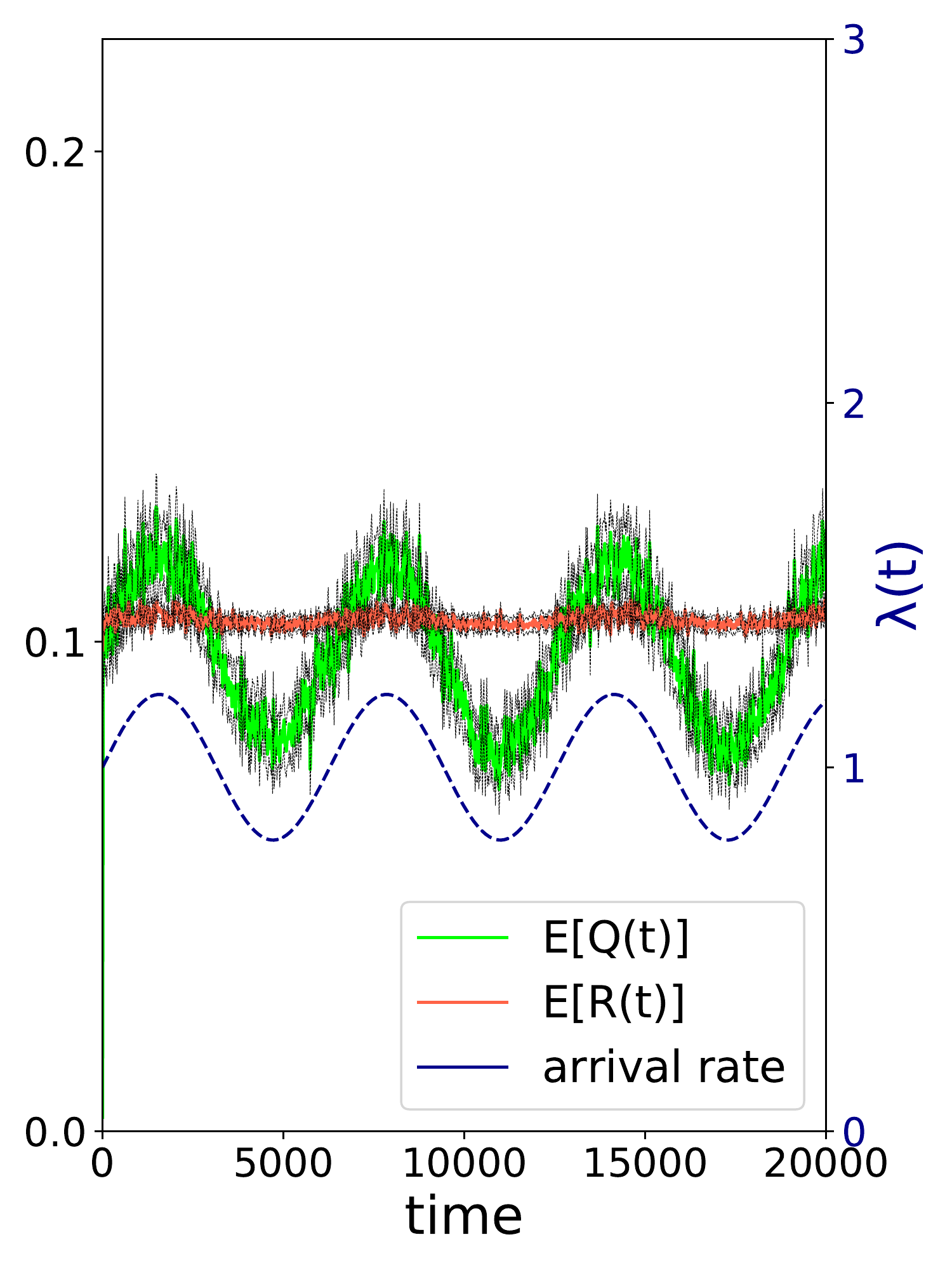}}
		\label{fig:s_01_sr_erer_0001}
	}
	
	\caption{General performance measures of $ER_t/ER_t/1/PS$ queues under SR control where $\lambda(t)=1+0.2\sin{(\gamma t)}$ with target response time $0.1$ (light-traffic)}
	\label{fig:s_01_sr_erer}
\end{figure}

\begin{figure}
	\centering
	\subfloat[][$\gamma=0.1$]
	{
		\centering\resizebox{0.28\textwidth}{!}{\includegraphics{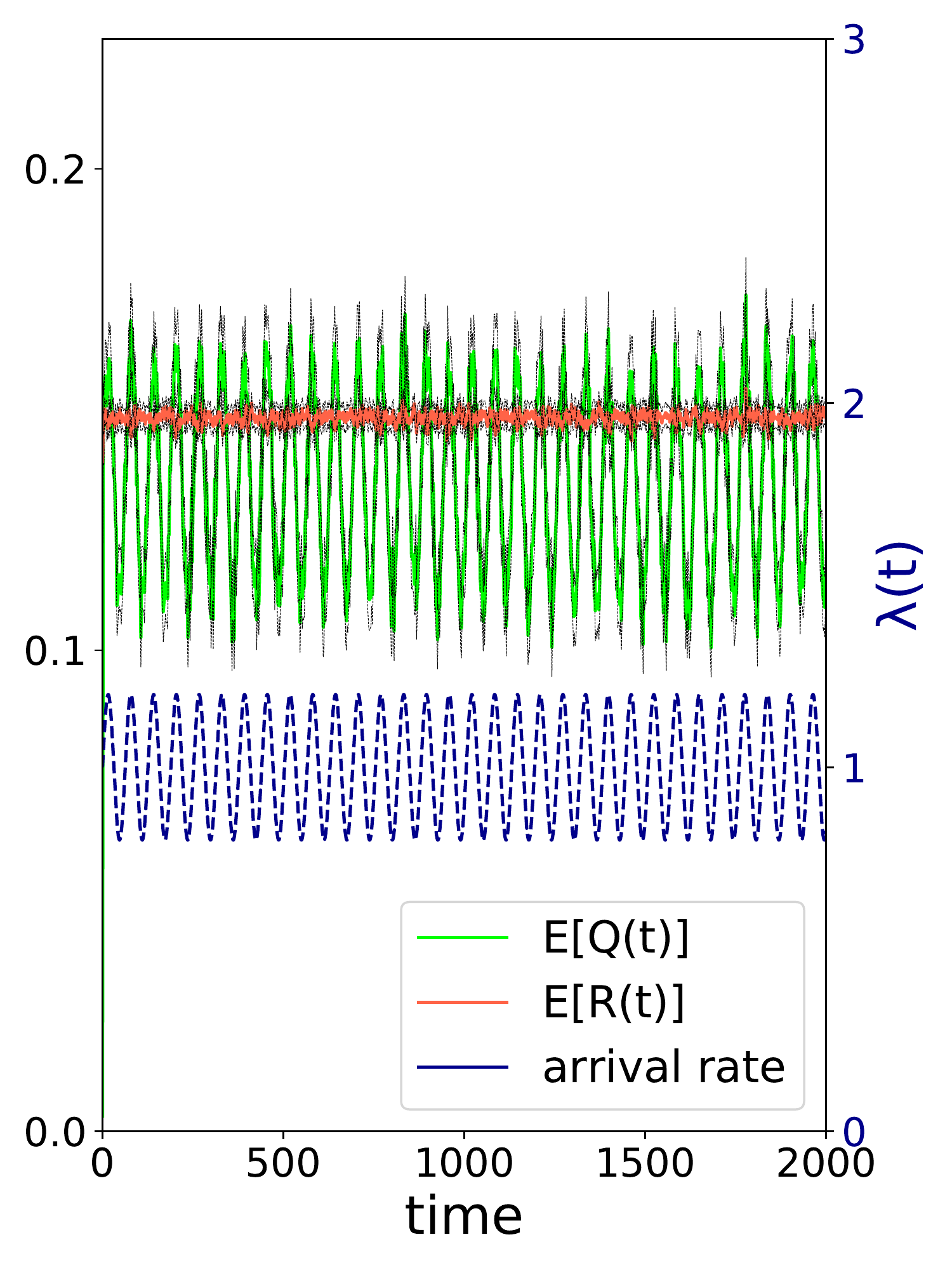}}
		\label{fig:s_01_dm_erer_01}
	}
	~
	\subfloat[][$\gamma=0.01$]
	{
		\centering\resizebox{0.28\textwidth}{!}{\includegraphics{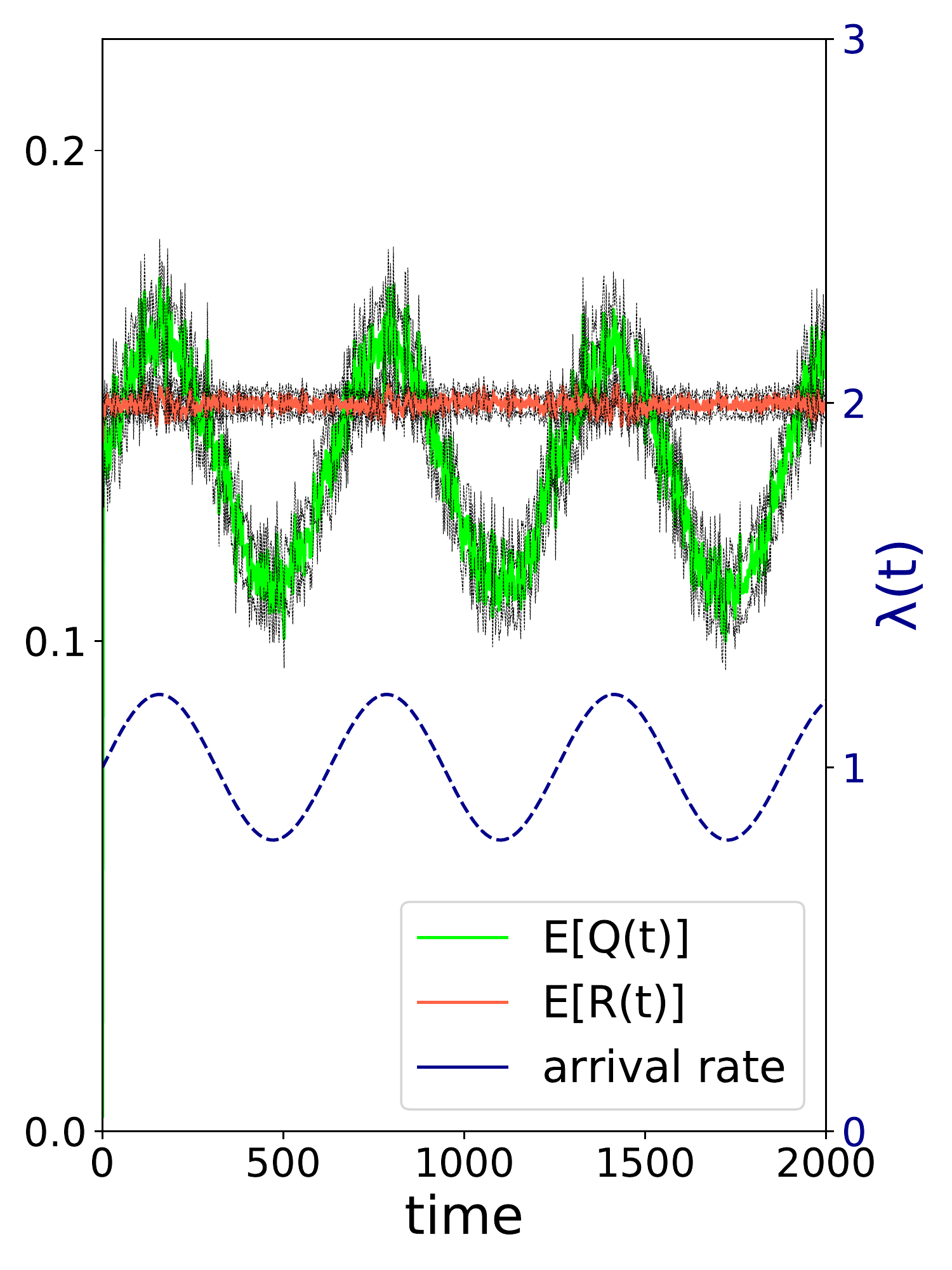}}
		\label{fig:s_01_dm_erer_001}
	}
	~
	\subfloat[][$\gamma=0.001$]
	{
		\centering\resizebox{0.28\textwidth}{!}{\includegraphics{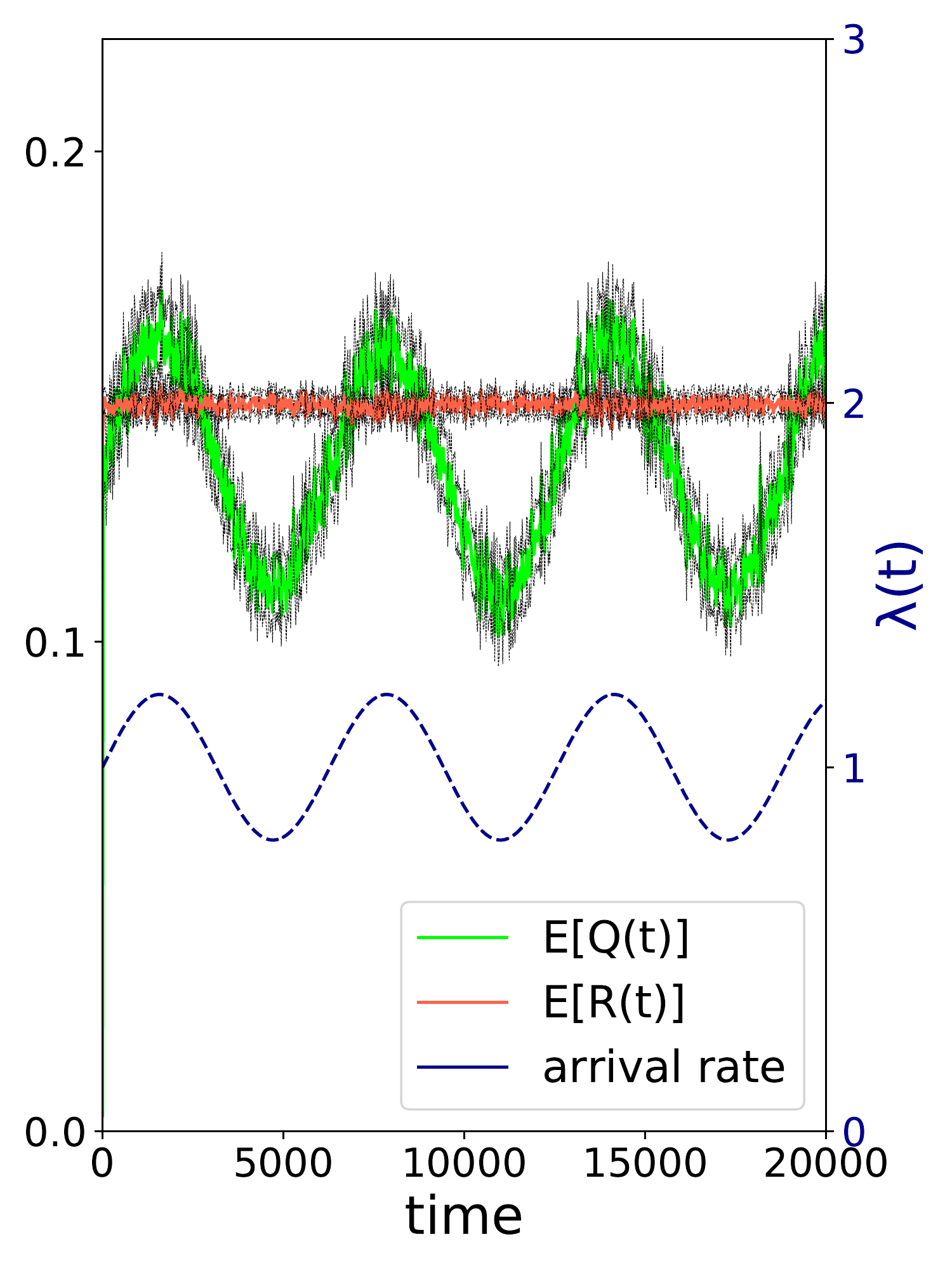}}
		\label{fig:s_01_dm_erer_0001}
	}
	
	\caption{General performance measures of $ER_t/ER_t/1/PS$ queues under DM control where $\lambda(t)=1+0.2\sin{(\gamma t)}$ with target response time $0.1$  (light-traffic)}
	\label{fig:s_01_dm_erer}
\end{figure}

\subsubsection{Control performances in heavy-traffic systems (s=10)} \label{subsec:result:ht}
Figures \ref{fig:s_10_sr} and \ref{fig:s_10_dm} depict the heavy-traffic systems under the two controls and three pairs of base distributions (EXP/EXP, ER/ER, LN/LN). Compared to the light-traffic systems, we do not observe perfectly controlled results. For quickly time-varying arrival rate ($\gamma=0.1$), the poor control performance is obvious since the PSA is not appropriate. We observe positive results for the slowest time-varying arrival rate ($\gamma=0.001$) despite the imperfect stabilization. The consistently better accuracy of DM control justifies its use in heavy-traffic systems. 

\begin{figure}
	\centering
	\subfloat[][EXP/EXP, $\gamma=0.1$]
	{
		\centering\resizebox{0.27\textwidth}{!}{\includegraphics{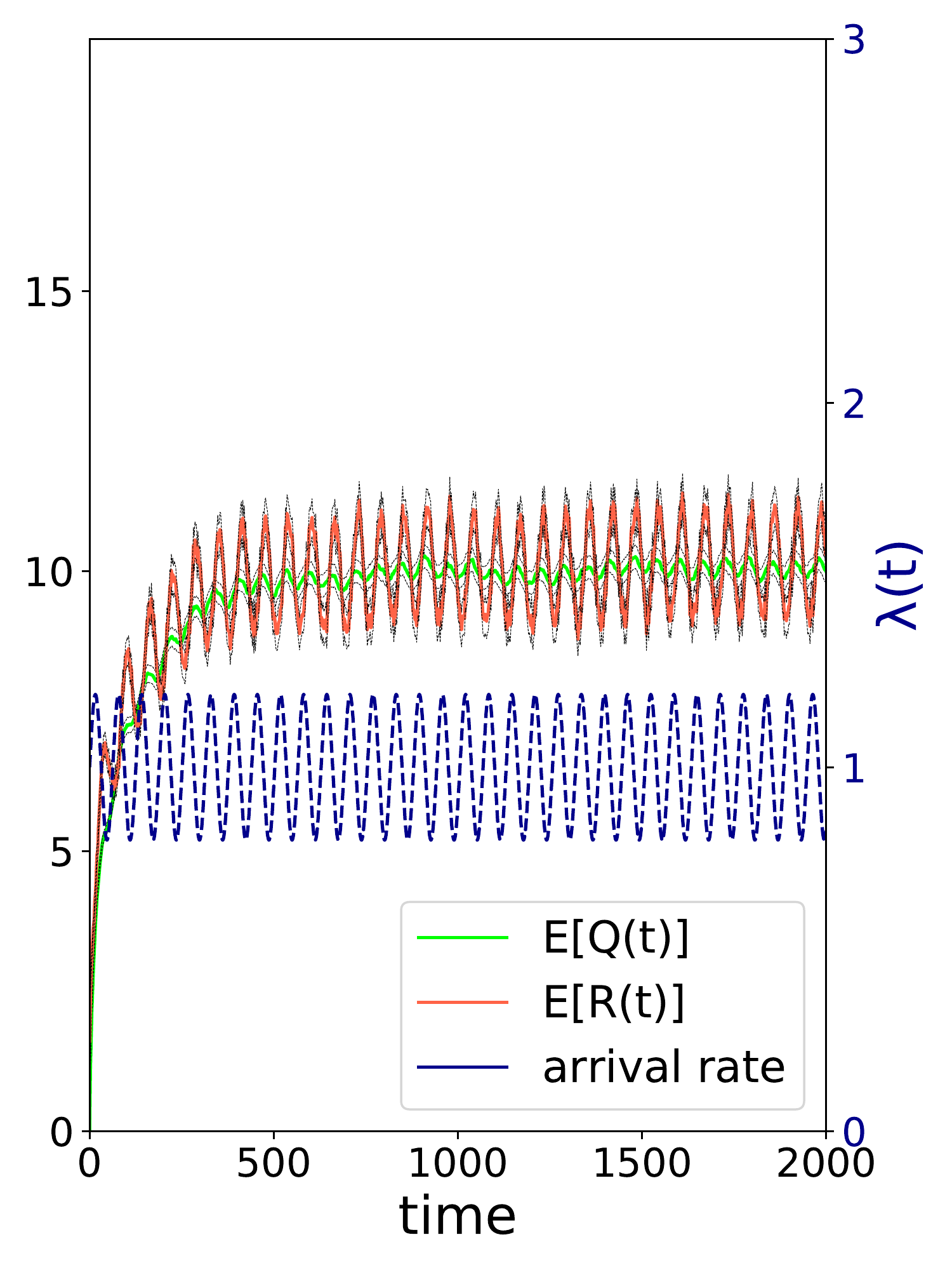}}
		\label{fig:s_10_sr_expexp_01}
	}
	~
	\subfloat[][EXP/EXP, $\gamma=0.01$]
	{
		\centering\resizebox{0.27\textwidth}{!}{\includegraphics{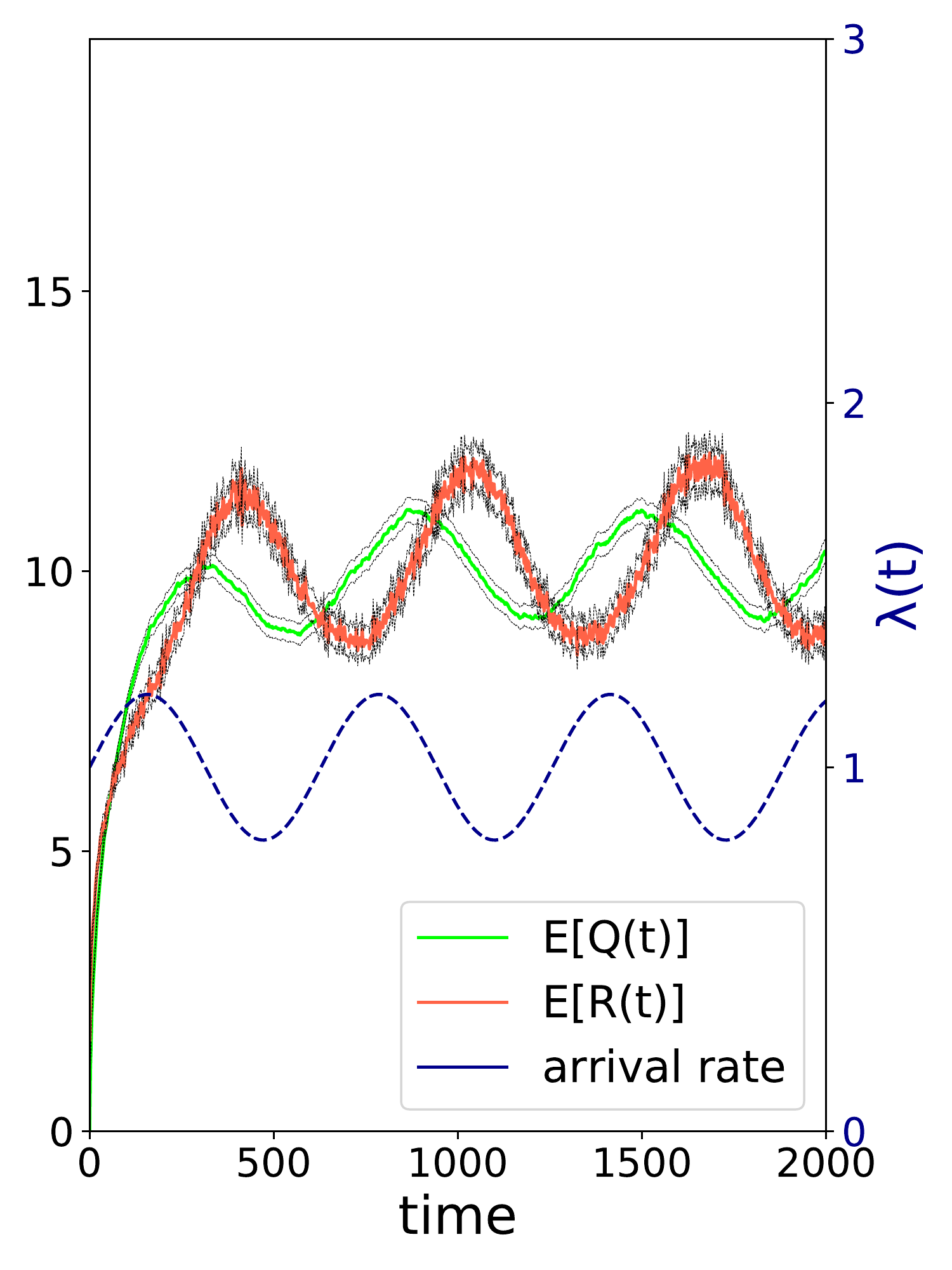}}
		\label{fig:s_10_sr_expexp_001}
	}
	~
	\subfloat[][EXP/EXP, $\gamma=0.001$]
	{
		\centering\resizebox{0.27\textwidth}{!}{\includegraphics{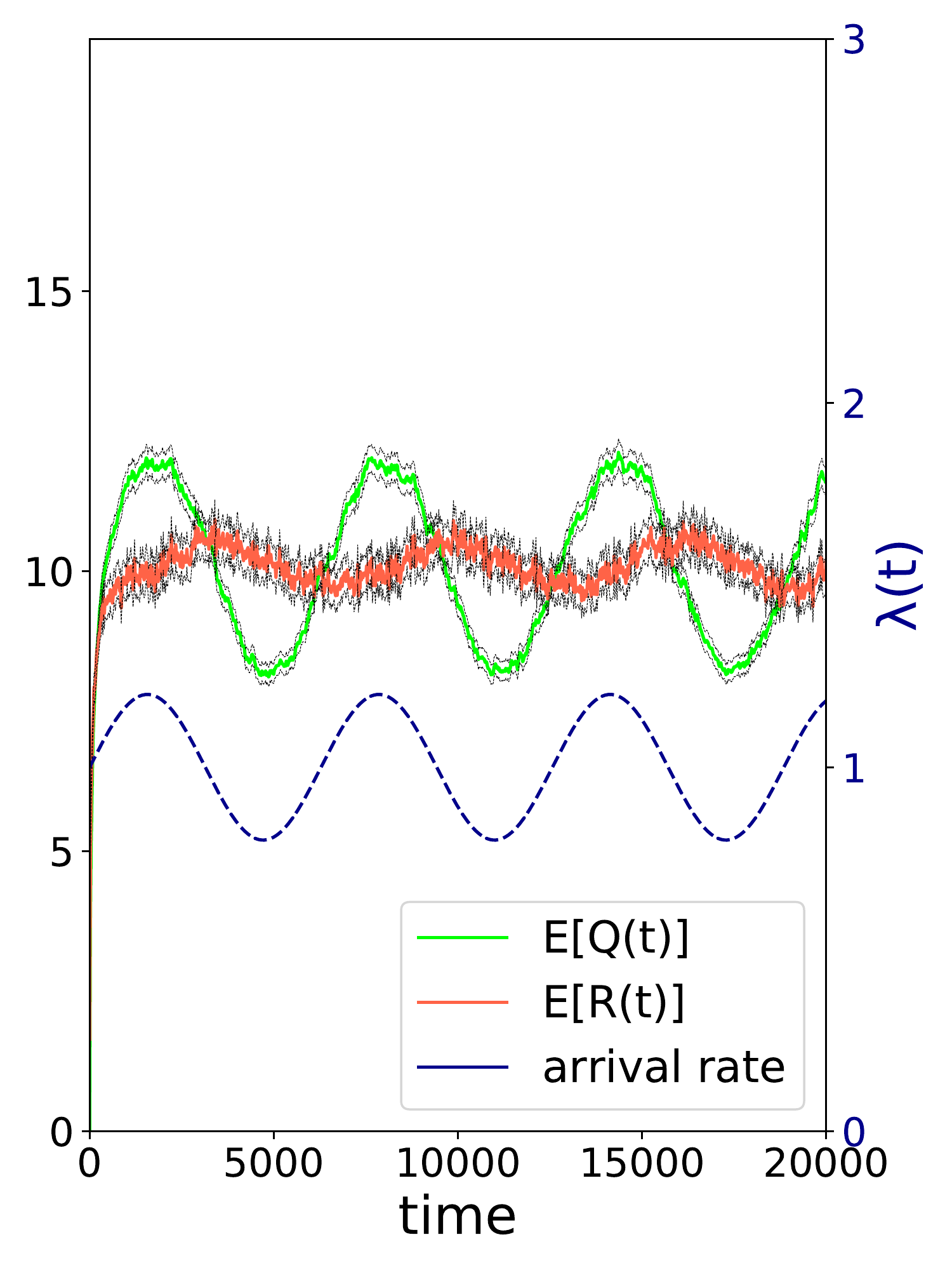}}
		\label{fig:s_10_sr_expexp_0001}
	}
	
	\subfloat[][ER/ER, $\gamma=0.1$]
	{
		\centering\resizebox{0.27\textwidth}{!}{\includegraphics{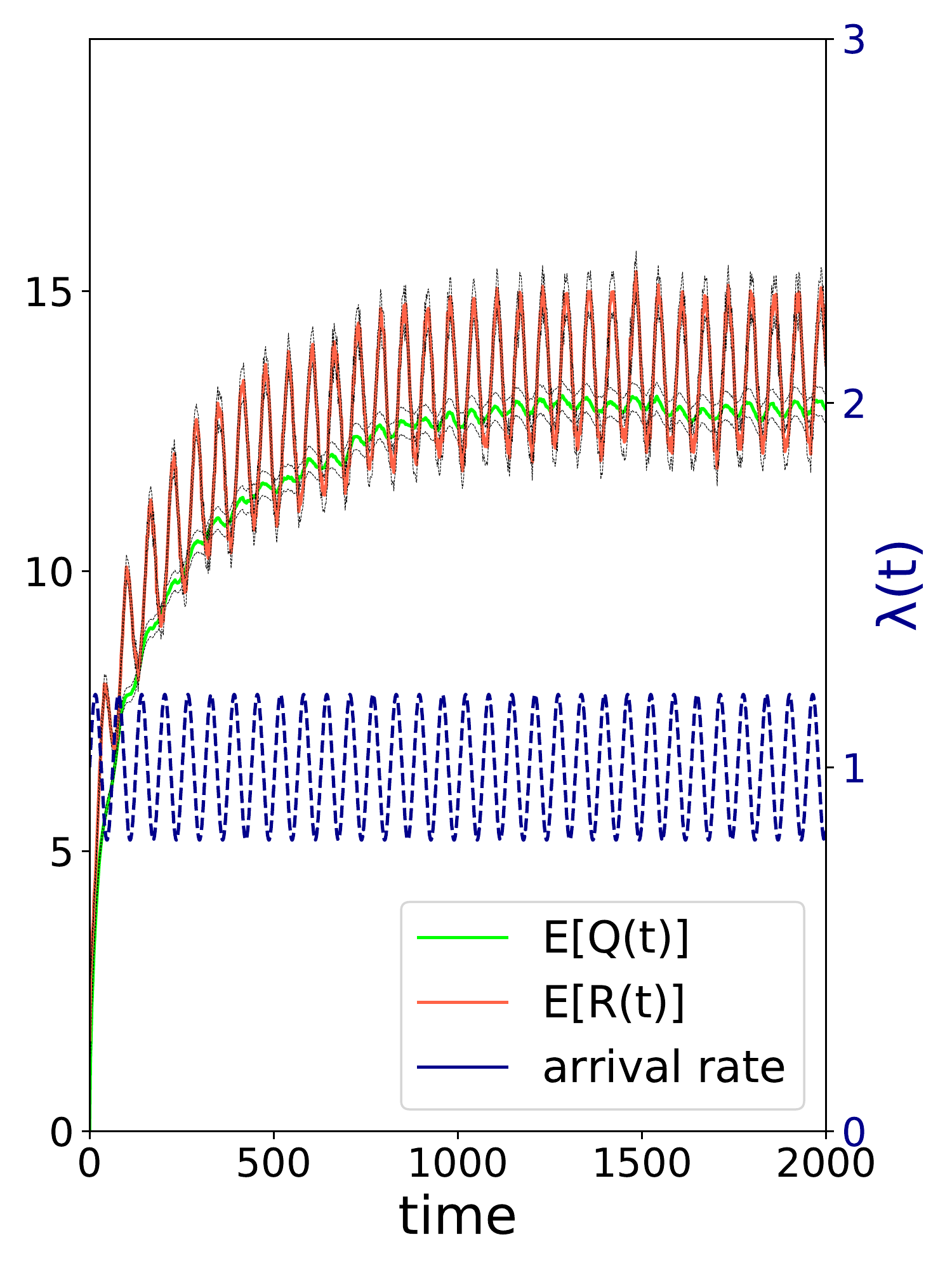}}
		\label{fig:s_10_sr_erer_01}
	}
	~
	\subfloat[][ER/ER, $\gamma=0.01$]
	{
		\centering\resizebox{0.27\textwidth}{!}{\includegraphics{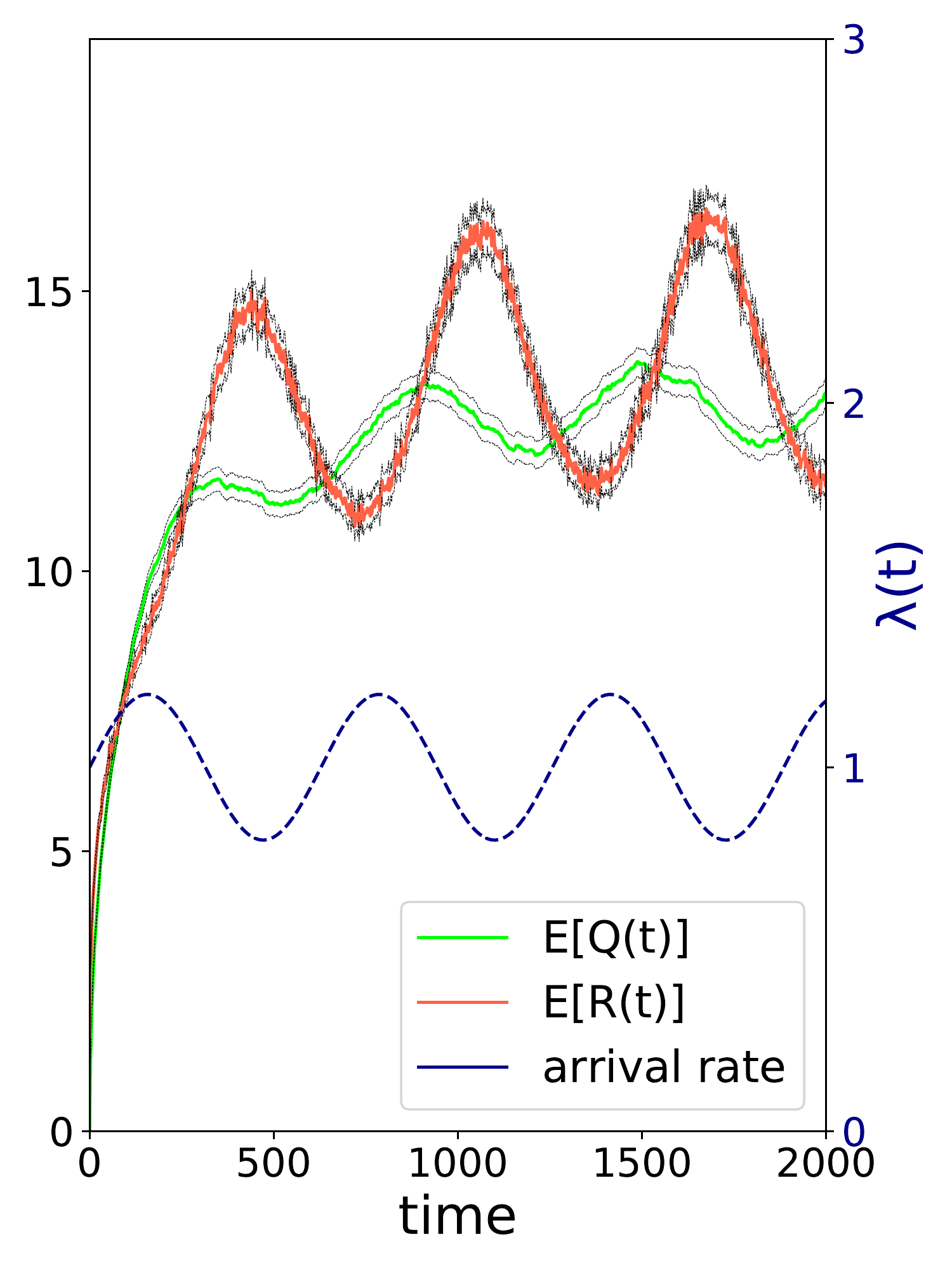}}
		\label{fig:s_10_sr_erer_001}
	}
	~
	\subfloat[][ER/ER, $\gamma=0.001$]
	{
		\centering\resizebox{0.27\textwidth}{!}{\includegraphics{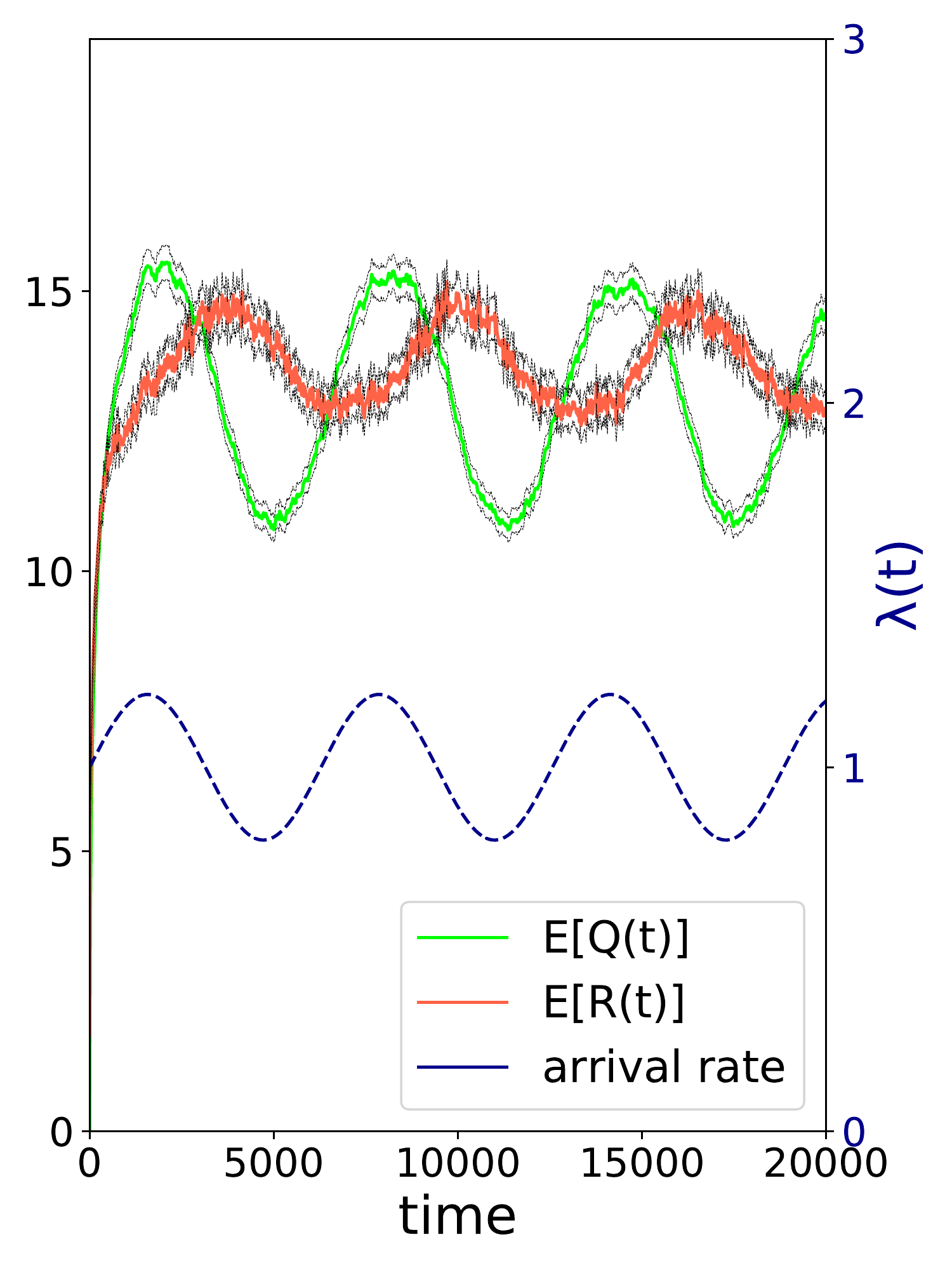}}
		\label{fig:s_10_sr_erer_0001}
	}
	
	\subfloat[][LN/LN, $\gamma=0.1$]
	{
		\centering\resizebox{0.27\textwidth}{!}{\includegraphics{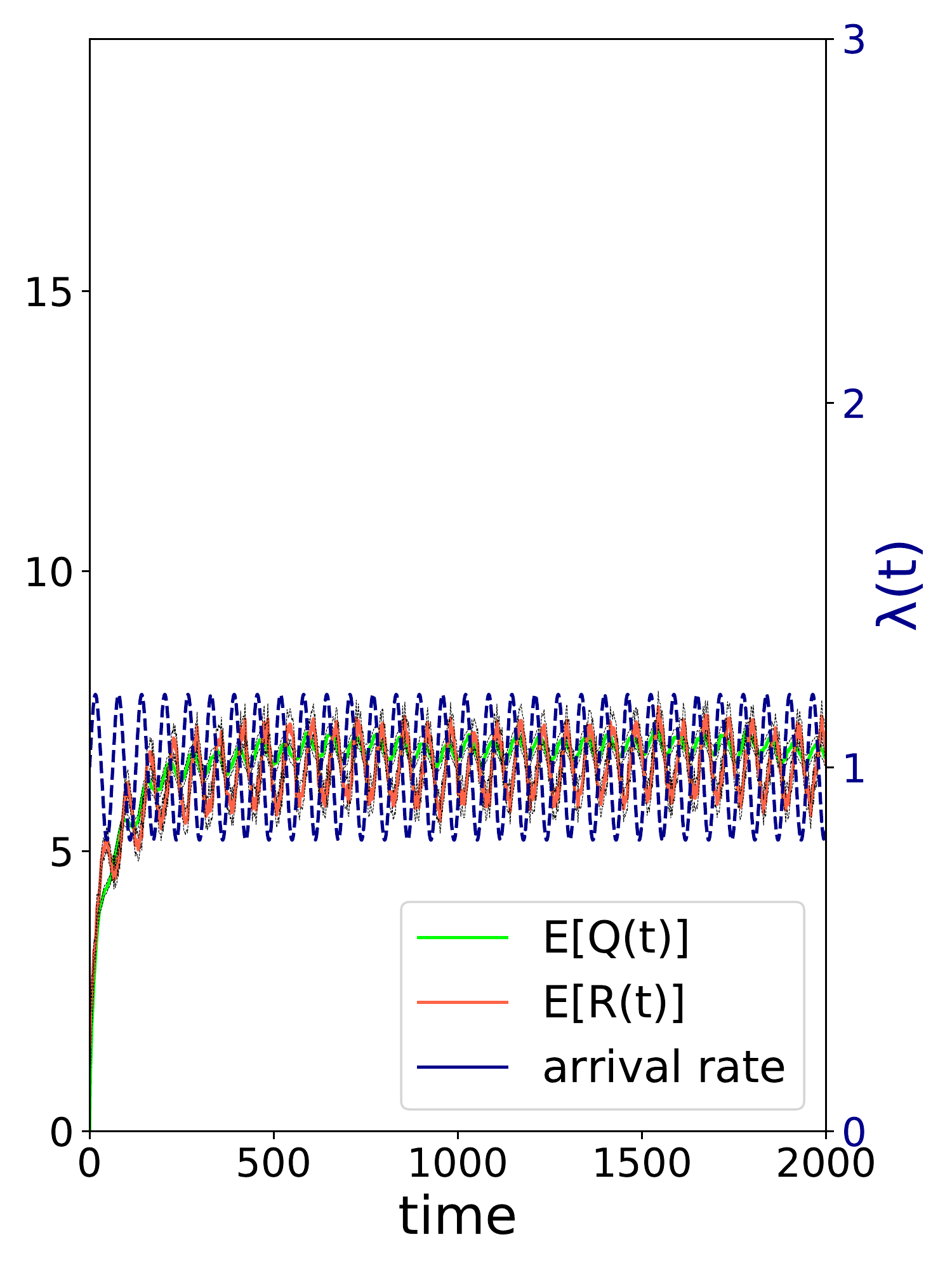}}
		\label{fig:s_10_sr_lnln_01}
	}
	~
	\subfloat[][LN/LN, $\gamma=0.01$]
	{
		\centering\resizebox{0.27\textwidth}{!}{\includegraphics{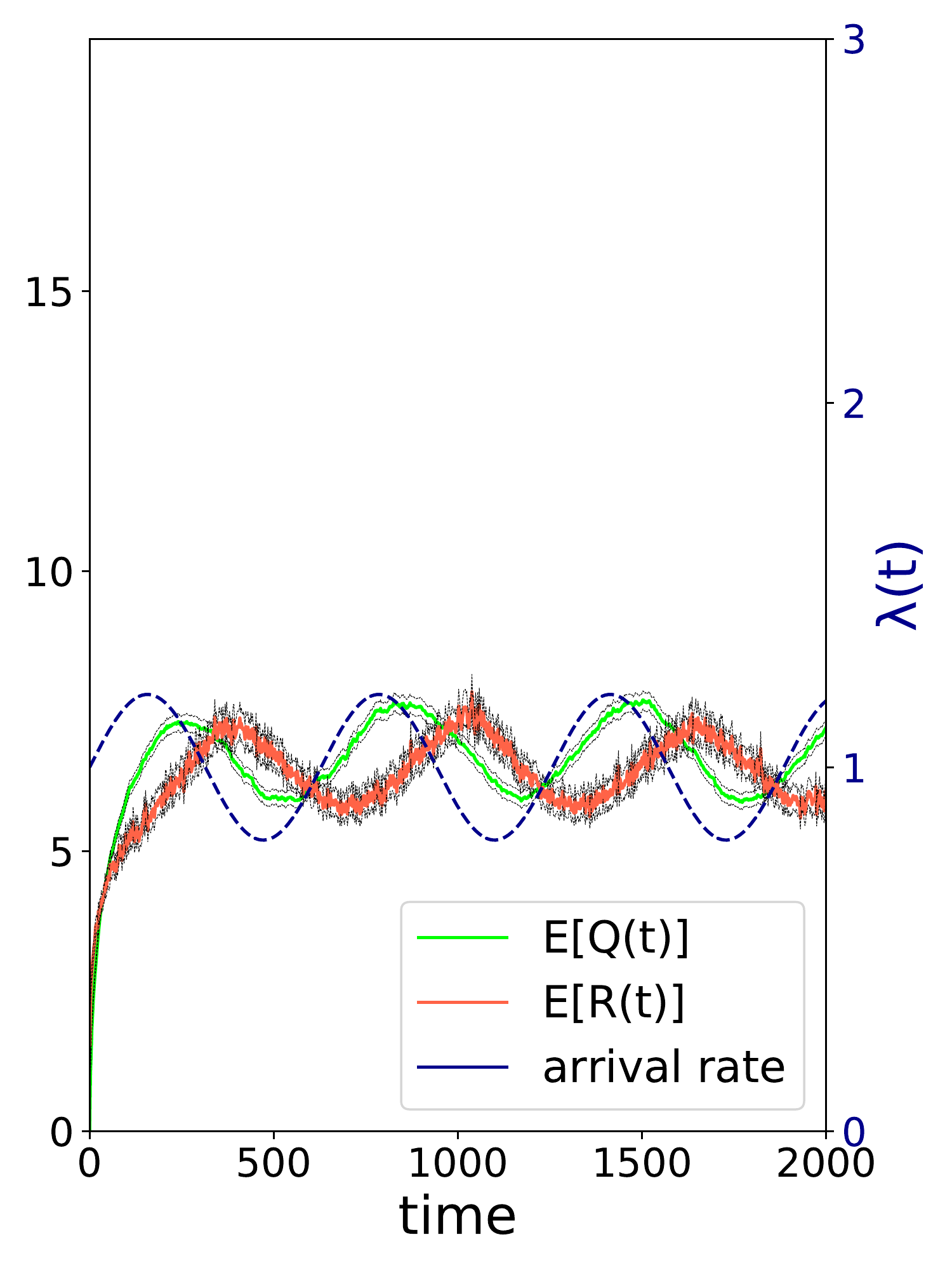}}
		\label{fig:s_10_sr_lnln_001}
	}
	~
	\subfloat[][LN/LN, $\gamma=0.001$]
	{
		\centering\resizebox{0.27\textwidth}{!}{\includegraphics{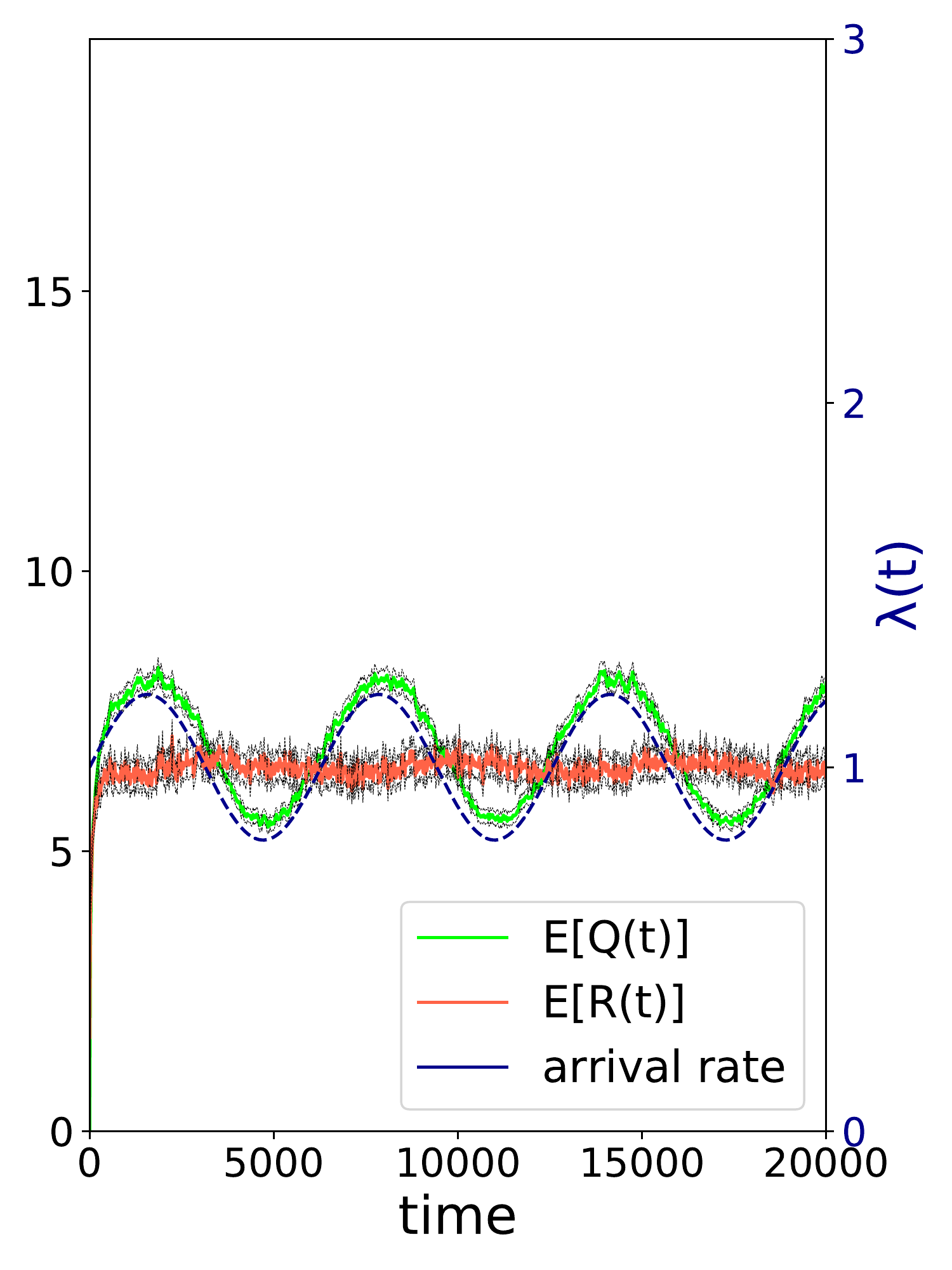}}
		\label{fig:s_10_sr_lnln_0001}
	}
	
	\caption{General performance measures of \TVGGPS\ queues under the SR control with target response time $10$ (heavy-traffic)}
	\label{fig:s_10_sr}
\end{figure}

\begin{figure}
	\centering
	\subfloat[][EXP/EXP, $\gamma=0.1$]
	{
		\centering\resizebox{0.27\textwidth}{!}{\includegraphics{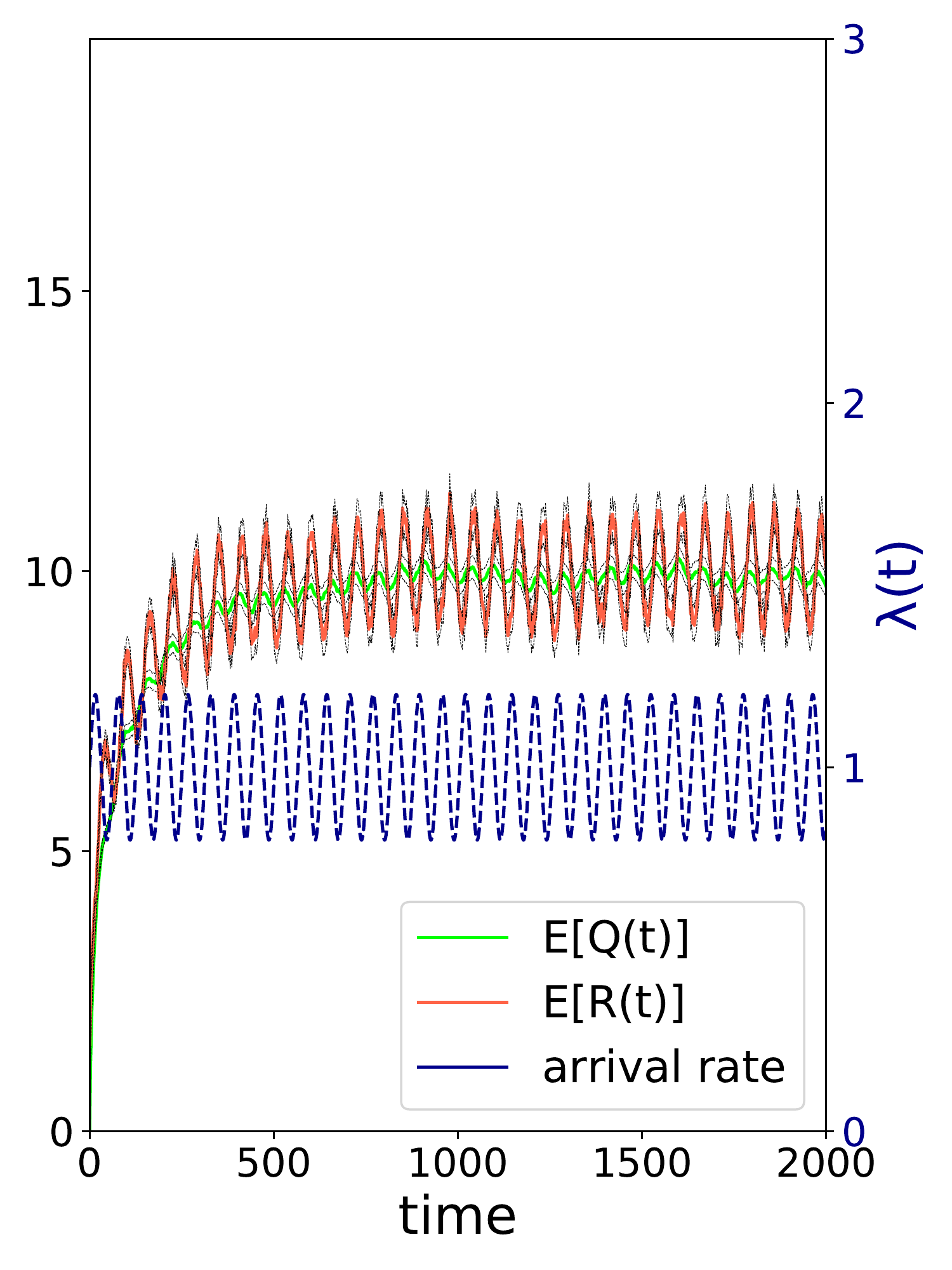}}
		\label{fig:s_10_dm_expexp_01}
	}
	~
	\subfloat[][EXP/EXP, $\gamma=0.01$]
	{
		\centering\resizebox{0.27\textwidth}{!}{\includegraphics{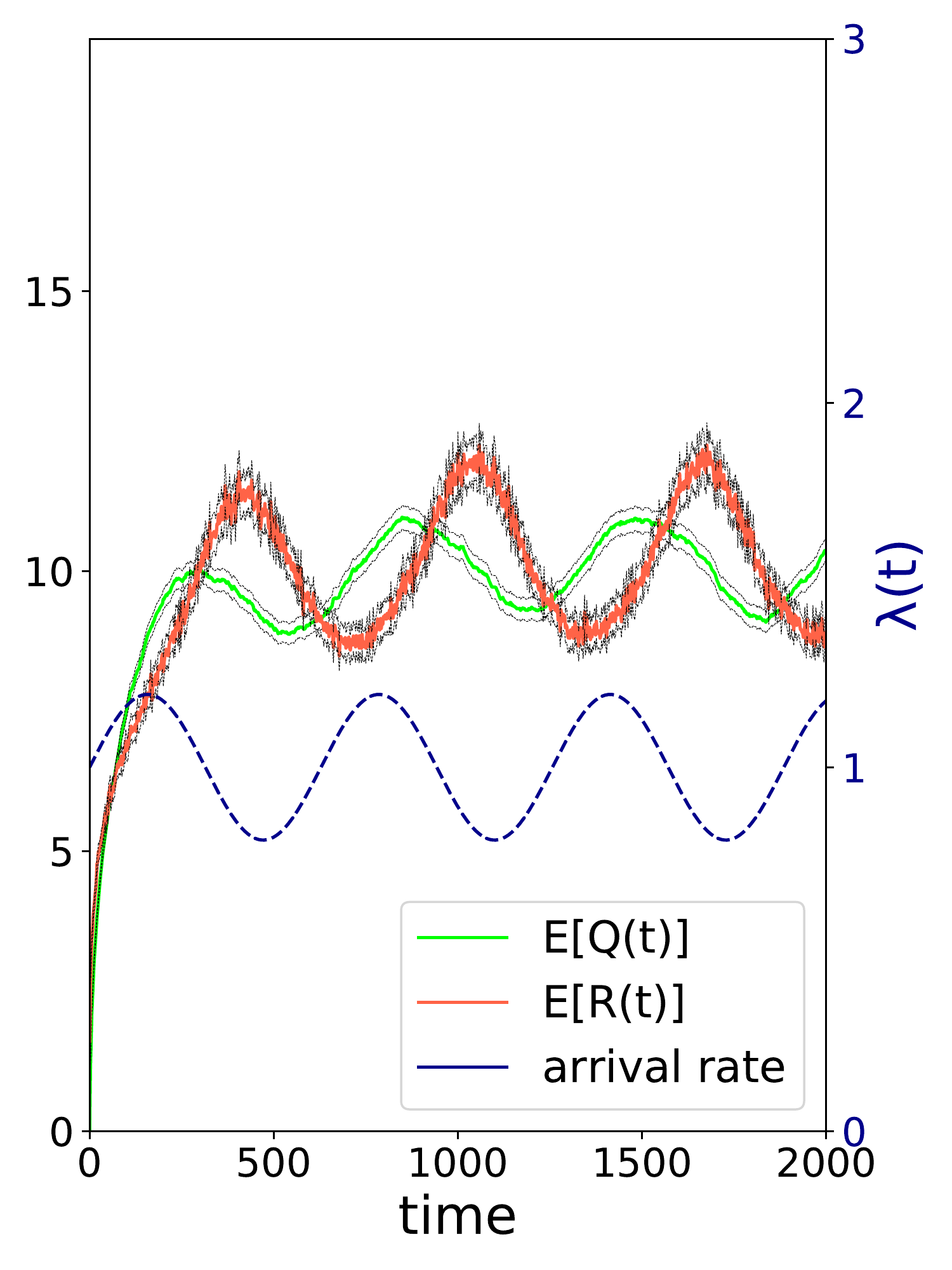}}
		\label{fig:s_10_dm_expexp_001}
	}
	~
	\subfloat[][EXP/EXP, $\gamma=0.001$]
	{
		\centering\resizebox{0.27\textwidth}{!}{\includegraphics{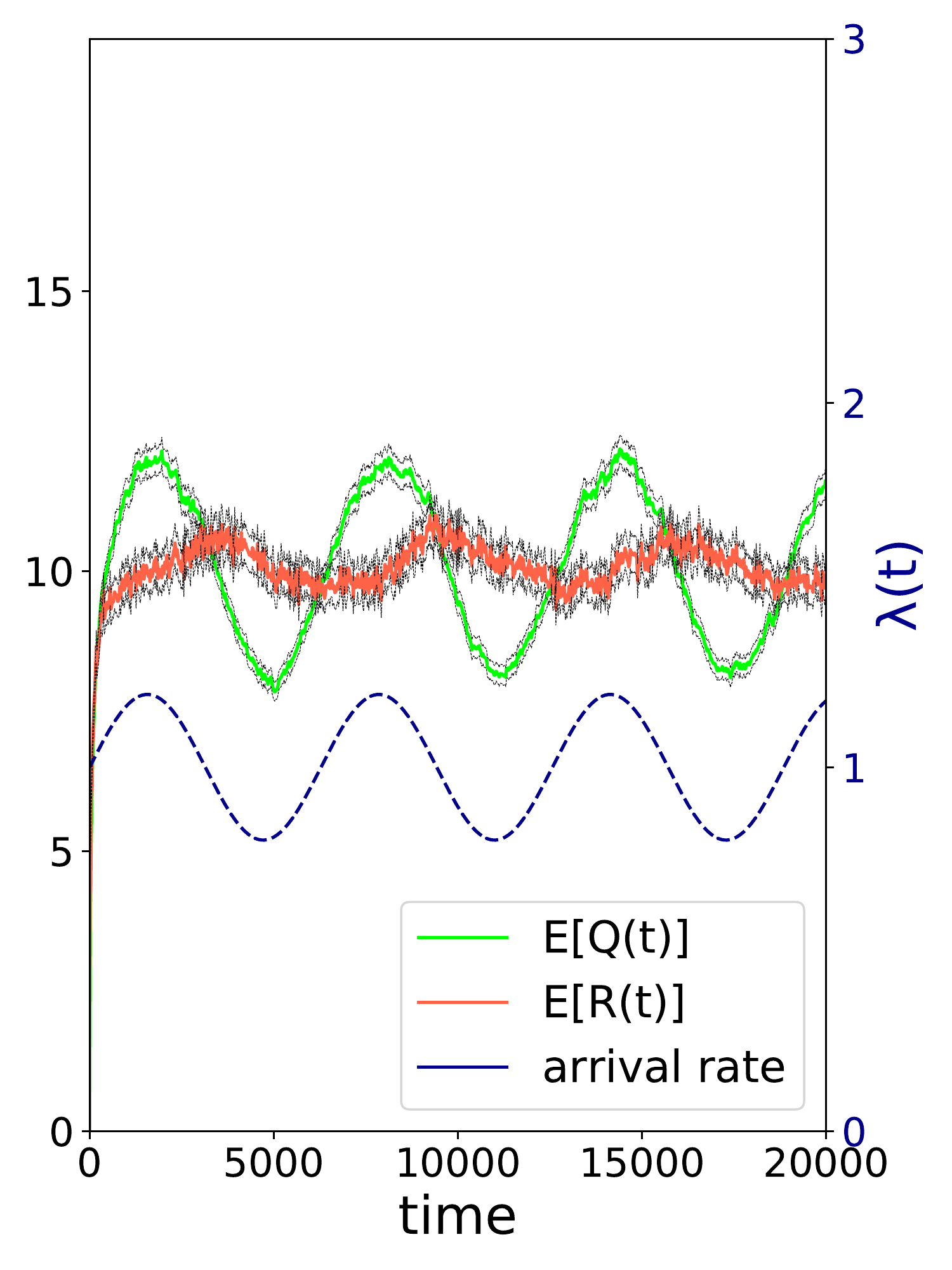}}
		\label{fig:s_10_dm_expexp_0001}
	}
	
	\subfloat[][ER/ER, $\gamma=0.1$]
	{
		\centering\resizebox{0.27\textwidth}{!}{\includegraphics{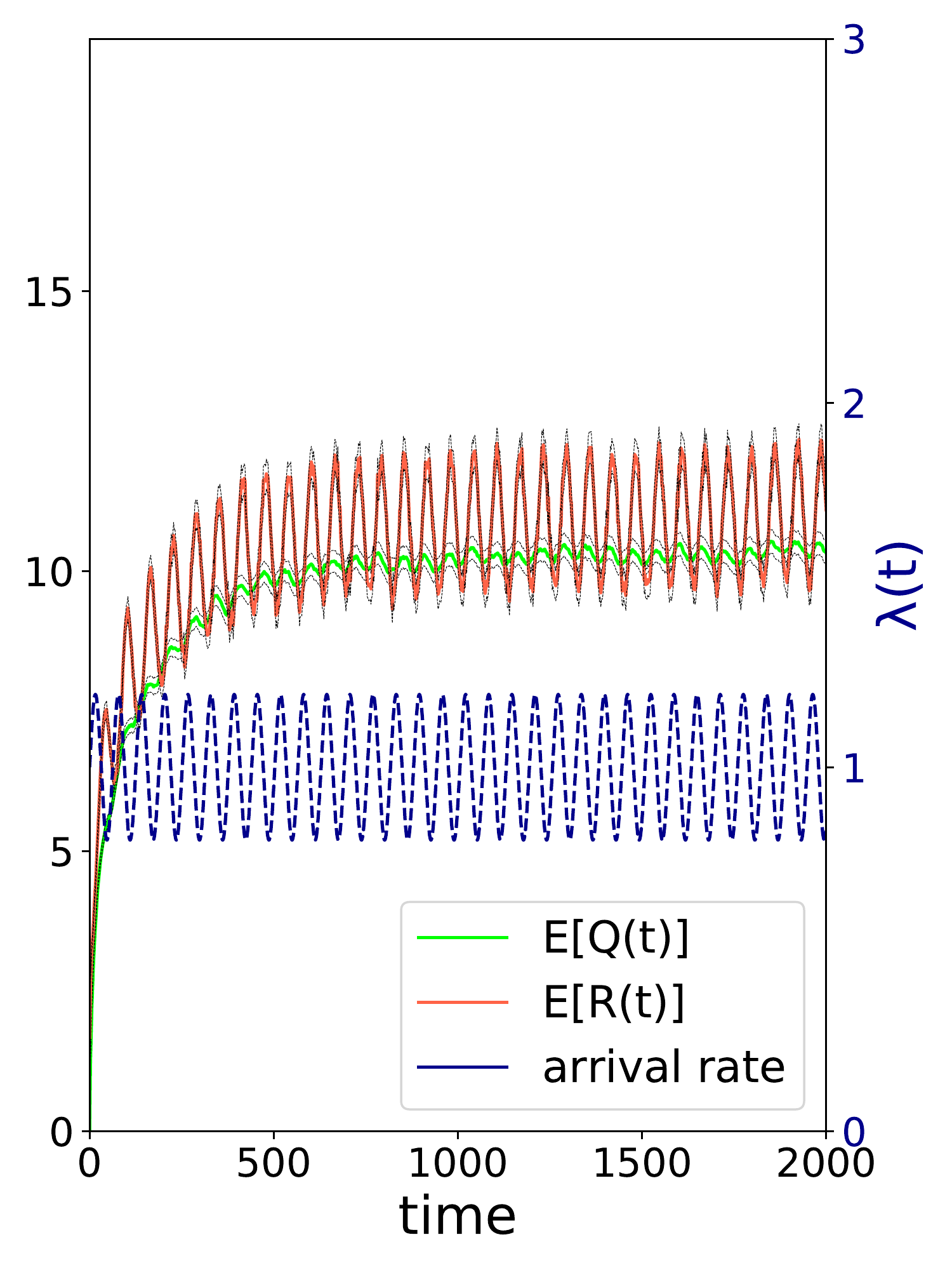}}
		\label{fig:s_10_dm_erer_01}
	}
	~
	\subfloat[][ER/ER, $\gamma=0.01$]
	{
		\centering\resizebox{0.27\textwidth}{!}{\includegraphics{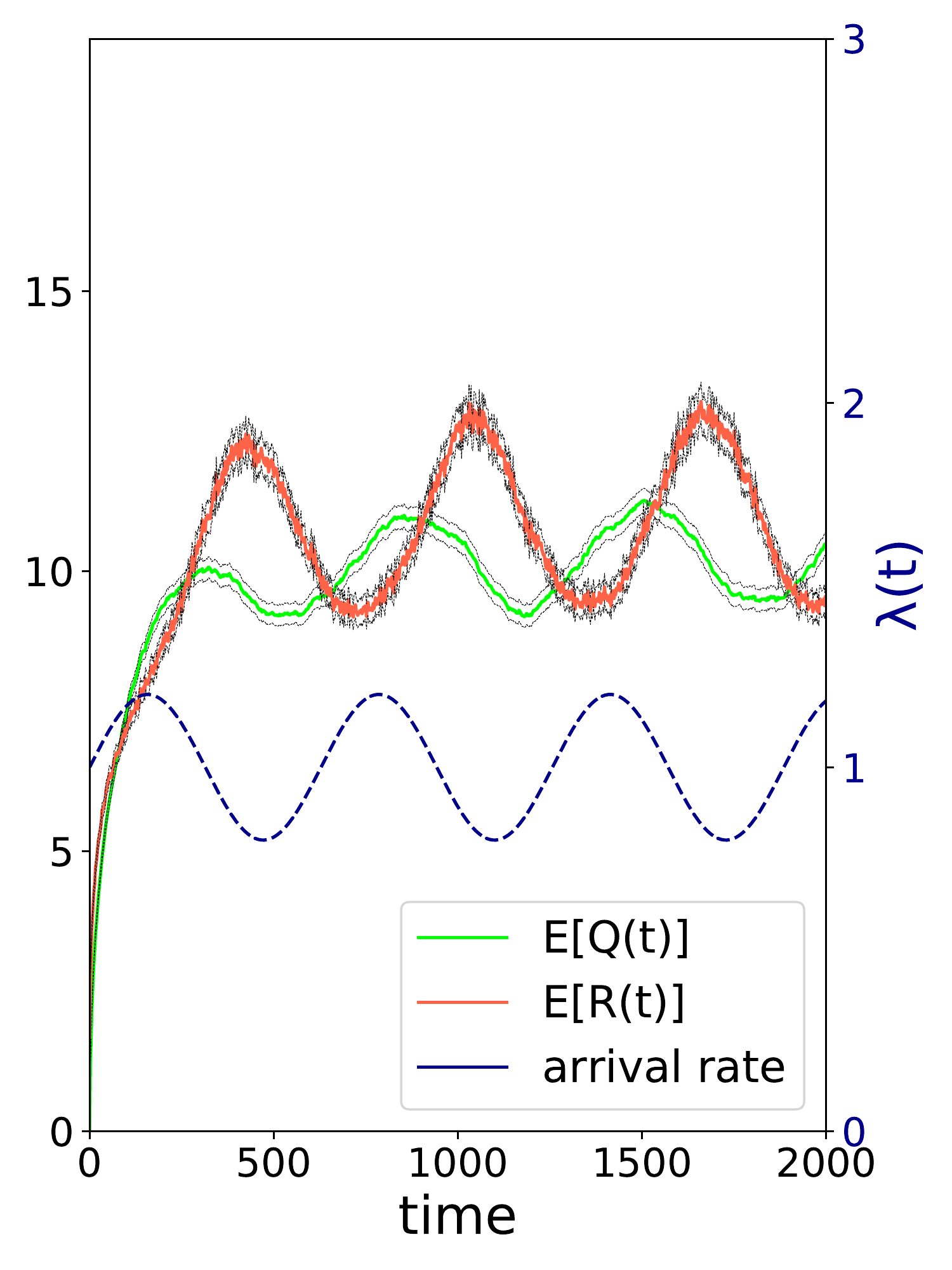}}
		\label{fig:s_10_dm_erer_001}
	}
	~
	\subfloat[][ER/ER, $\gamma=0.001$]
	{
		\centering\resizebox{0.27\textwidth}{!}{\includegraphics{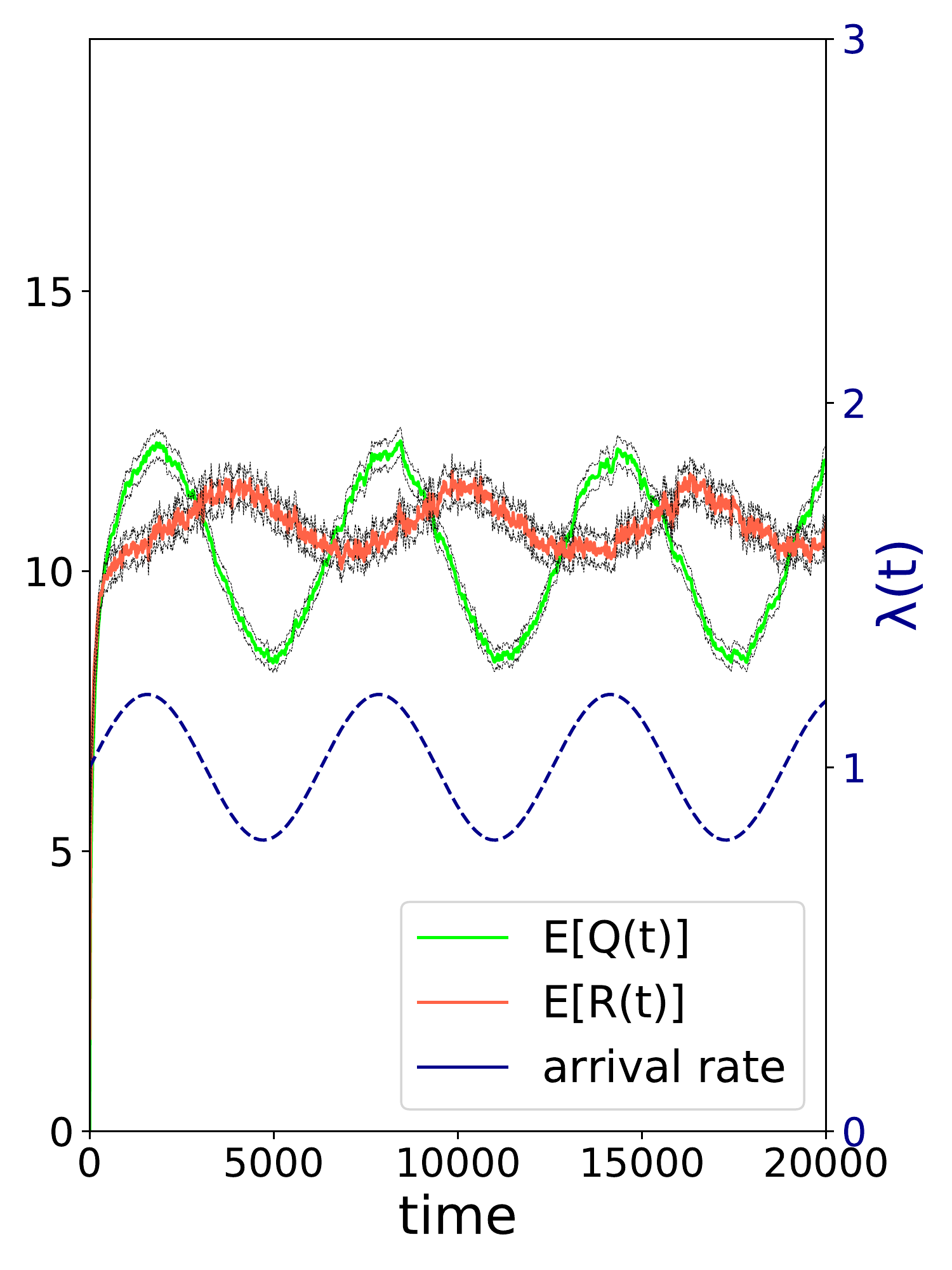}}
		\label{fig:s_10_dm_erer_0001}
	}
	
	\subfloat[][LN/LN, $\gamma=0.1$]
	{
		\centering\resizebox{0.27\textwidth}{!}{\includegraphics{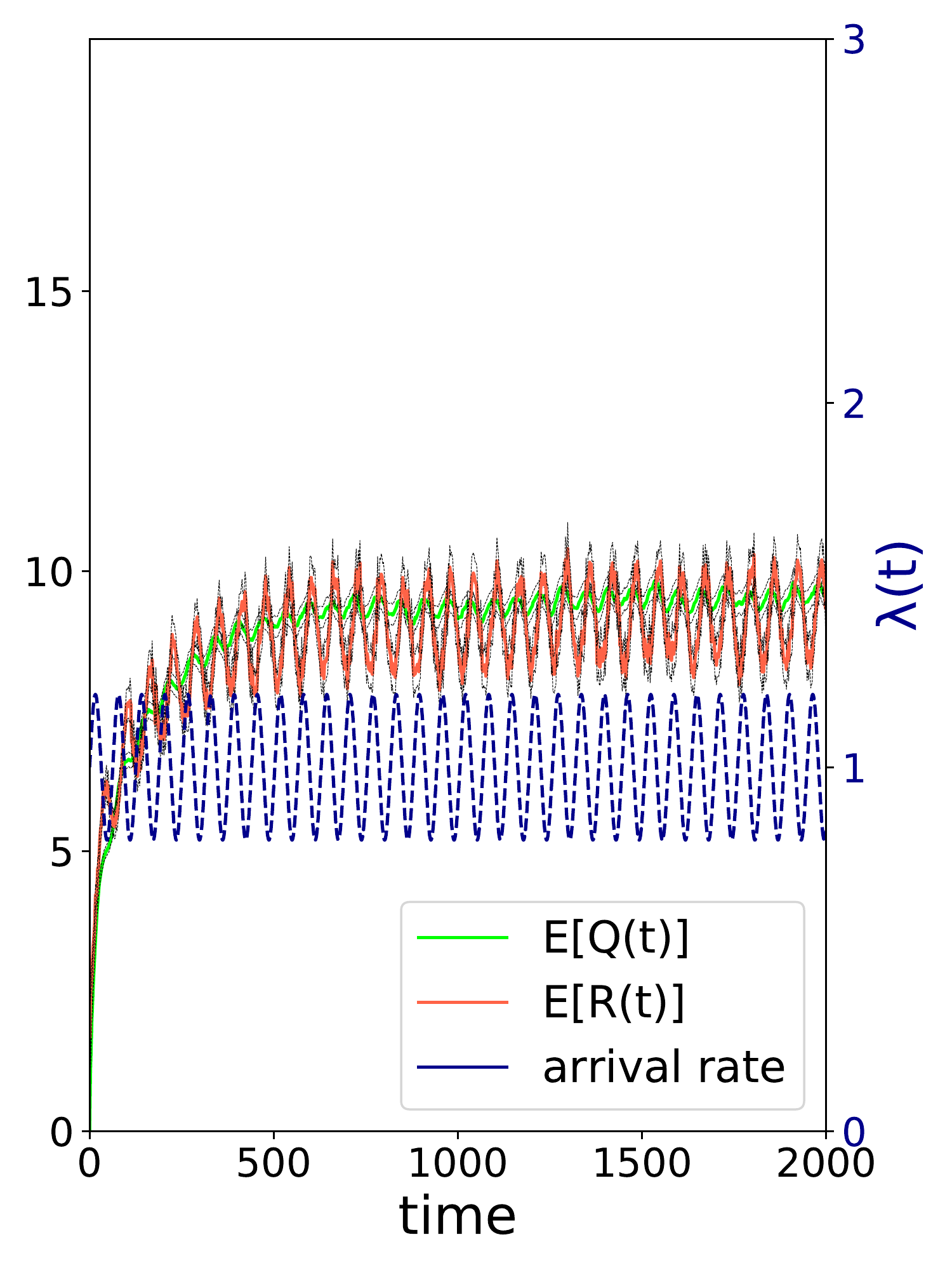}}
		\label{fig:s_10_dm_lnln_01}
	}
	~
	\subfloat[][LN/LN, $\gamma=0.01$]
	{
		\centering\resizebox{0.27\textwidth}{!}{\includegraphics{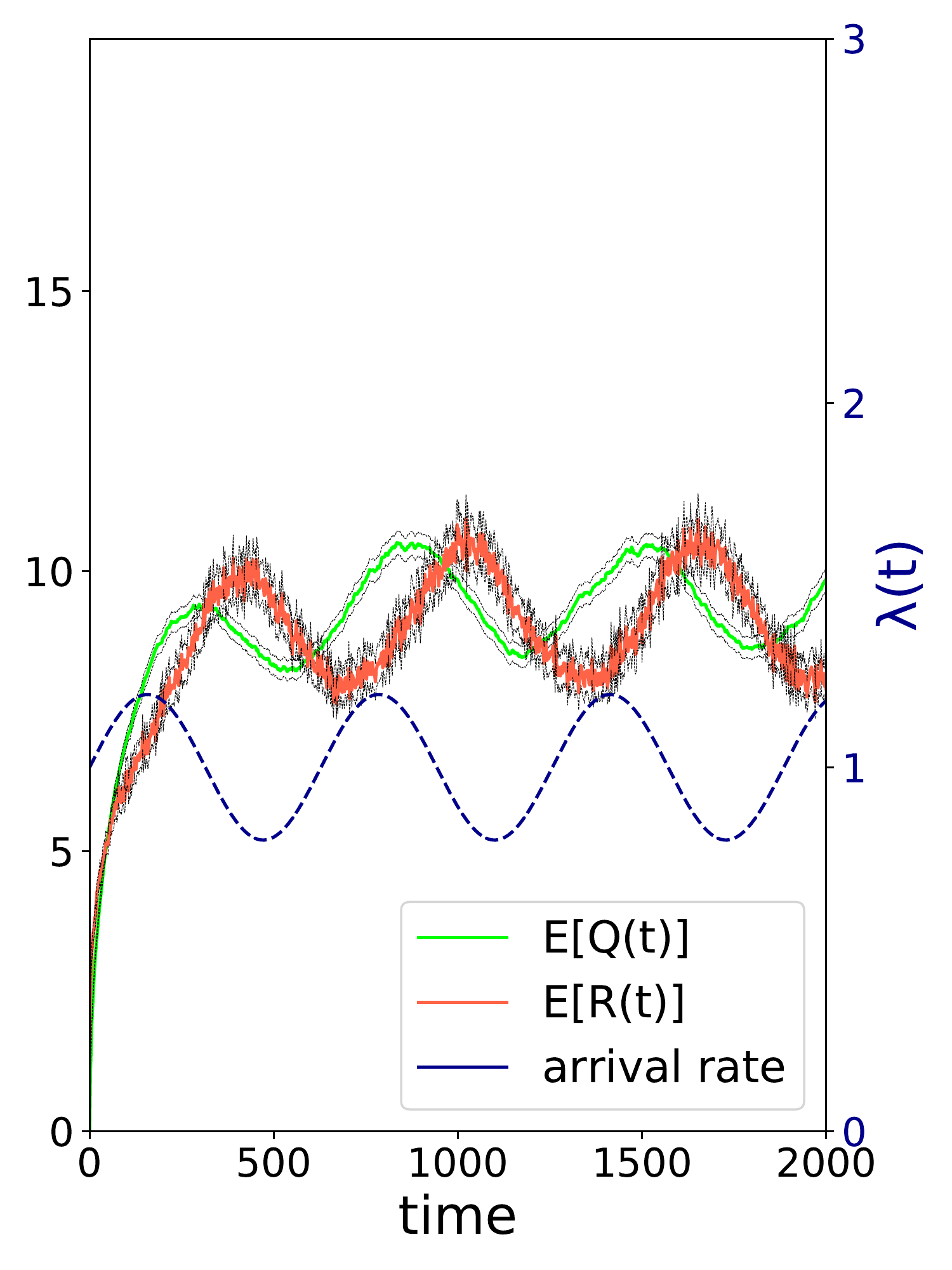}}
		\label{fig:s_10_dm_lnln_001}
	}
	~
	\subfloat[][LN/LN, $\gamma=0.001$]
	{
		\centering\resizebox{0.27\textwidth}{!}{\includegraphics{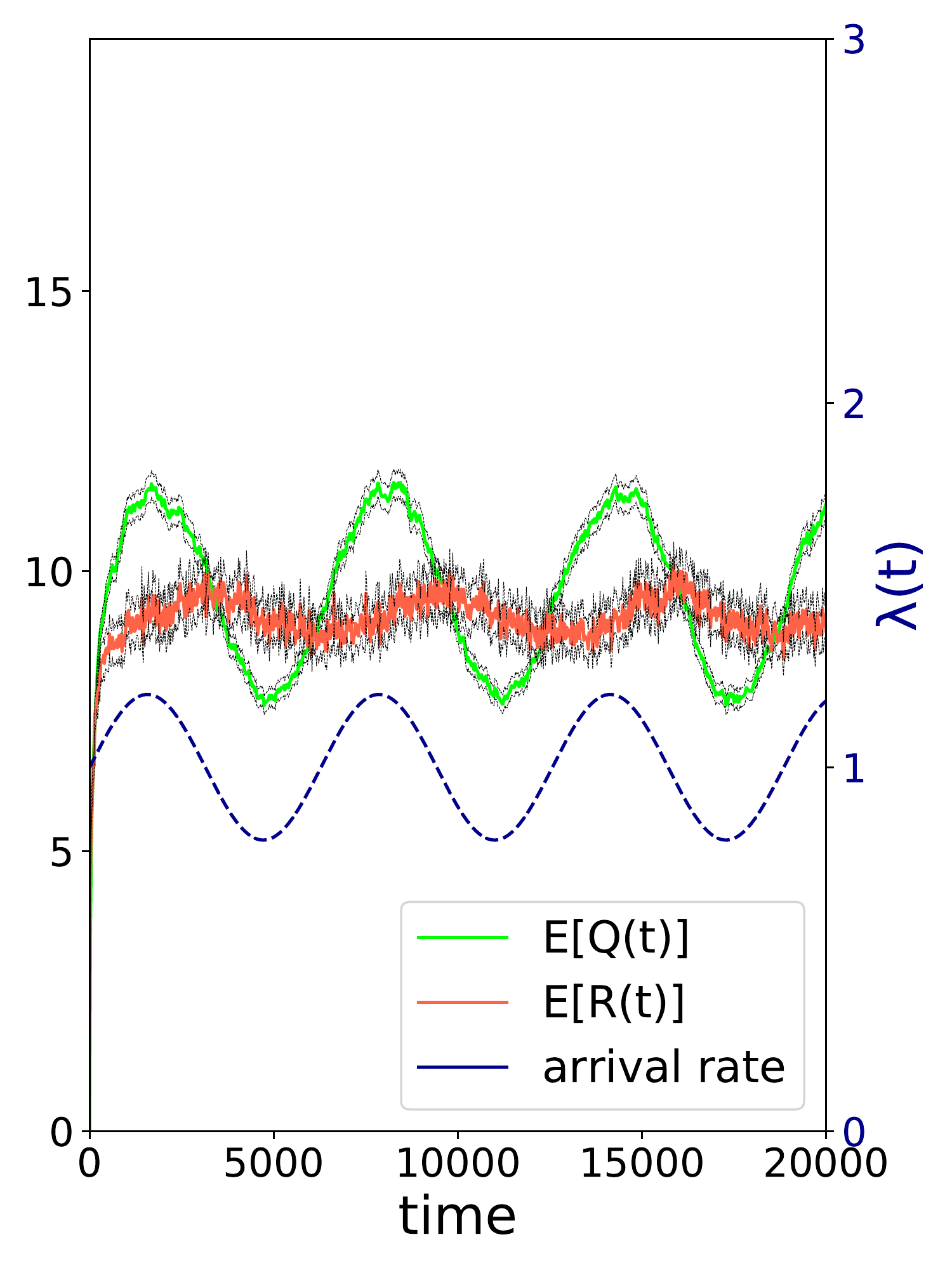}}
		\label{fig:s_10_dm_lnln_0001}
	}
	
	\caption{General performance measures of \TVGGPS\ queues under the DM control with target response time $10$ (heavy-traffic)}
	%	\begin{minipage}
	%		{0.65\textwidth}{\footnotesize *\suc\ instance is too huge and extremely sparse to plot more than 1 scenario}
	%	\end{minipage}
	\label{fig:s_10_dm}
\end{figure}

\subsubsection{Why does DM fail to meet the target response time in light-traffic?}\label{subsec:interpretation_lt}
%Since DM control is derived based on the heavy-traffic result of \GGPS\ queue, its poor performance in light-traffic might be somewhat obvious. We provide more detailed interpretation on this phenomenon.
In light-traffic systems, the probability that two or more jobs will present simultaneously becomes smaller, i.e., it is rare that multiple jobs will share the same processor. We recall the following heavy-traffic based PSAs of the expected virtual response time processes (Equations (\ref{htpsa_fcfs}) and (\ref{htpsa_ps}) in Section \ref{chap2:method}).
\begin{align}
\EE[R_{FCFS}(t)]&\approx\frac{\beta}{\mu(t)}+\frac{\beta}{\mu(t)}\cdot\frac{\rho(t)}{1-\rho(t)}\cdot V_{FCFS}, \tag{\ref{htpsa_fcfs}}\\
\EE[R_{PS}(t)]&\approx\frac{\beta}{\mu(t)}\cdot\frac{1}{1-\rho(t)}\cdot V_{PS}. \tag{\ref{htpsa_ps}}
\end{align}
Letting $\rho(t)\to 0$, the two approximations above converge, respectively, to
\begin{align}
\EE[R_{FCFS}(t)]&\approx\frac{\beta}{\mu(t)},\label{FCFS_0}\\
\EE[R_{PS}(t)]&\approx\frac{\beta}{\mu(t)}\cdot V_{PS}, \label{PS_0}
\end{align}
and then the two controls reduce to the constants:
\begin{align}
\mu_{SR}(t;s)&\equiv\frac{\beta}{s}, \label{SR_0} \\
\mu_{DM}(t;s)&\equiv\frac{\beta V_{PS}}{s}. \label{DM_0}
\end{align}
Throughout the simulation experiments, we use the base distributions having mean $\beta=1$ and the target mean virtual response time $s=0.1$ for light-traffic so $\mu_{SR}(t)=10$ regardless of the distributions. In comparison, $\mu_{DM}(t)$ varies depending on both the base arrival and job size distributions because of variability factor $V_{PS}$. 

In a \TVGGPS\ queue with a service rate function $\mu(\cdot)$, the response time of a job with random size $S$ and arrival time $t$ denoted by $R(t;S,\mu)$, is expressed as
\begin{align}
R(t;S,\mu)=\inf{\left\{y>0: \int_{t}^{y}\frac{\mu(s)}{Q(s)}\textrm{d}s\ge S\right\}}-t.
\end{align}
Approximating $Q(t)\approx 1$ under the light-traffic condition and letting $\mu(\cdot)\leftarrow\mu_{DM}(\cdot)$, the above expression reduces to the analytic form:
\begin{align}
R(t;S,\mu_{DM})=\frac{S}{\mu_{DM}(t)}.
\end{align}
Replacing $S$ by the mean job size $\beta=1$ obtains the numerical values shown in Table \ref{table:number}. We observe that the simulation results and the approximately calculated values coincide, e.g., $\EE[R(t)]$ in Figure \ref{fig:s_01_dm_erer_0001} and $R(t;\beta,\mu_{DM})$ for Erlang/Erlang in Table \ref{table:number} are 0.15. In comparison, we observe that $R(t;S,\mu_{SR})$ is consistently 0.1 regardless of the distributions based on the reasoning we use to obtain the numerical values in Table \ref{table:number}.

We gain two insights into light-traffic systems. First, the PS queue exhibits behavior similar to the FCFS queue. Second, the two service rate controls do not require time dependency. Thus, we conclude that the SR control is appropriate for stabilizing response times in \TVGGPS\ queues as the target response time shortens.

\begin{table}
	\centering
	\caption{Approximately calculated expected response times ($s=0.1,\beta=1$)}
	\label{table:number}
	\begin{tabular}{@{}|c|c|c|@{}}
		\hline
		Distribution pair & $\mu_{DM}(t)$& $R(t;\beta,\mu_{DM})$ \\ \hline
		Exponential/Exponential          & 10            & 0.1                             \\
		Erlang/Erlang            & 6.667         & 0.15                            \\
		Lognormal/Lognormal            & 13.333        & 0.07                            \\
		Erlang/Lognormal            & 8.333         & 0.12                            \\
		Lognormal/Erlang            & 16.666        & 0.06                            \\ \hline
	\end{tabular}
\end{table}

\subsubsection{Heavy-traffic behavior of the two controls}\label{subsec:interpretation_ht}
Figure \ref{fig:comparison} plots the the result of two controls under the different distribution pairs. As we calculated in Section \ref{subsec:interpretation_lt}, the two controls are significantly different when the target response time is short where the target response time is around zero. The difference between them diminishes as the target response time becomes longer. However, the difference in convergence speeds causes the two controls to perform differently in non-asymptotic heavy-traffic systems, e.g., traffic intensities are around 0.9 throughout the heavy-traffic systems. 

\begin{figure}
	\centering
	\subfloat[][Erlang/Erlang]
	{
		\centering\includegraphics[width=0.45\linewidth]{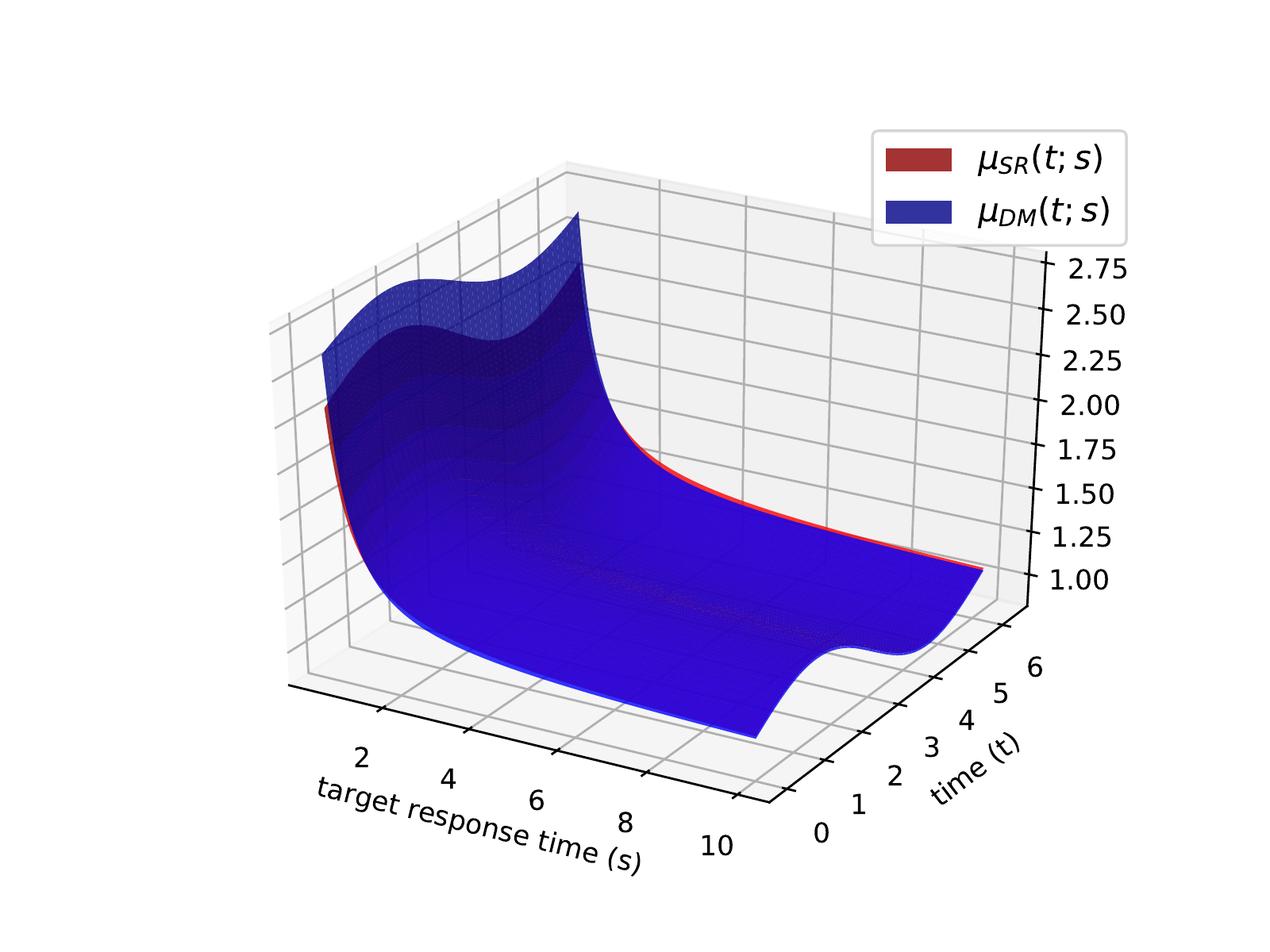}
		\label{fig:comparison_erer}
	}
	~
	\subfloat[][Lognormal/Lognormal]
	{
		\centering\includegraphics[width=0.45\linewidth]{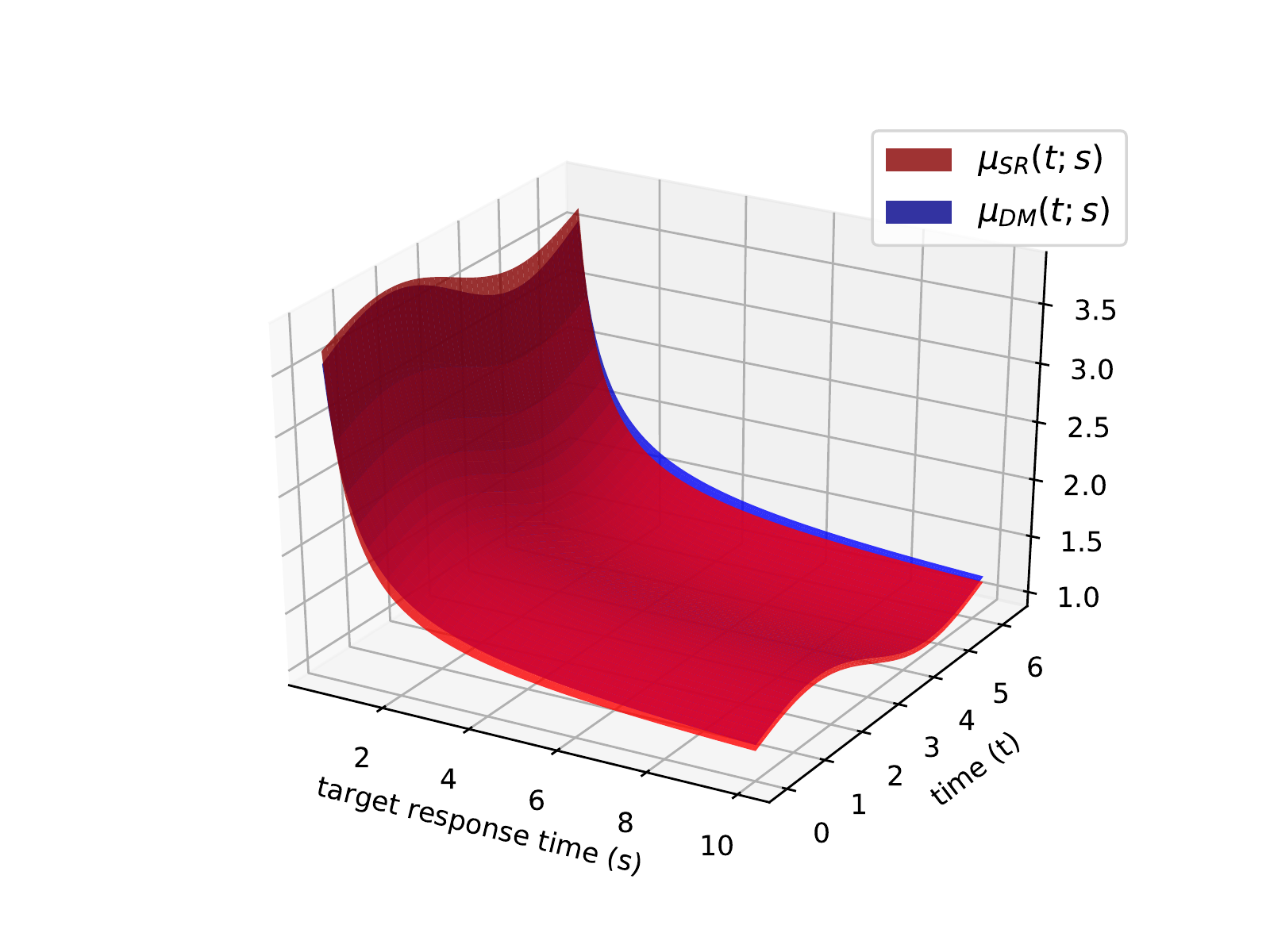}
		\label{fig:comparison_lnln}
	}
	
	\subfloat[][Lognormal/Erlang]
	{
		\centering\includegraphics[width=0.45\linewidth]{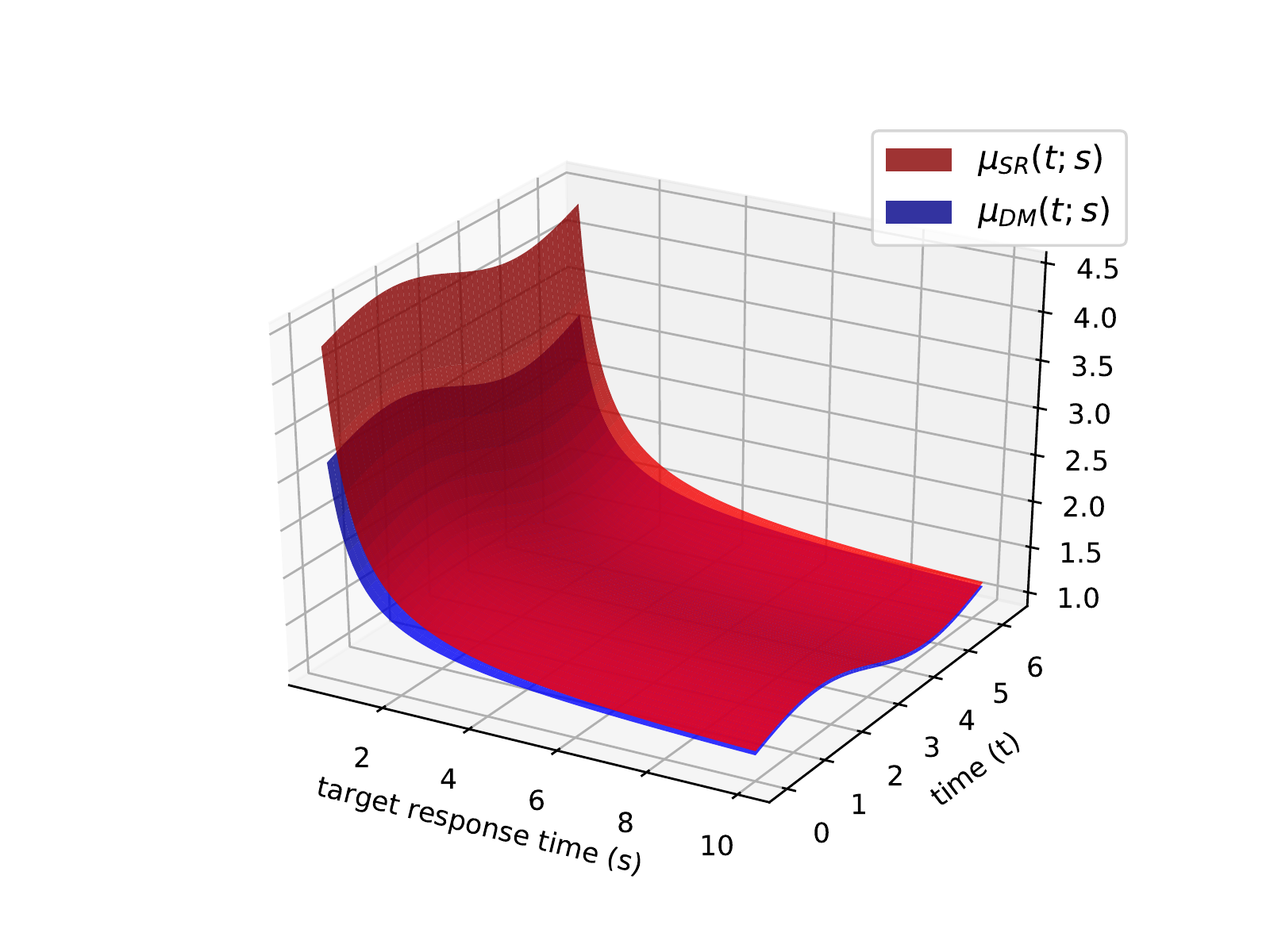}
		\label{fig:comparison_lner}
	}
	~
	\subfloat[][Erlang/Lognormal]
	{
		\centering\includegraphics[width=0.45\linewidth]{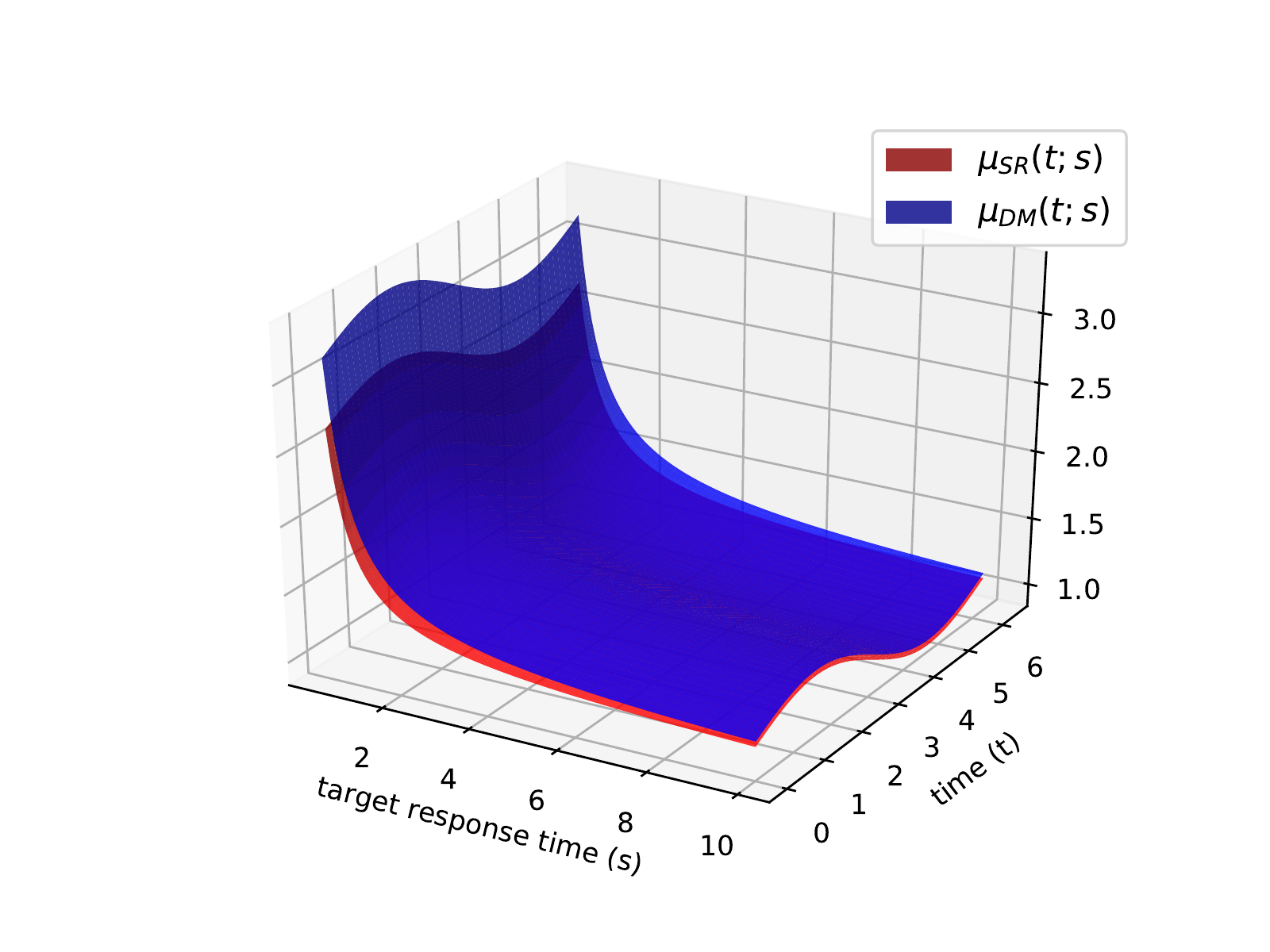}
		\label{fig:comparison_erln}
	}
	
	\caption{Comparison of the two controls}
	\label{fig:comparison}
\end{figure}

\section{Conclusion}\label{chap2:conclusion}
This paper studied the service rate functions that control the mean virtual response time required to obtain stabilization in \TVGGPS\ queues with slowly time-varying arrival rates. Modifying \citet{journal:W2015}'s method for analyzing PS queues resulted in a modified square-root (SR) service rate control and we introduced a new difference-matching (DM) service rate control that appears practically advantageous due to its ease of use and simplicity. Extensive simulation experiments were performed to investigate the performance of two controls. The SR control was effective under a light-traffic condition with a short target response time relative to the inter-arrival times. Neither control, however, perfectly stabilized the response time under a heavy-traffic condition. The DM control outperformed the SR control in terms of meeting the target mean virtual response time.

We suggest several research directions based on the results presented in this paper. Limit theorems, e.g., fluid and diffusion limits, can be derived for \TVGGPS\ queues with periodically time-varying arrival rate functions. We believe that such supporting theories should provide important clues to achieving perfect stabilization of the response time process. Light-traffic behaviors in queueing situations also deserve more analysis, since studies of time-varying queues are scarse to the best of our knowledge. Conceivably, interpolating the two controls could extend the coverage of the target response time beyond short and long. Of course, practical applications in ICT infrastructures should be accompanied.

	\section*{Acknowledgment}
	This research was supported by the Basic Science Research Program through the National Research Foundation of Korea (NRF) funded by the Ministry of Education (NRF-2016R1D1A1B04933453).



%%       \citet{<label>} ==> Jones et al. [21]
%%       \citep{<label>} ==> [21]
%%


\section*{References}
\bibliographystyle{elsarticle-num-names}
\bibliography{manuscript}

\appendix
%\section*{Appendix}
\section{Performances of the suggested controls}\label{app:performance}
\subsection{Numerical data}
\begin{table}[H]
	\centering
	\caption{Control performance of $\mu_{SR}$ and $\mu_{DM}$: $M_t/M_t/1/PS$}
	\label{table:tvgg1ps_expexp}
	\resizebox{\textwidth}{!}{%
	\begin{tabular}{cc|rr|rr|}
		\cline{3-6}
		\multicolumn{1}{l}{}                        & \multicolumn{1}{l|}{} & \multicolumn{2}{c|}{$\mu_{SR}$}                                                & \multicolumn{2}{c|}{$\mu_{DM}$}                                                \\ \hline
		\multicolumn{1}{|c|}{$s$}                   & $\gamma$              & \multicolumn{1}{c}{Amplitude (RA)} & \multicolumn{1}{c|}{Spatial Average (TG)} & \multicolumn{1}{c}{Amplitude (RA)} & \multicolumn{1}{c|}{Spatial Average (TG)} \\ \hline
		\multicolumn{1}{|c|}{\multirow{3}{*}{0.1}}  & 0.001                 & 0.0044 (4.37\%)                    & 0.1 (0.0\%)                              & 0.0037 (3.7\%)                     & 0.1001 (0.0\%)                            \\
		\multicolumn{1}{|c|}{}                      & 0.01                  & 0.0035 (3.45\%)                    & 0.1001 (0.0\%)                            & 0.0039 (3.9\%)                     & 0.1001 (0.0\%)                            \\
		\multicolumn{1}{|c|}{}                      & 0.1                   & 0.003 (2.95\%)                     & 0.1015 (0.02\%)                           & 0.0025 (2.5\%)                     & 0.1012 (0.01\%)                           \\ \hline
		\multicolumn{1}{|c|}{\multirow{3}{*}{10.0}} & 0.001                 & 0.7555 (7.44\%)                    & 10.1493 (0.01\%)                          & 0.6808 (6.72\%)                    & 10.1321 (0.01\%)                          \\
		\multicolumn{1}{|c|}{}                      & 0.01                  & 1.7516 (17.07\%)                   & 10.264 (0.03\%)                           & 1.8011 (17.58\%)                   & 10.2452 (0.02\%)                          \\
		\multicolumn{1}{|c|}{}                      & 0.1                   & 1.1104 (10.82\%)                   & 10.2665 (0.03\%)                          & 1.0656 (10.57\%)                   & 10.0773 (0.01\%)                          \\ \hline
	\end{tabular}%
	}
\end{table}
\begin{table}[H]
	\centering
	\caption{Control performance of $\mu_{SR}$ and $\mu_{DM}$: $ER_t/ER_t/1/PS$}
	\label{table:tvgg1ps_erer}
	\resizebox{\textwidth}{!}{%
	\begin{tabular}{cc|rr|rr|}
		\cline{3-6}
		\multicolumn{1}{l}{}                        & \multicolumn{1}{l|}{} & \multicolumn{2}{c|}{$\mu_{SR}$}                                                & \multicolumn{2}{c|}{$\mu_{DM}$}                                                \\ \hline
		\multicolumn{1}{|c|}{$s$}                   & $\gamma$              & \multicolumn{1}{c}{Amplitude (RA)} & \multicolumn{1}{c|}{Spatial Average (TG)} & \multicolumn{1}{c}{Amplitude (RA)} & \multicolumn{1}{c|}{Spatial Average (TG)} \\ \hline
		\multicolumn{1}{|c|}{\multirow{3}{*}{0.1}}  & 0.001                 & 0.0037 (3.52\%)                    & 0.1047 (0.05\%)                           & 0.005 (3.37\%)                     & 0.1487 (0.49\%)                           \\
		\multicolumn{1}{|c|}{}                      & 0.01                  & 0.0039 (3.72\%)                    & 0.1045 (0.04\%)                           & 0.0044 (2.93\%)                    & 0.1485 (0.49\%)                           \\
		\multicolumn{1}{|c|}{}                      & 0.1                   & 0.0026 (2.46\%)                    & 0.1064 (0.06\%)                           & 0.0027 (1.78\%)                    & 0.1513 (0.51\%)                           \\ \hline
		\multicolumn{1}{|c|}{\multirow{3}{*}{10.0}} & 0.001                 & 1.1872 (8.64\%)                    & 13.7467 (0.37\%)                          & 0.7828 (7.20\%)                     & 10.8678 (0.09\%)                          \\
		\multicolumn{1}{|c|}{}                      & 0.01                  & 2.5553 (18.56\%)                   & 13.765 (0.38\%)                           & 1.8884 (17.19\%)                   & 10.9824 (0.10\%)                           \\
		\multicolumn{1}{|c|}{}                      & 0.1                   & 1.4838 (10.74\%)                   & 13.822 (0.38\%)                           & 1.3446 (12.04\%)                   & 11.1719 (0.12\%)                          \\ \hline
	\end{tabular}%
	}
\end{table}

\begin{table}[H]
	\centering
	\caption{Control performance of $\mu_{SR}$ and $\mu_{DM}$: $LN_t/LN_t/1/PS$}
	\label{table:tvgg1ps_lnln}
	\resizebox{\textwidth}{!}{%
	\begin{tabular}{cc|rr|rr|}
		\cline{3-6}
		\multicolumn{1}{l}{}                        & \multicolumn{1}{l|}{} & \multicolumn{2}{c|}{$\mu_{SR}$}                                                & \multicolumn{2}{c|}{$\mu_{DM}$}                                                \\ \hline
		\multicolumn{1}{|c|}{$s$}                   & $\gamma$              & \multicolumn{1}{c}{Amplitude (RA)} & \multicolumn{1}{c|}{Spatial Average (TG)} & \multicolumn{1}{c}{Amplitude (RA)} & \multicolumn{1}{c|}{Spatial Average (TG)} \\ \hline
		\multicolumn{1}{|c|}{\multirow{3}{*}{0.1}}  & 0.001                 & 0.0049 (5.27\%)                    & 0.0922 (-0.08\%)                          & 0.004 (5.28\%)                     & 0.0751 (-0.25\%)                          \\
		\multicolumn{1}{|c|}{}                      & 0.01                  & 0.0057 (6.14\%)                    & 0.0921 (-0.08\%)                          & 0.0044 (5.88\%)                    & 0.075 (-0.25\%)                           \\
		\multicolumn{1}{|c|}{}                      & 0.1                   & 0.0041 (4.38\%)                    & 0.0936 (-0.06\%)                          & 0.0041 (5.42\%)                    & 0.0761 (-0.24\%)                          \\ \hline
		\multicolumn{1}{|c|}{\multirow{3}{*}{10.0}} & 0.001                 & 0.4776 (7.34\%)                    & 6.5104 (-0.35\%)                          & 0.8366 (9.05\%)                    & 9.2406 (-0.08\%)                          \\
		\multicolumn{1}{|c|}{}                      & 0.01                  & 0.9325 (14.33\%)                   & 6.5061 (-0.35\%)                          & 1.6268 (17.64\%)                   & 9.2206 (-0.08\%)                          \\
		\multicolumn{1}{|c|}{}                      & 0.1                   & 0.8908 (13.47\%)                   & 6.6132 (-0.34\%)                          & 0.9665 (10.3\%)                    & 9.3818 (-0.06\%)                          \\ \hline
	\end{tabular}%
	}
\end{table}

\begin{table}[H]
	\centering
	\caption{Control performance of $\mu_{SR}$ and $\mu_{DM}$: $ER_t/LN_t/1/PS$}
	\label{table:tvgg1ps_erln}
	\resizebox{\textwidth}{!}{%
	\begin{tabular}{cc|rr|rr|}
		\cline{3-6}
		\multicolumn{1}{l}{}                        & \multicolumn{1}{l|}{} & \multicolumn{2}{c|}{$\mu_{SR}$}                                                & \multicolumn{2}{c|}{$\mu_{DM}$}                                                \\ \hline
		\multicolumn{1}{|c|}{$s$}                   & $\gamma$              & \multicolumn{1}{c}{Amplitude (RA)} & \multicolumn{1}{c|}{Spatial Average (TG)} & \multicolumn{1}{c}{Amplitude (RA)} & \multicolumn{1}{c|}{Spatial Average (TG)} \\ \hline
		\multicolumn{1}{|c|}{\multirow{3}{*}{0.1}}  & 0.001                 & 0.0049 (5.07\%)                    & 0.0976 (-0.02\%)                          & 0.0055 (4.56\%)                    & 0.1197 (0.2\%)                            \\
		\multicolumn{1}{|c|}{}                      & 0.01                  & 0.0047 (4.81\%)                    & 0.0974 (-0.03\%)                          & 0.005 (4.22\%)                     & 0.1193 (0.19\%)                           \\
		\multicolumn{1}{|c|}{}                      & 0.1                   & 0.0031 (3.08\%)                    & 0.0995 (-0.01\%)                          & 0.0046 (3.8\%)                     & 0.1209 (0.21\%)                           \\ \hline
		\multicolumn{1}{|c|}{\multirow{3}{*}{10.0}} & 0.001                 & 0.5891 (8.26\%)                    & 7.1291 (-0.29\%)                          & 1.0496 (10.06\%)                   & 10.4293 (0.04\%)                          \\
		\multicolumn{1}{|c|}{}                      & 0.01                  & 1.2382 (17.11\%)                   & 7.2366 (-0.28\%)                          & 2.0119 (19.39\%)                   & 10.375 (0.04\%)                           \\
		\multicolumn{1}{|c|}{}                      & 0.1                   & 0.8216 (11.29\%)                   & 7.2763 (-0.27\%)                          & 1.2382 (11.84\%)                   & 10.456 (0.05\%)                           \\ \hline
	\end{tabular}%
	}
\end{table}
\begin{table}[H]
	\centering
	\caption{Control performance of $\mu_{SR}$ and $\mu_{DM}$: $LN_t/ER_t/1/PS$}
	\label{table:tvgg1ps_lner}
	\resizebox{\textwidth}{!}{%
	\begin{tabular}{cc|rr|rr|}
		\cline{3-6}
		\multicolumn{1}{l}{}                        & \multicolumn{1}{l|}{} & \multicolumn{2}{c|}{$\mu_{SR}$}                                                & \multicolumn{2}{c|}{$\mu_{DM}$}                                                \\ \hline
		\multicolumn{1}{|c|}{$s$}                   & $\gamma$              & \multicolumn{1}{c}{Amplitude (RA)} & \multicolumn{1}{c|}{Spatial Average (TG)} & \multicolumn{1}{c}{Amplitude (RA)} & \multicolumn{1}{c|}{Spatial Average (TG)} \\ \hline
		\multicolumn{1}{|c|}{\multirow{3}{*}{0.1}}  & 0.001                 & 0.0034 (3.51\%)                    & 0.0978 (-0.02\%)                          & 0.003 (4.91\%)                     & 0.0601 (-0.4\%)                           \\
		\multicolumn{1}{|c|}{}                      & 0.01                  & 0.0039 (3.99\%)                    & 0.0977 (-0.02\%)                          & 0.003 (4.97\%)                     & 0.0602 (-0.4\%)                           \\
		\multicolumn{1}{|c|}{}                      & 0.1                   & 0.0026 (2.64\%)                    & 0.0997 (-0.0\%)                           & 0.0023 (3.82\%)                    & 0.0609 (-0.39\%)                          \\ \hline
		\multicolumn{1}{|c|}{\multirow{3}{*}{10.0}} & 0.001                 & 0.6518 (5.63\%)                    & 11.5748 (0.16\%)                          & 0.4418 (5.43\%)                    & 8.1343 (-0.19\%)                          \\
		\multicolumn{1}{|c|}{}                      & 0.01                  & 1.8191 (15.4\%)                    & 11.8151 (0.18\%)                          & 1.0204 (12.49\%)                   & 8.1706 (-0.18\%)                          \\
		\multicolumn{1}{|c|}{}                      & 0.1                   & 1.223 (10.41\%)                    & 11.7491 (0.17\%)                          & 0.9637 (11.56\%)                   & 8.3389 (-0.17\%)                          \\ \hline
	\end{tabular}%
	}
\end{table}

\subsection{Plots}
%\subsection{Arrival base/Job size: Exponential/Exponential ($M_t/M_t/1/PS$)}
% s=0.1, M(t)/M(t)/1/PS
\begin{figure}[H]
	\label{fig:tvgg1ps_0.1_expexp}
	\centering\includegraphics[width=\linewidth]{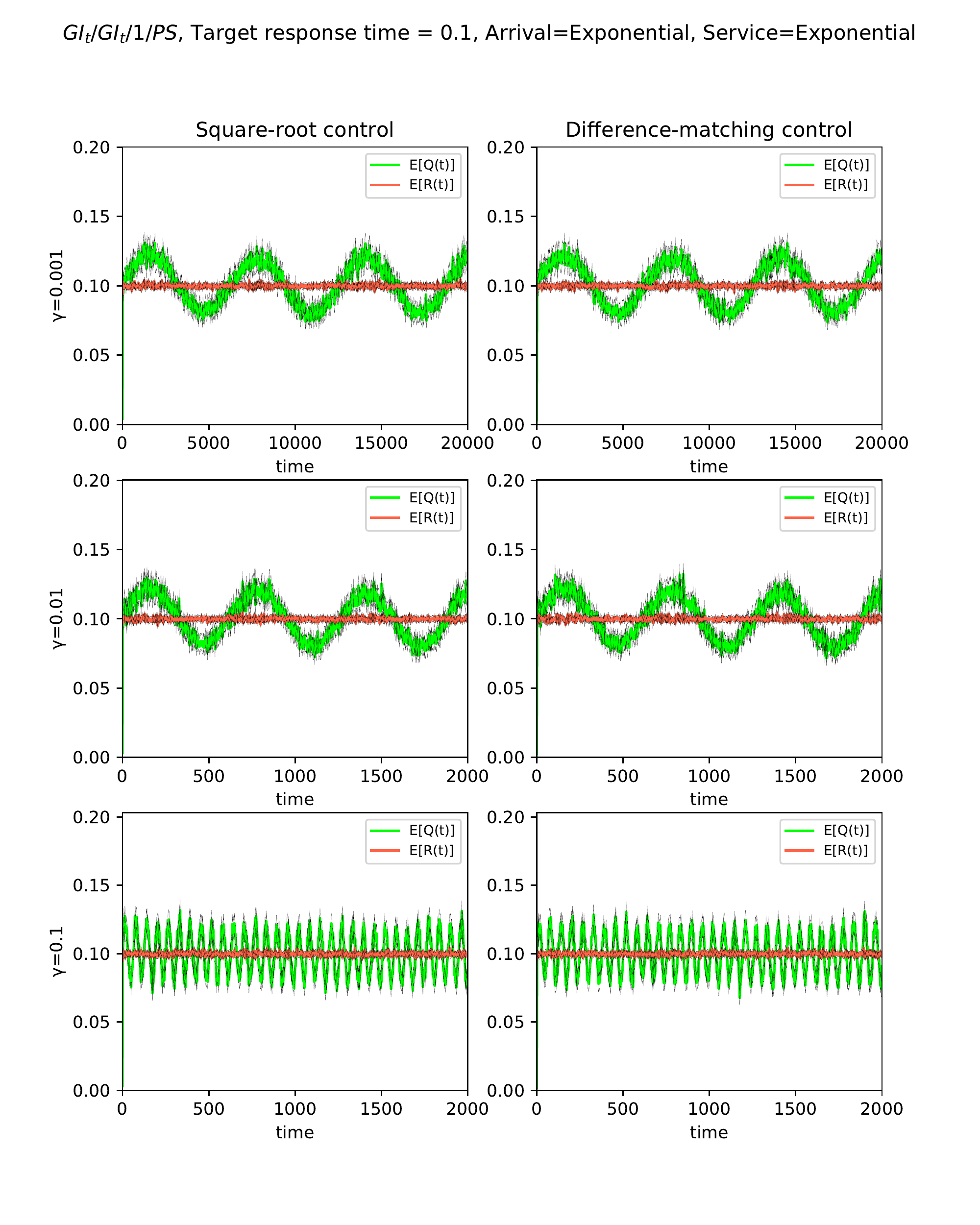}
	\caption{General performance measures of $M_t/M_t/1/PS$ queues under $\mu_{SR}$ and $\mu_{DM}$ with target response time 0.1 ($s=0.1$)}
\end{figure}

% s=10.0, M(t)/M(t)/1/PS
\begin{figure}[H]
	\label{fig:tvgg1ps_10.0_expexp}
	\centering\includegraphics[width=\linewidth]{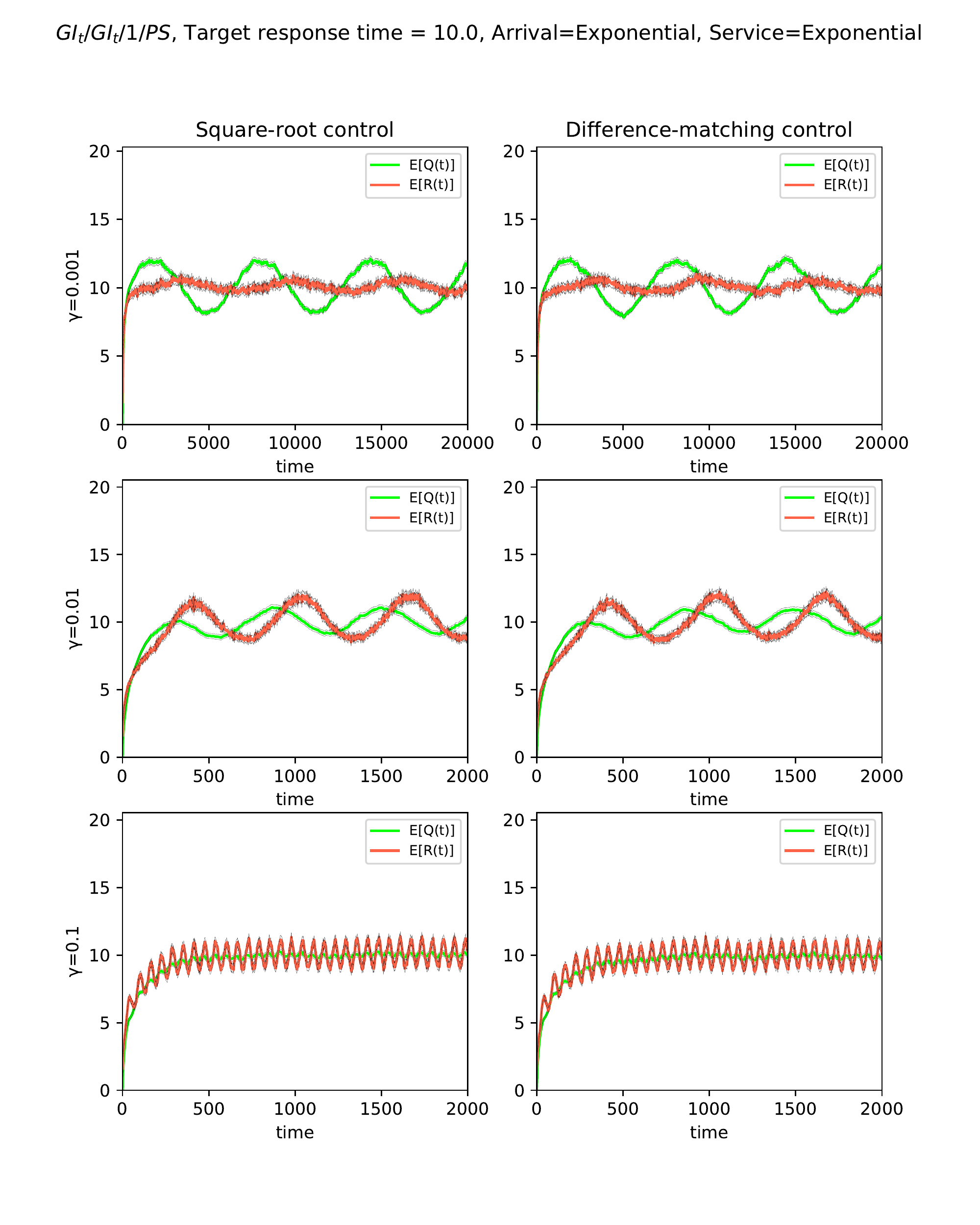}
	\caption{General performance measures of $M_t/M_t/1/PS$ queues under $\mu_{SR}$ and $\mu_{DM}$ with target response time 10.0 ($s=10.0$)}
\end{figure}

% ER(t)/ER(t)/1/PS
%\subsection{Arrival base/Job size: Erlang/Erlang ($ER_t/ER_t/1/PS$)}
% s=0.1, ER(t)/ER(t)/1/PS
\begin{figure}[H]
	\label{fig:tvgg1ps_0.1_erer}
	\centering\includegraphics[width=\linewidth]{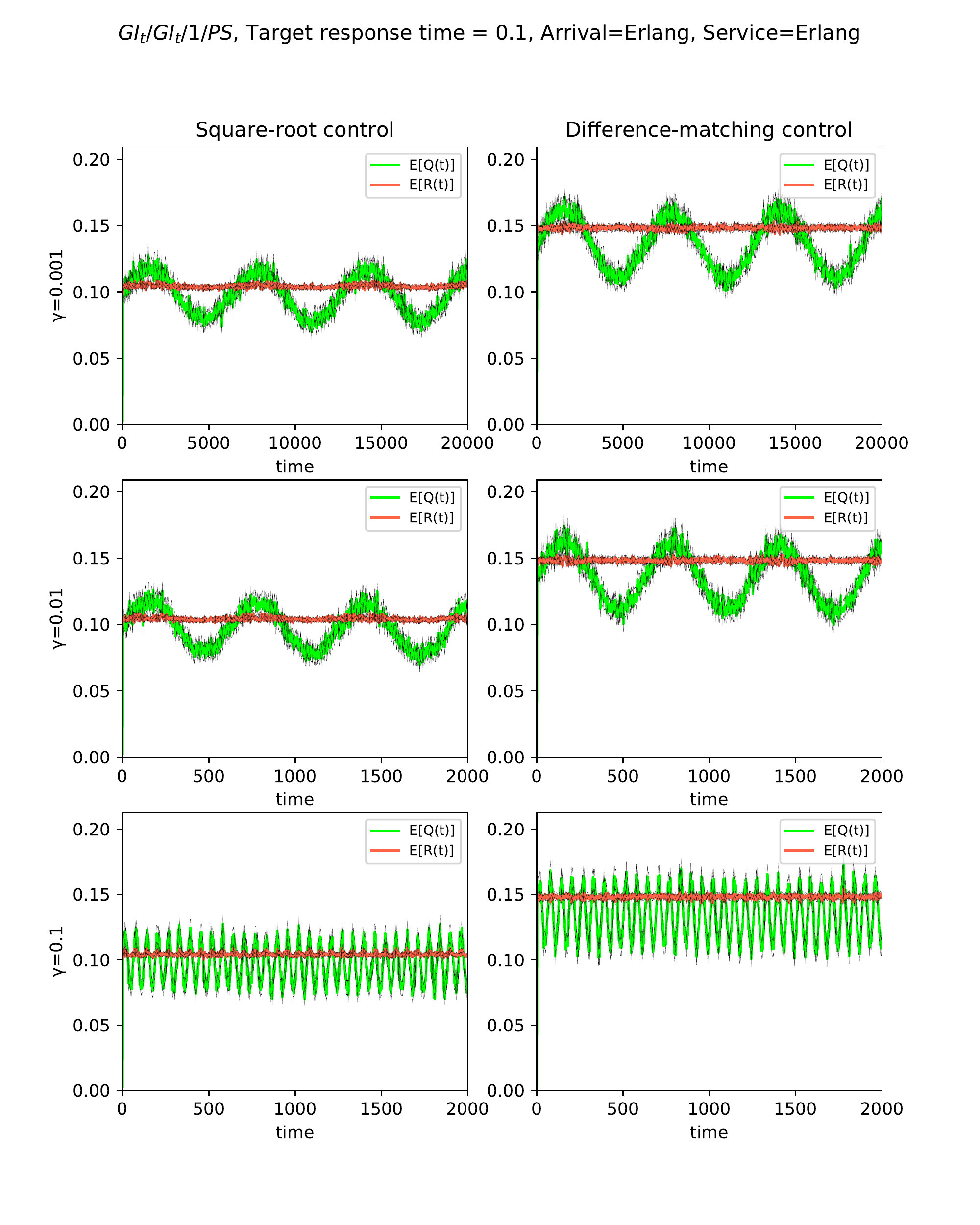}
	\caption{General performance measures of $ER_t/ER_t/1/PS$ queues under $\mu_{SR}$ and $\mu_{DM}$ with target response time 0.1 ($s=0.1$)}
\end{figure}

% s=10.0, ER(t)/ER(t)/1/PS
\begin{figure}[H]
	\label{fig:tvgg1ps_10.0_erer}
	\centering\includegraphics[width=\linewidth]{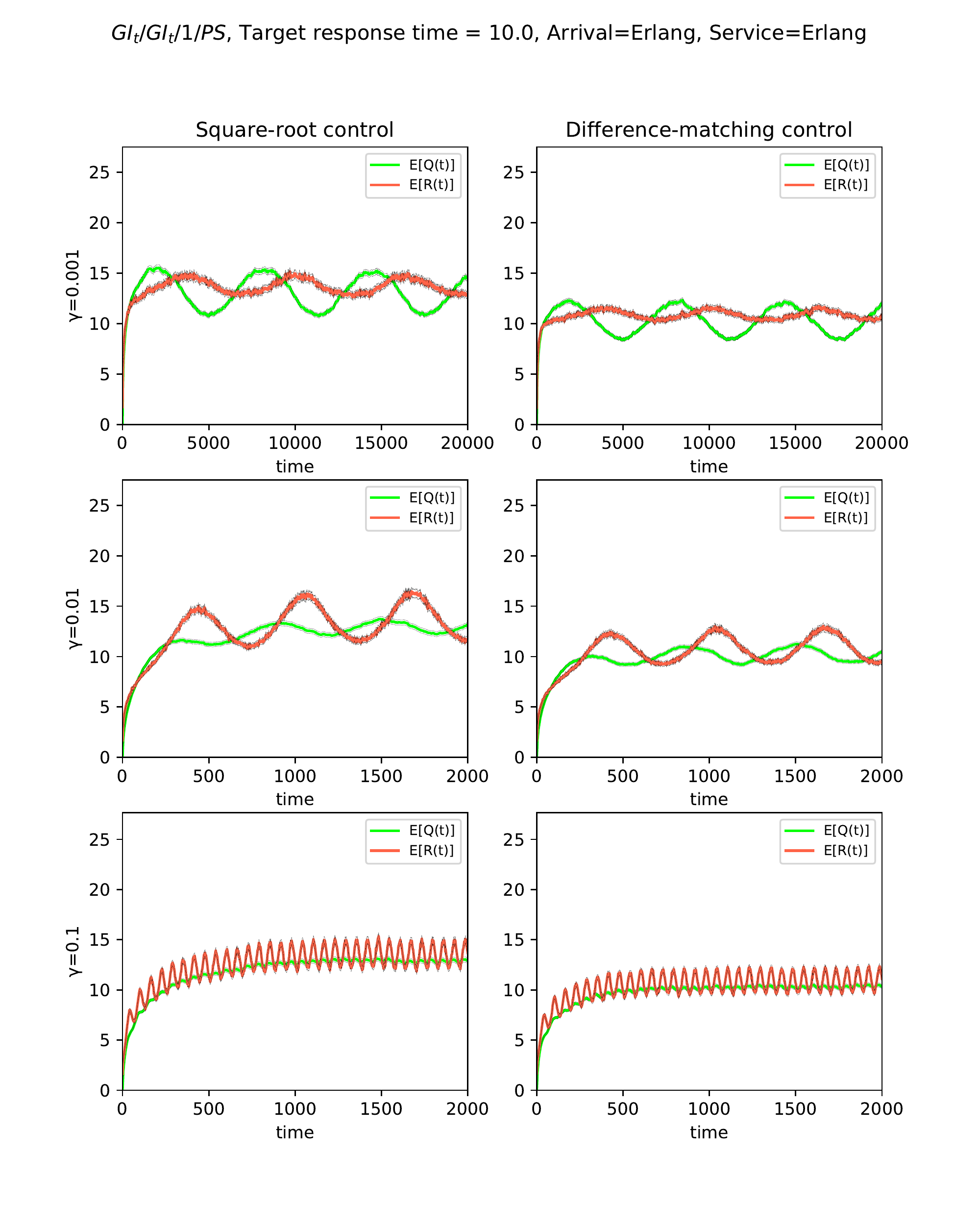}
	\caption{General performance measures of $ER_t/ER_t/1/PS$ queues under $\mu_{SR}$ and $\mu_{DM}$ with target response time 10.0 ($s=10.0$)}
\end{figure}
\pagebreak

% LN(t)/LN(t)/1/PS
%\subsection{Arrival base/Job size: Lognormal/Lognormal ($LN_t/LN_t/1/PS$)}
% s=0.1, LN(t)/LN(t)/1/PS
\begin{figure}[H]
	\label{fig:tvgg1ps_0.1_lnln}
	\centering\includegraphics[width=\linewidth]{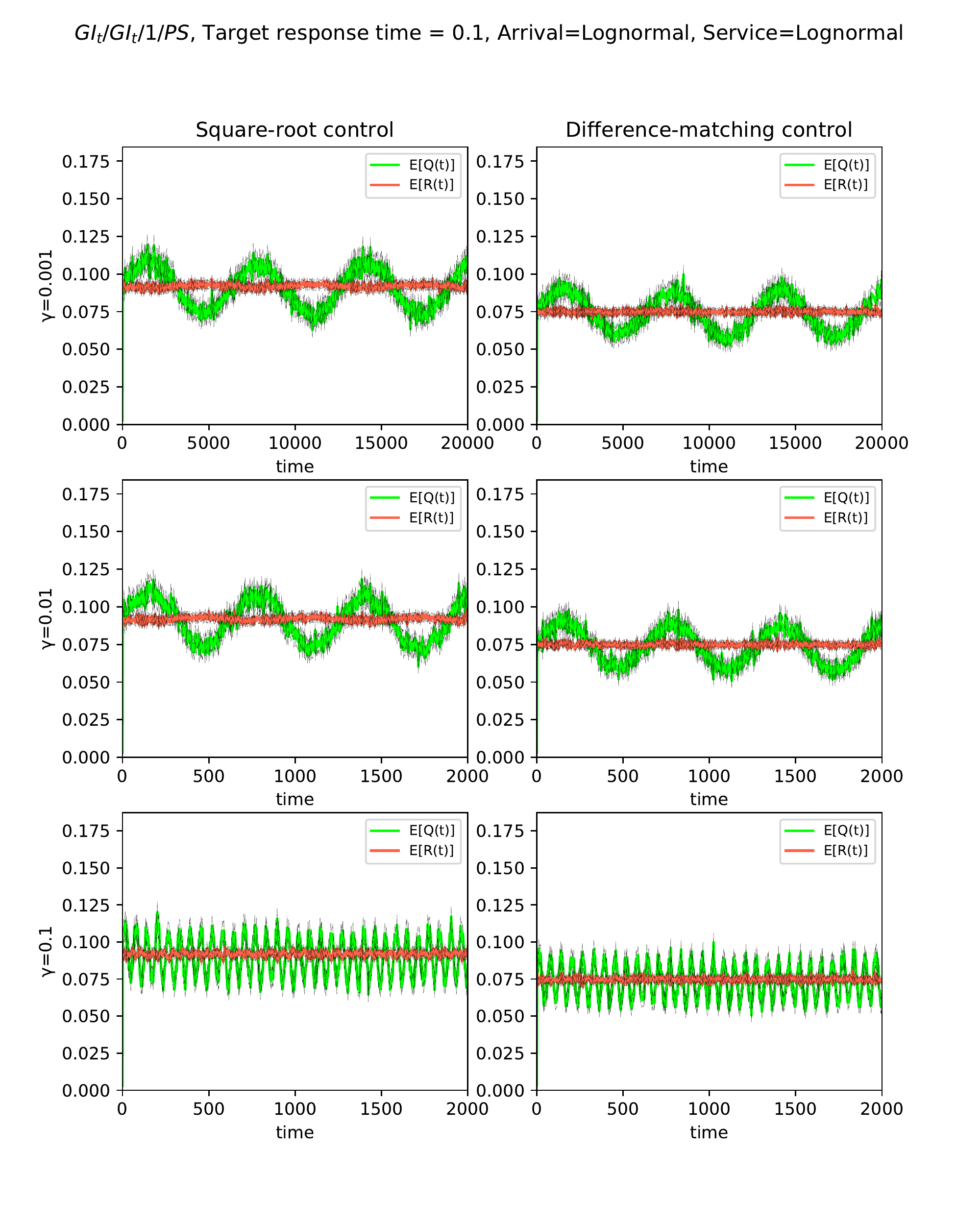}
	\caption{General performance measures of $LN_t/LN_t/1/PS$ queues under $\mu_{SR}$ and $\mu_{DM}$ with target response time 0.1 ($s=0.1$)}
\end{figure}

% s=10.0, LN(t)/LN(t)/1/PS
\begin{figure}[H]
	\label{fig:tvgg1ps_10.0_lnln}
	\centering\includegraphics[width=\linewidth]{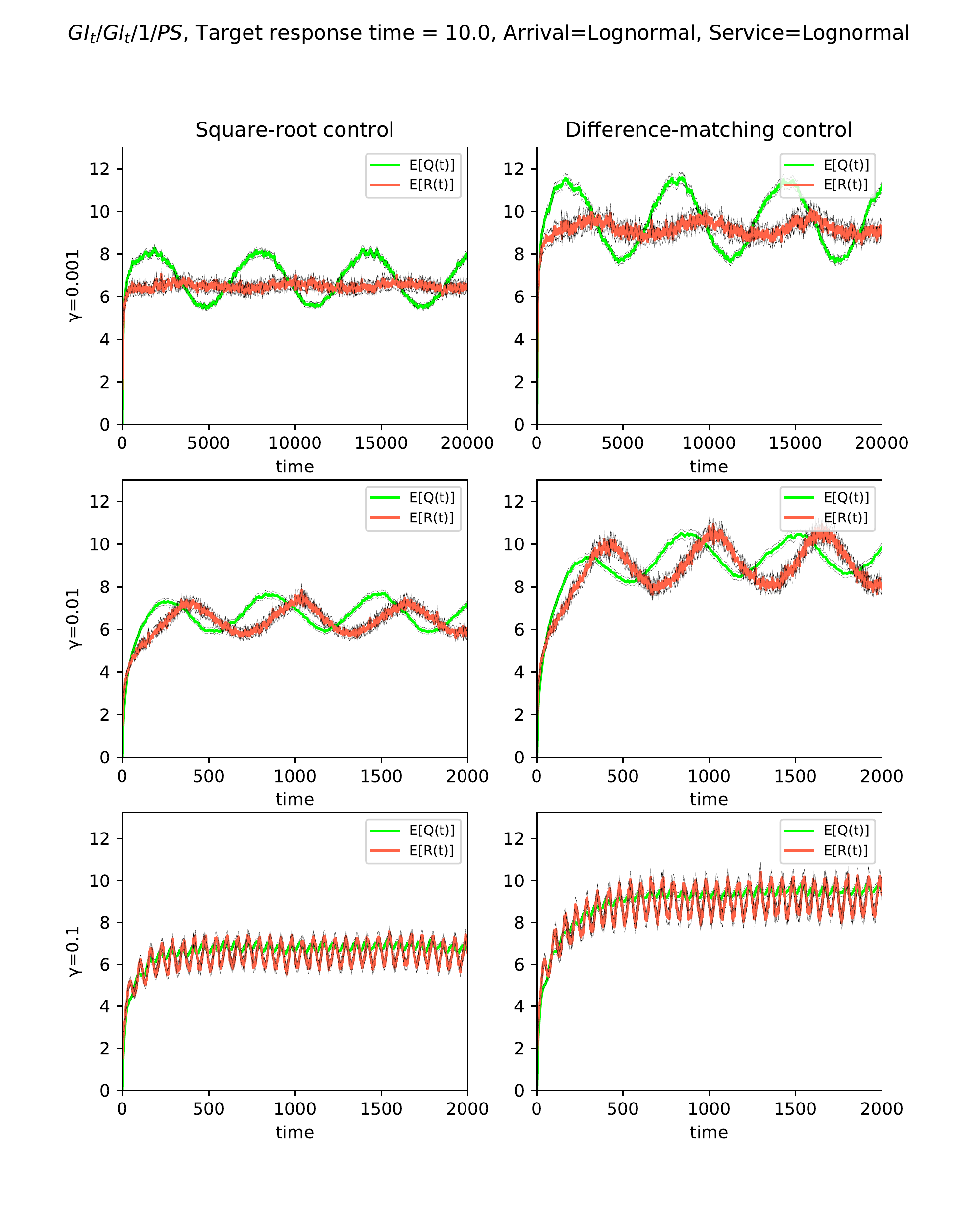}
	\caption{General performance measures of $LN_t/LN_t/1/PS$ queues under $\mu_{SR}$ and $\mu_{DM}$ with target response time 10.0 ($s=10.0$)}
\end{figure}

% ER(t)/LN(t)/1/PS
%\subsection{Arrival base/Job size: Erlang/Lognormal ($ER_t/LN_t/1/PS$)}
% s=0.1, ER(t)/LN(t)/1/PS
\begin{figure}[H]
	\label{fig:tvgg1ps_0.1_erln}
	\centering\includegraphics[width=\linewidth]{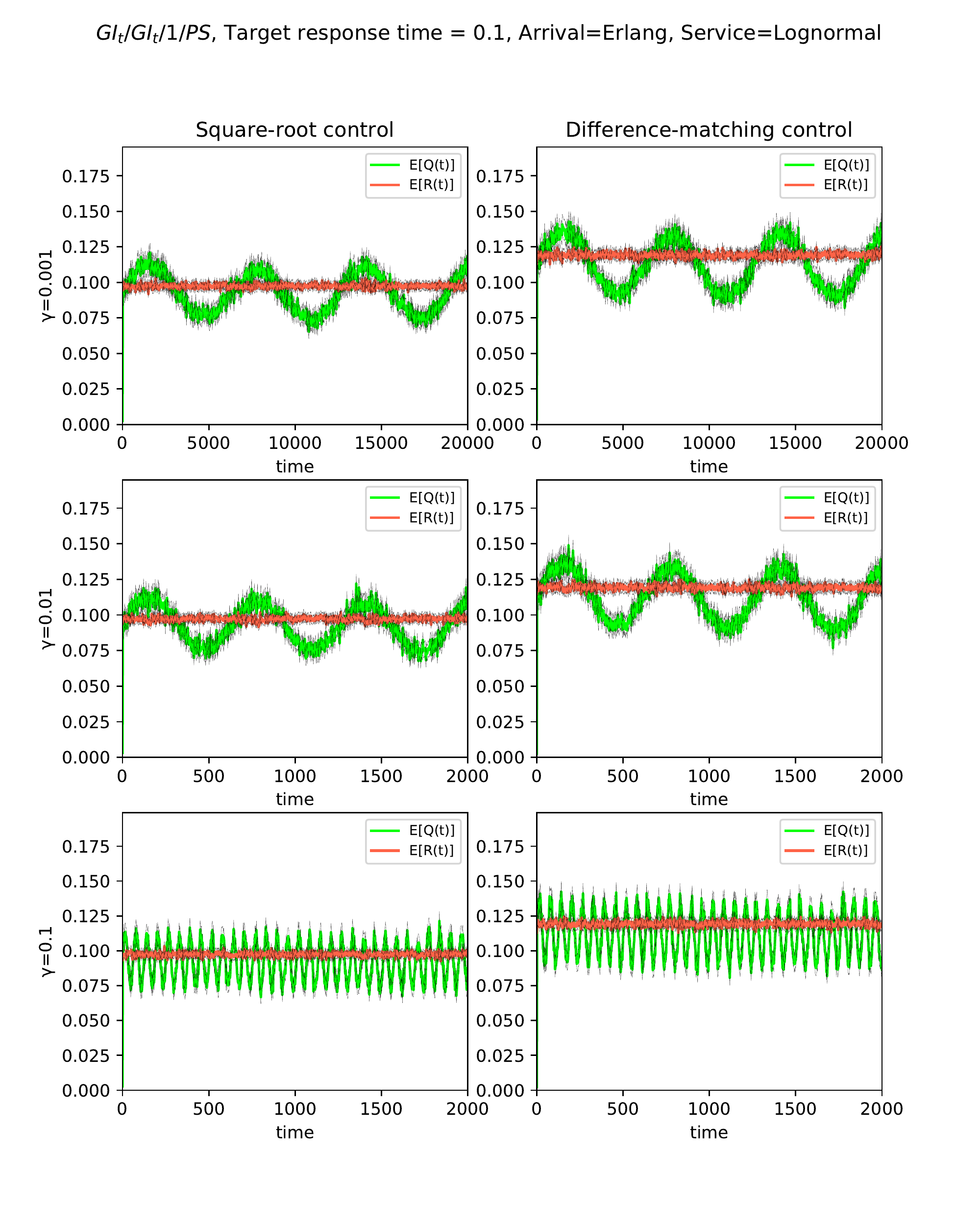}
	\caption{General performance measures of $ER_t/LN_t/1/PS$ queues under $\mu_{SR}$ and $\mu_{DM}$ with target response time 0.1 ($s=0.1$)}
\end{figure}

% s=10.0, ER(t)/LN(t)/1/PS
\begin{figure}[H]
	\label{fig:tvgg1ps_10.0_erln}
	\centering\includegraphics[width=\linewidth]{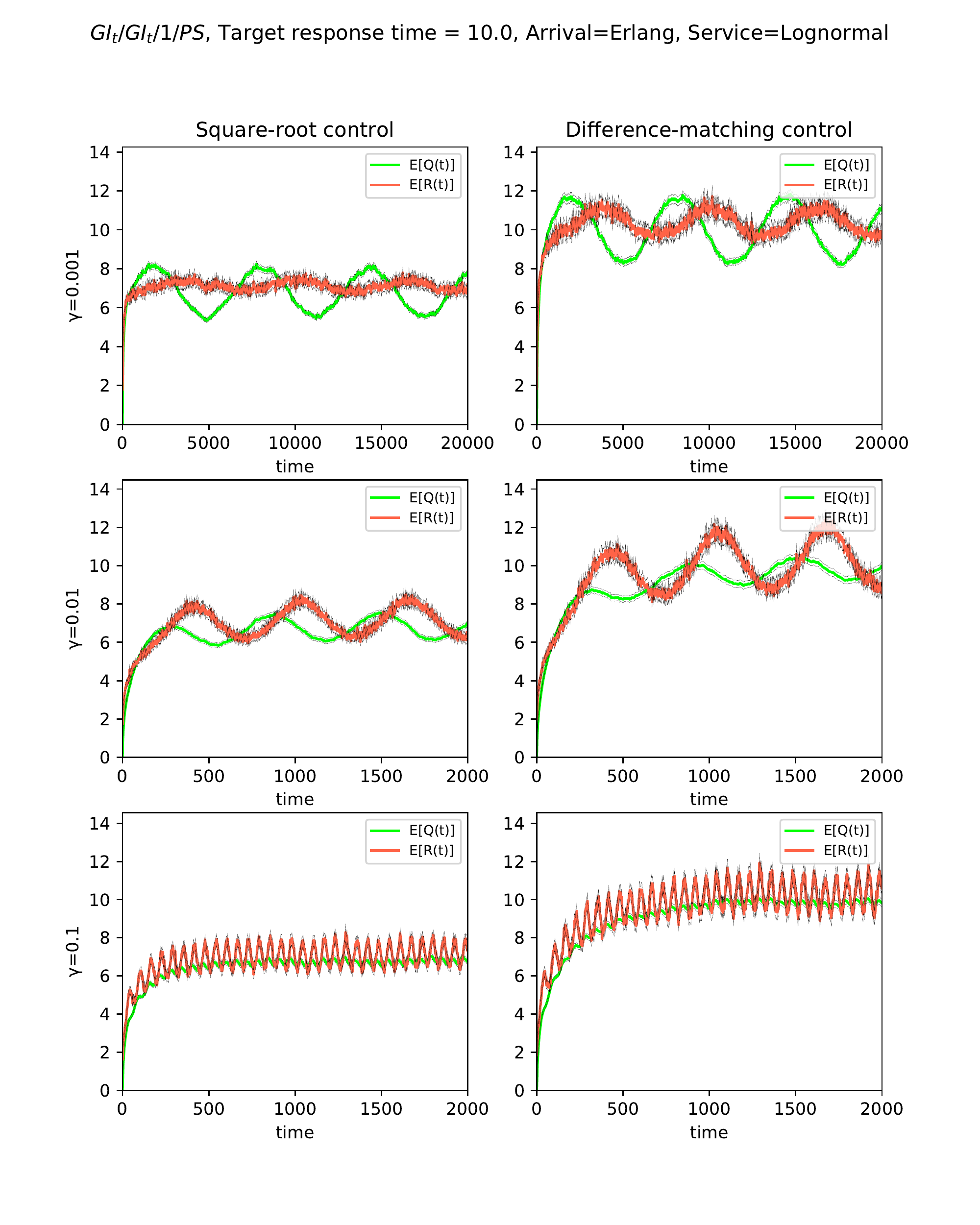}
	\caption{General performance measures of $ER_t/LN_t/1/PS$ queues under $\mu_{SR}$ and $\mu_{DM}$ with target response time 10.0 ($s=10.0$)}
\end{figure}

% LN(t)/ER(t)/1/PS
%\subsection{Arrival base/Job size: Lognormal/Erlang ($LN_t/ER_t/1/PS$)}
% s=0.1, LN(t)/ER(t)/1/PS
\begin{figure}[H]
	\label{fig:tvgg1ps_0.1_lner}
	\centering\includegraphics[width=\linewidth]{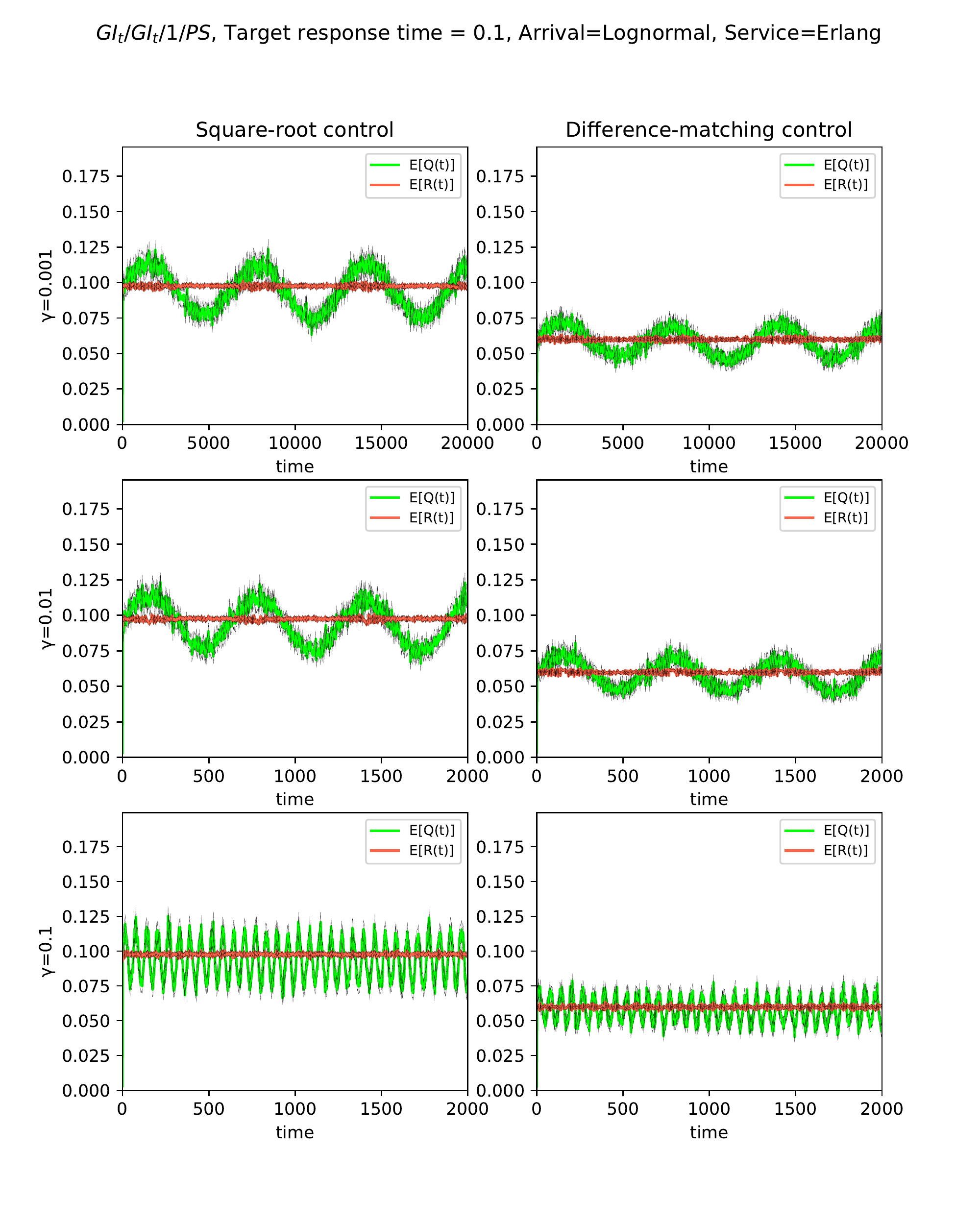}
	\caption{General performance measures of $LN_t/ER_t/1/PS$ queues under $\mu_{SR}$ and $\mu_{DM}$ with target response time 0.1 ($s=0.1$)}
\end{figure}

% s=10.0, N(t)/ER(t)/1/PS
\begin{figure}[H]
	\label{fig:tvgg1ps_10.0_lner}
	\centering\includegraphics[width=\linewidth]{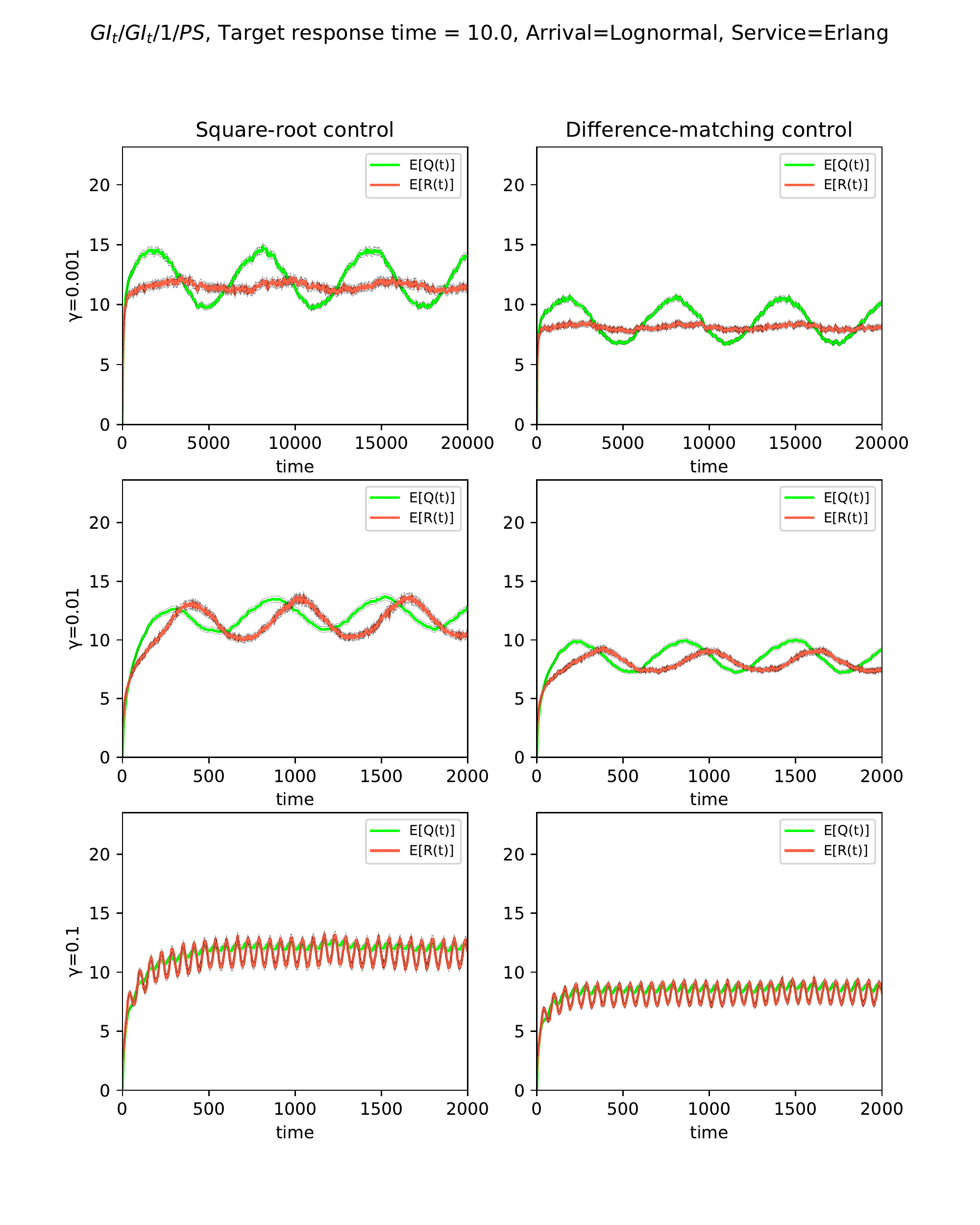}
	\caption{General performance measures of $LN_t/ER_t/1/PS$ queues under $\mu_{SR}$ and $\mu_{DM}$ with target response time 10.0 ($s=10.0$)}
\end{figure}
\end{document}